\documentclass[aps,pra,a4paper,showkeys,footinbib,longbibliography,notitlepage,twocolumn, 10pt]{revtex4-1}

\usepackage{amsmath,amssymb,amsthm,amsfonts,graphicx,float}
\usepackage[utf8]{inputenc}

\usepackage{bm}
\usepackage[dvipsnames]{xcolor}
\usepackage[colorlinks,citecolor=OliveGreen,urlcolor=blue,linkcolor=blue,bookmarks=true]{hyperref}      
\usepackage[nameinlink,capitalize]{cleveref}

\usepackage{dsfont}
\usepackage{mdwlist}
\usepackage[english]{babel}
\addto\extrasenglish{%
}
\usepackage[utf8]{inputenc}
\usepackage[T1]{fontenc}
\usepackage{tikz}
\usepackage{quantikz}

\usepackage{thm-restate}
\usepackage{mathtools}
\usepackage{enumitem}
\usepackage{algpseudocode}
\usepackage{algorithm}
\usepackage{bm}

\DeclarePairedDelimiter\abs{\lvert}{\rvert}

\newcommand{\bbC}{\mathbb{C}}

\newcommand{\Tr}{\operatorname{tr}}
\def\iden{\mathds{1}}
\def\dist{\mathrm{dist}}

\def\Pr{\mathrm{Pr}}
\newcommand{\ot}{\otimes}

\newcommand{\brho}{\boldsymbol{\rho}}

\newcommand{\bI}{\boldsymbol{I}}

\newcommand{\bV}{\boldsymbol{V}}
\newcommand{\bX}{\boldsymbol{X}}

\newcommand{\Bell}[2]{\ket{\mathsf{Bell}_{#1, #2}}}

\newcommand{\epr}[1]{\ket{\mathsf{EPR}_{#1}}}
\usepackage{nicefrac}

\newcommand{\fig}[1]{Figure~\ref{fig:#1}}
\newcommand{\secref}[1]{Section~\ref{sec:#1}}
\newcommand{\poly}{\operatorname{poly}}
 
\newcommand{\E}{\mathbb{E\/}}
\def\vcentcolon{\mathrel{\mathop\ordinarycolon}}
\begingroup \catcode`\:=\active
  \lowercase{\endgroup
  \let :\vcentcolon
  }

\newtheorem{theorem}{Theorem}[section]  
\newtheorem*{theorem*}{Theorem}         
\newtheorem{definition}[theorem]{Definition} 
\newtheorem{fact}{Fact}[section]
\newtheorem{remark}[theorem]{Remark}
\newtheorem{proposition}[theorem]{Proposition}
\newtheorem{lemma}[theorem]{Lemma}

\newtheorem{corollary}{Corollary}[section]
\newcommand{\appref}[1]{Appendix~\ref{sec:#1}}
\newcommand{\lemref}[1]{Lemma~\ref{lem:#1}}
\newcommand{\thmref}[1]{Theorem~\ref{thm:#1}}
\newcommand{\propref}[1]{Proposition~\ref{prop:#1}}
 \newcommand{\defref}[1]{Definition~\ref{def:#1}}
 \newcommand{\corref}[1]{Corollary~\ref{cor:#1}}
 
\newcommand{\factref}[1]{Fact~\ref{fact:#1}}
\def\iden{\mathds{1}}

\def\norm#1{ {|\hspace{-.022in}|#1|\hspace{-.022in}|} }
\def\Norm#1{ {\big|\hspace{-.022in}\big| #1 \big|\hspace{-.022in}\big|} }

\def\NORM#1{ {\left|\hspace{-.022in}\left| #1 \right|\hspace{-.022in}\right|} }

\def\iden{\mathds{1}}
\def\dist{\mathrm{dist}}

\def\Pr{\mathrm{Pr}}
\newcommand{\ketbra}[2]{|#1\rangle\langle #2|}
\usepackage{quantikz}

\newcommand{\lightconeee}{\raisebox{-0.5\height}{\includegraphics[height=12ex]{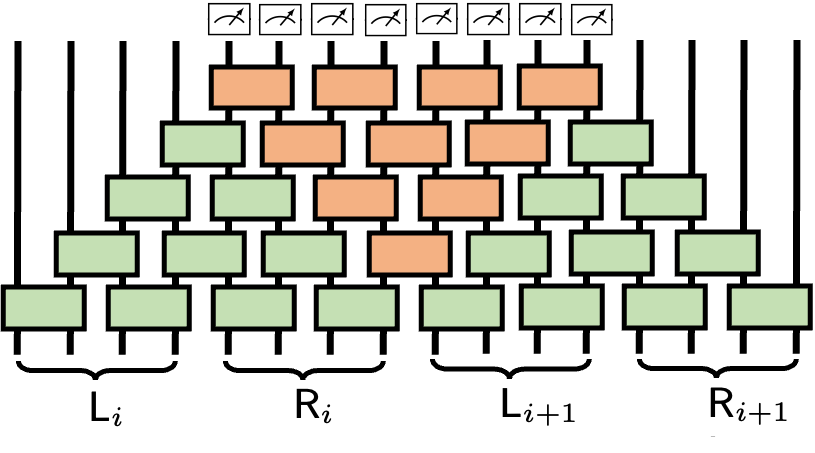}}}

\newcommand{\lightconex}{\raisebox{-0.5\height}{\includegraphics[height=12ex]{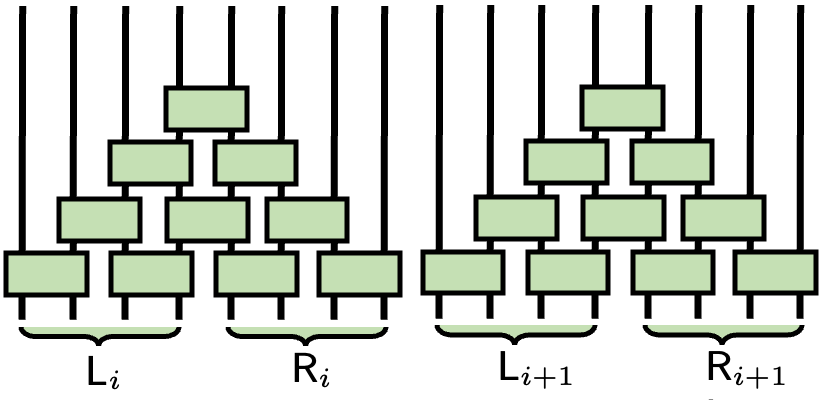}}}

\begin{document}

\title{When can classical neural networks represent quantum states?}
\author{Tai-Hsuan Yang$^{* 1,2,3}$}
\author{Mehdi Soleimanifar$^{* 3}$}
\author{Thiago Bergamaschi$^{4}$}
\author{John Preskill$^{3,5}$}

\affiliation{$^1$University of Illinois Urbana-Champaign, Urbana, IL, USA}
\affiliation{$^2$National Taiwan University, Taipei City, Taiwan}
\affiliation{$^3$California Institute of Technology, Pasadena, CA, USA}
\affiliation{$^4$University of California Berkeley, Berkeley, CA, USA}
\affiliation{$^5$AWS Center for Quantum Computing, Pasadena, CA, USA}

\thanks{Co-first author.}

\begin{abstract} 
    A naive classical representation of an $n$-qubit state requires specifying exponentially many amplitudes in the computational basis. 
    Past works have demonstrated that classical neural networks can succinctly express these amplitudes for many physically relevant states,  leading to computationally powerful representations known as neural quantum states.
    What underpins the efficacy of such representations?
    We show that \emph{conditional correlations} present in the measurement distribution of quantum states control the performance of their neural representations.
    Such conditional correlations are basis dependent, arise due to measurement-induced entanglement, and reveal features not accessible through conventional few-body correlations often examined in studies of phases of matter.
    By combining theoretical and numerical analysis, we demonstrate how the state’s entanglement and sign structure, along with the choice of measurement basis, give rise to distinct patterns of short- or long-range conditional correlations.
    Our findings provide a rigorous framework for exploring the expressive power of neural quantum states.
\end{abstract}
\maketitle

\section{Introduction}\label{sec:intro}
Modern machine learning (ML) algorithms based on artificial neural networks have achieved impressive results in image and language processing. 
Recent works have shown that similar architectures can also be employed to classically simulate the properties of many-body quantum systems \cite{carleo2017solving, melko2019restricted, Pfau2020, Carleo2019CNN, Hibat2020RNN, melko2024language}, or learn a classical representation of a quantum state from measurement data  \cite{carrasquilla2019reconstructing, melko2019restricted, torlai2018neural, iouchtchenko2023neural, zhao2023empirical, Torlai2020Precise}. 
To achieve this, a machine learning model such as a recurrent neural network \cite{Hibat2020RNN, Hibat2023RnnTopology}, convolutional neural network \cite{Carleo2019CNN}, or restricted Boltzmann machine \cite{carleo2017solving} is used to concisely represent the amplitudes $\psi(x) = \braket{x}{\psi}$ of an $n$-qubit state $\ket{\psi} = \sum_{x \in \{0,1\}^n} \psi(x) \ket{x}$ in some fixed basis. 
That is, upon querying with a bit string $x \in \{0,1\}^n$, the ML model returns a complex value $\psi(x)$ that we interpret as the amplitude of a possibly un-normalized state~$\ket{\psi}$. 
The ML model admits a series of weights collectively denoted by $w$. 
Each configuration of these weights~$w$, leads to the representation of a particular quantum state $\ket{\psi_w}$. 
For such a representation to be useful, the set of weights and the running time of the ML model should scale at most polynomially $\poly(n)$ with the number of qubits. 
Overall, this yields a parameterized family of $n$-qubit states $\{\ket{\psi_w}\}$ referred to as \emph{neural quantum~states}~\cite{carleo2017solving}.

The neural representation of quantum states provides a powerful computational and statistical model, allowing for sampling from the measurement distribution of the state \cite{Bravyi2022MCMC, Bravyi2023rapidlymixingmarkov, Torlai2020Precise}, Monte Carlo estimation of sparse observables \cite{carleo2017solving, melko2019restricted, Pfau2020, Carleo2019CNN, Hibat2020RNN, melko2024language}, state certification with few single-qubit measurements \cite{huang2024certifying}, and dequantization of quantum algorithms \cite{Tang2022Sampling, Tang2019Recom}.
Despite their wide range of applications and empirical effectiveness, the theoretical underpinning of why neural networks could, in principle, be used to represent the many-body states of interest is not fully understood.
Indeed at first sight, the practical success of neural quantum states seems in conflict with the fact that general $n$-qubit states require exponentially many parameters to be precisely specified.
However, in practice, the constraints imposed by the physics of the system may be leveraged to achieve an efficient classical representation of its state.
A notable example of such constraints is the area law for entanglement entropy in gapped ground states, which underpins our understanding of tensor network representations of these states \cite{hastings2007area_law, AradRigorousRG, Arad2013AreaLaw, anshu2019AreaLaw2D, AradFrustrationFreeAreaLaw}.

In this paper, we reveal a physical feature of quantum states that controls the efficiency of their neural representation. 
This property can be best introduced by considering a representation of the quantum state using~recurrent~neural~networks (RNN) \cite{Hibat2020RNN, Hibat2023RnnTopology, melko2024language}, although as we will see, our findings extend to other architectures. 
In the RNN representation of a quantum state, we express the amplitudes in polar form, $\psi(x) = |\psi(x)| \cdot e^{i g(x)}$, with the magnitude $|\psi(x)|$ and phase $e^{i g(x)}$ components treated separately.
We focus primarily on the power and limitations of RNN in representing the magnitudes $|\psi(x)|$, setting aside the phase component $e^{ig(x)}$. 
This simplification incurs no loss of generality for the sign-free states considered in \secref{nonNegativAmp}. 
More broadly, as discussed in \secref{long-rangeCMI}, the presence of complex phases often hinders finding neural representations for quantum states—a challenge encountered across various simulation and learning applications \cite{parkKastoryano2020expressive, Carrasquilla2019Gutzwiller, Castelnovo2020signproblem, Bukov2021nonStoquastic}.

The squared-magnitudes $p(x):=|\psi(x)|^2$ simply correspond to the distribution obtained by measuring the state $\ket{\psi}$ in the computational basis.
The RNN representation of this distribution is based on the chain rule of probabilities
\begin{align}
    p(x_1,\dots, x_n) = p(x_1) p(x_2|x_1)\cdots p(x_n|x_1,\dots, x_{n-1}).\label{eq:chianRule}
\end{align}
Starting with $p(x_1)$ and proceeding successively, each conditional probability $p(x_k|x_1,\dots, x_{k-1})$ is computed by a feedforward neural network acting on $x_1,\dots, x_{k-1}$ and a hidden state $h_{k-1}$ from the last step,~see~\fig{FNNRNN}.

The conditional probability $p(x_k|x_1,\dots, x_{k-1})$ for $k\leq n$ could in general be a complicated $k$-bit function with an exponential complexity $\mathcal{O}(2^{k})$.  
Hence, for the RNN to be able to efficiently compute these probabilities, we need to impose further conditions. 
A sufficient (though not necessary) criterion is \emph{conditional independence} of the variables.
That is, the distribution of each variable $x_k$ is only dependent on the previous $\ell$ variables $x_{k-\ell},\dots, x_{k-1}$ for some $\ell\in[k-1]$ such that 
\begin{align}
    p(x_k|x_1,\dots, x_{k-1}) = p(x_k|x_{k-\ell}, \dots, x_{k-1}).\label{eq:CMIProbabilityVersion}
\end{align}
This reduces the worst-case complexity of computing each conditional distribution to $\mathcal{O}(2^{\ell})$, which can be polynomially-bounded for a sufficiently small $\ell\leq \mathcal{O}(\log n)$, and leads to an efficient neural network representation as we show in \secref{ShortRange}.

In fact, a similar argument applies even when the conditional independence in Eq.~\eqref{eq:CMIProbabilityVersion} only approximately holds.
This is relevant because many quantum states of interest, such as those examined in \thmref{CMI_line_2d}, are known to satisfy the approximate version.
To quantify this, we use \emph{classical conditional mutual information} $I(\mathsf{A}:\mathsf{C}|\mathsf{B})$ as a measure of conditional correlations.
This quantity represents the average correlation between two sets of variables, $\mathsf{A}$ and $\mathsf{C}$, conditioned on the value of the variables in $\mathsf{B}$ (see \secref{ShortRange} for a formal definition).

Altering the measurement basis can significantly affect the amount of conditional correlations, reflecting the basis-dependent nature of neural quantum states and our analysis throughout this work. 
This contrasts with `localizable entanglement' \cite{Verstraete2004Localizable, Popp2005Localizable}, which corresponds to the \emph{maximum} entanglement that can be created on average between subsystems $\mathsf{A}$ and $\mathsf{C}$ through local measurements on the qubits in $\mathsf{B}$. 
Conditional correlations are also distinct from conventional two-point (or bipartite) correlations, which are often used to characterize phases of matter or explain the success of classical simulation techniques. Quantum states with short-range two-point correlations can show either short- or long-range conditional correlations upon measurement (e.g., cluster state in the $Z$ versus $X$ basis). 
Likewise, those with long-range two-point correlations can also display either type of conditional correlations (e.g., output states of deep quantum circuits versus quantum phase states \cite{ji2018pseudorandom, brakerski2019pseudo}).

Previous works have identified distinct computational and algebraic barriers that limit the (worst-case) expressiveness of neural quantum states \cite{gao2017efficient, huang2020predicting, NNrepresentationTNSharir, jiang2023local}.
We will see that our information-theoretic approach, based on characterizing conditional correlations, provides a unified and physically tangible framework for studying these and more general obstacles to the neural representation of quantum states.

Motivated by these observations, we explore the presence or absence of conditional independence in the measurement distribution $p(x_1,\dots, x_n)$ of an $n$-qubit quantum state, and how this influences both the existence and the computational tractability of finding an efficient neural representation.
We consider quantum states that, roughly speaking, exhibit one of the following features:
\begin{enumerate}[leftmargin=1.5em]
\item[(1)]The distribution $p(x_1, \dots, x_n)$ has short-range conditional correlations, corresponding to a small $\ell$ in Eq.~\eqref{eq:CMIProbabilityVersion}, and can be efficiently represented by shallow neural networks. 
\item[(2)] The distribution $p(x_1, \dots, x_n)$ exhibits long-range conditional correlations, but still allows for an efficient neural network representation. However, variationally finding such a representation (e.g., via gradient descent) may be challenging. 
\item[(3)] The distribution $p(x_1, \dots, x_n)$ has long-range conditional correlations that are difficult to represent with efficient neural networks. 
\end{enumerate}

We emphasize that this is a non-exhaustive list of the types of conditional correlations that can occur in quantum systems, and our work does not rule out many other possibilities.
For instance, the conditional distributions in \eqref{eq:CMIProbabilityVersion} may depend on a small number of variables $\ell$, which are not necessarily geometrically local. 
Alternatively, the conditional mutual information $I(\mathsf{A}:\mathsf{C}|\mathsf{B})$ may decay, but not necessarily at an exponential rate.
There might also exist efficiently trainable neural quantum states with long-range conditional correlations.

Through rigorous analysis and numerical experiments, we explore how each of the above cases (1), (2), and (3) arises in physically motivated states.
We identify the amount of (pre-measurement) entanglement in the state, the sign structure of the amplitudes, and the local measurement basis as crucial factors determining the type of measurement-induced correlations in the state.
Our results are summarized as follows:

\vspace{0.5em}

\noindent \textbf{Result 1.} Quantum states with short-range conditional correlations admit shallow neural network representations.

\vspace{0.5em}

We state this result more formally in \thmref{CMI_line_2d} and prove it in \appref{CNNfromCMI}.

\vspace{0.5em}

\noindent \textbf{Result 2.} The ground states of prototypical gapped spin systems, states generated by random shallow 1D quantum circuits, and random tensor network states with predominantly positive tensor entries display short-range conditional correlations.

\vspace{0.5em}

The setting of random shallow 1D quantum circuits is stated in \thmref{CMIdecayRandom1D} and proved in \appref{CMIShallowCkt}. Numerical evidence for spin systems and random tensor network states is presented in \secref{numerics} and \secref{nonNegativAmp} respectively.

\vspace{0.5em}

\noindent\textbf{Result 3.} The 1D cluster state, subjected to single-qubit rotations, exhibits conditional correlations that transition from short- to long-range as the rotation angle varies. This transition is accompanied by a decline in the performance of variational algorithms in finding the neural representation of the state.

The conditional correlations in cluster states are bounded in \thmref{decayCMIClusterState}, and the numerical evidence for the computational hardness of finding the rotated cluster state is discussed in \secref{clusterStateCMI}.

Lastly, in \secref{HigherDimensions}, we discuss how recent studies on random 2D tensor networks and states prepared by random 2D shallow quantum circuits suggest that the proliferation of long-range correlations induced by measurements poses a significant obstacle to efficiently representing these states with neural networks.

While our primary focus is on neural representations of quantum states and the Monte Carlo algorithms based on them, our work lies within the broader context of measurement-induced entanglement and monitored quantum circuits, which have been the subject of much interest in relation to phases of many-body entanglement \cite{Li2018Monitored, Skinner2019Monitored, Chan2019Monitored,bao2021finite, google2023measurementInduced, Popp2005Localizable}, classical simulation of random quantum circuits \cite{napp2019efficient2d, BeneWatts2024MIE}, and probing the sign problem \cite{hastings2016signfree, lin2022probingsign}.

\begin{figure}[t!]
    \centering
\includegraphics[width=0.85 \columnwidth]{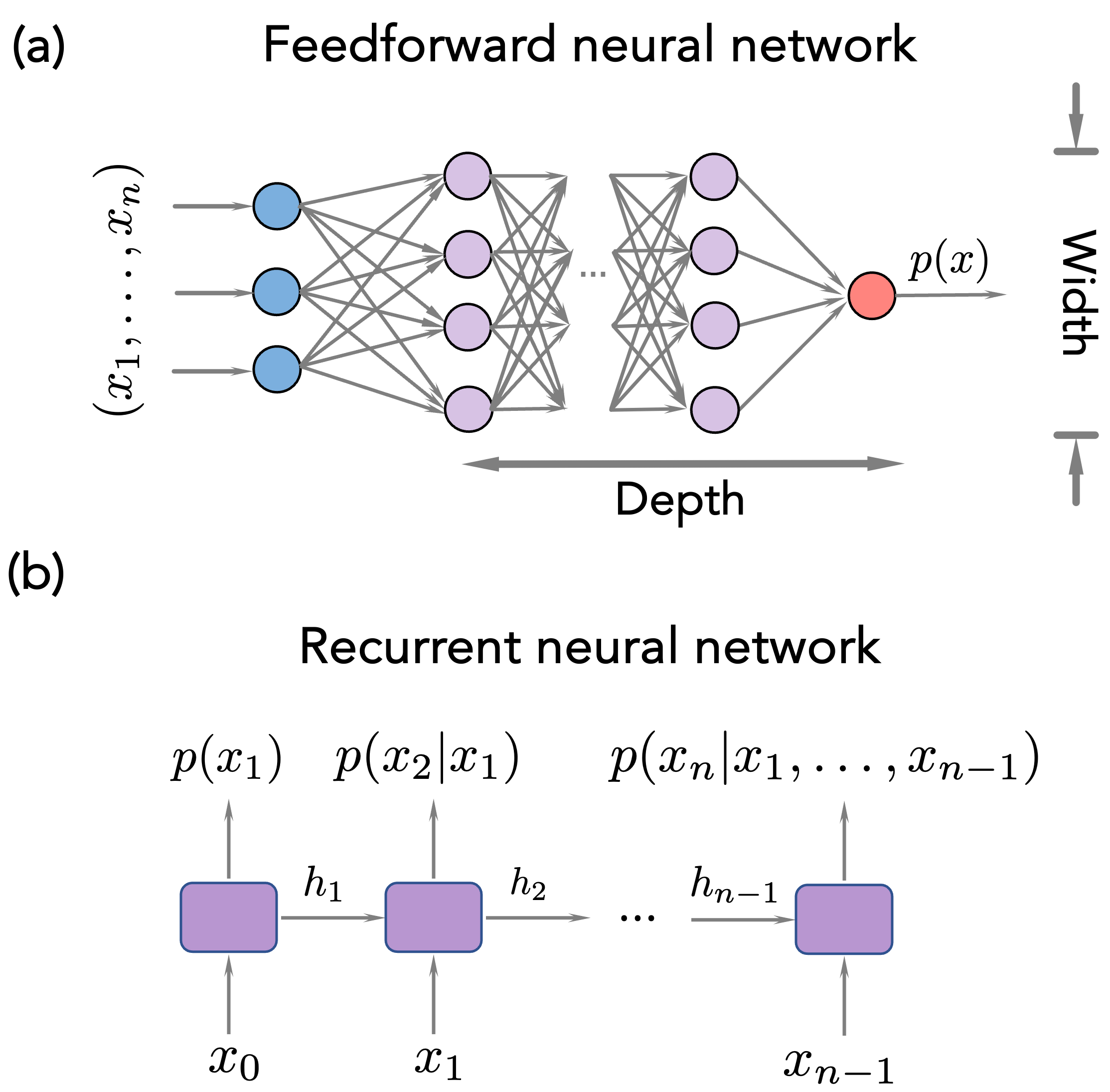}
    \caption{\textbf{Neural representation of} $\bm{p(x) = |\psi(x)|^2}$.
\textbf{(a)} A feedforward neural network with $n$ input nodes and specified depth and width.
    \textbf{(b)} A recurrent neural network involves multiple recurrent cells, each efficiently computing a conditional distribution using a feedforward neural network and a hidden state, which is passed to the next cell.
    This allows for an autoregressive sampling from the joint distribution~$p(x_1,\dots, x_n)$ using the chain rule \eqref{eq:chianRule}. Here $h_1,\dots, h_{n-1}$ are the hidden states, $x_0$ is the start token initialized as zero, and $x_1,\dots, x_n$ are inputs bits. 
    }
    \label{fig:FNNRNN}
\end{figure}

\section{Short-range conditional correlations}\label{sec:ShortRange}

We consider an entangled system of qubits arranged on the vertices of a finite-dimensional lattice such as a one-dimensional chain or two-dimensional grid. 
Consider three subsets of qubits $\mathsf{A}$, $\mathsf{B}$, and $\mathsf{C}$ where $\mathsf{A} \cup \mathsf{B} \cup \mathsf{C} \subseteq [n]$, with $\dist(\mathsf{A}, \mathsf{C})$ denoting the length of the shortest path connecting $\mathsf{A}$ and $\mathsf{B}$ on the lattice.
Suppose when the qubits in the subset $\mathsf{B}$ are measured in the computational basis, we observe the outcome ${x}_{\mathsf{B}} \in \{0,1\}^{|\mathsf{B}|}$. 
We let $p_{\mathsf{A},\mathsf{C}|\mathsf{B}}(x_\mathsf{A}, x_\mathsf{C}|{x}_\mathsf{B})$ denote the marginal distribution of measurement outcomes for qubits in $\mathsf{A}$ and $\mathsf{C}$ conditioned on observing ${x}_B$.
The marginal distributions $p_{\mathsf{A}|\mathsf{B}}(x_{\mathsf{A}}|{x}_{\mathsf{B}})$ and $p_{\mathsf{C}|\mathsf{B}}(x_{\mathsf{C}}|{x}_{\mathsf{B}})$ are defined in a similar way.
We say the measurement outcomes in regions $\mathsf{A}$ and $\mathsf{C}$ are \emph{conditionally independent} with respect to the outcomes in region $\mathsf{B}$ when for any $x_{\mathsf{B}}$ such that $p_{\mathsf{B}}(x_{\mathsf{B}}) > 0$ and any choices of $x_{\mathsf{A}}$ and $x_{\mathsf{C}}$, we have
\begin{align}
p_{\mathsf{A},\mathsf{C}|\mathsf{B}}(x_\mathsf{A}, x_\mathsf{C}|{x}_\mathsf{B}) = p_{\mathsf{A}|\mathsf{B}}(x_{\mathsf{A}}|{x}_{\mathsf{B}}) p_{\mathsf{C}|\mathsf{B}}(x_{\mathsf{C}}|{x}_{\mathsf{B}}).\nonumber
\end{align}

A stronger notion of conditional independence is captured using the classical conditional mutual information~(CMI) defined by 
\begin{align}
    I(\mathsf{A}:\mathsf{C}|\mathsf{B}) :=  \E_{\mathsf{B}}D_{\text{KL}} (p_{\mathsf{A},\mathsf{C}|\mathsf{B}} \parallel p_{\mathsf{A}|\mathsf{B}}\otimes p_{\mathsf{C}|\mathsf{B}}),\label{eq:defCMI}
\end{align}
where the expectation is over $x_{\mathsf{B}}\sim p_{\mathsf{B}}(x_{\mathsf{B}})$ and the Kullback–Leibler~(KL) divergence between two distributions $q$ and $q'$ is given by $D_{\text{KL}}(q\parallel q') = \E_{{x}\sim q}\log \left({q({x})/q'({x})}\right)$ in their support.
These two notions are connected by $\E_{\mathsf{B}}\Norm{p_{\mathsf{A},\mathsf{C}|\mathsf{B}}-p_{\mathsf{A}|\mathsf{B}}\otimes p_{\mathsf{C}|\mathsf{B}}}_1 \leq  \sqrt{2 \cdot I(\mathsf{A}:\mathsf{C}|\mathsf{B})}$, which follows from the well-known Pinsker's inequality.

Quantum states with short-range conditional correlations may exhibit an \emph{approximate} version of conditional independence defined as follows:
\begin{definition}[Approximate conditional independence in measurement distribution]\label{def:ApproximateConditionalIndependence}
    The probability distribution obtained by measuring an $n$-qubit quantum state $\ket{\psi}$ in the computational basis satisfies the approximate conditional independence with respect to subsets $\mathsf{A} \cup \mathsf{B} \cup \mathsf{C} \subseteq [n]$, if there is a constant $\xi$ and a function $\alpha(\mathsf{A}, \mathsf{C})$ such that
    \begin{align}
    I(\mathsf{A}:\mathsf{C}|\mathsf{B}) \leq \alpha({\mathsf{A}, \mathsf{C}}) \cdot e^{-\dist(\mathsf{A},\mathsf{C})/\xi}.\label{eq:ConditionalIndependence}
\end{align}
    and $\alpha({\mathsf{A}, \mathsf{C}})\leq \poly(|\mathsf{A}|, |\mathsf{C}|)$. 
    We refer to parameter $\xi$ as the \emph{CMI length}.
\end{definition}
 \begin{figure}[t!]
    \centering
    \includegraphics[width=\columnwidth]{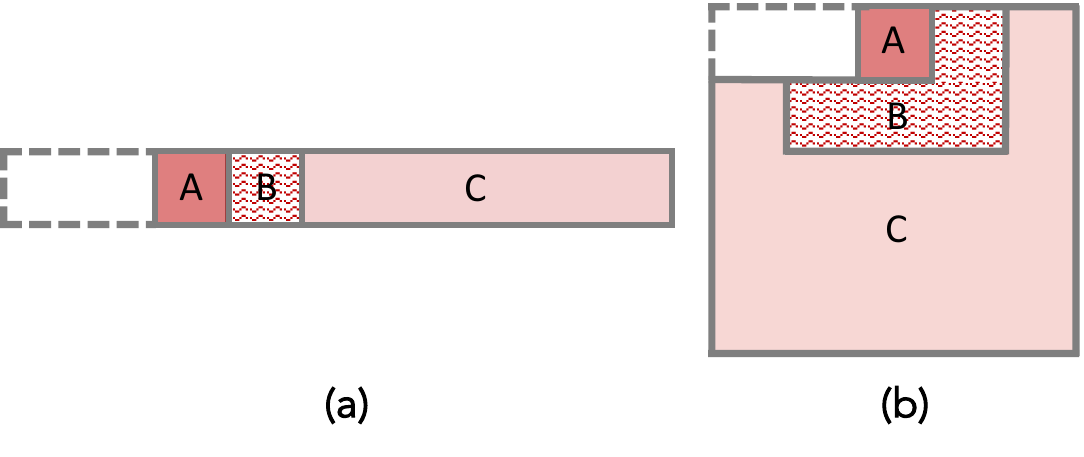}
    \caption{\textbf{Partitioning the lattice.} Geometrically contiguous regions $(\mathsf{A},\mathsf{B},\mathsf{C})$ needed for the neural network construction in \thmref{CMI_line_2d} in \textbf{(a)} one-dimensional, and \textbf{(b)} two-dimensional lattices.
    In each example, a contiguous part of the lattice has been traced out and a tripartition of the remainder is considered such that $\mathsf{A} \cap \mathsf{C} = \emptyset$.
    The approximate conditional independence in \defref{ApproximateConditionalIndependence} is assumed to hold for such tripartitions.}
     \label{fig:CMI_line_2d}
\end{figure} 
Based on \defref{ApproximateConditionalIndependence}, we classify conditional correlations as \emph{long-range} under two circumstances: (1) the conditional mutual information (CMI) does not decay exponentially as in \eqref{eq:ConditionalIndependence}; or (2) for a fixed number of qubits, as a Hamiltonian parameter (e.g., the rotation angle in \secref{clusterStateCMI}) is varied, the CMI length $\xi$ increases and becomes comparable to the size of the lattice. Otherwise, we describe the conditional correlations as being \emph{short-range}.

To connect this notion of approximate conditional independence to efficient neural network representation of the state $\ket{\psi}$, we consider two architectures for neural quantum states, depicted in \fig{FNNRNN} and introduced formally in \appref{NeuralNetworks}. 
The first is the feedforward neural network. This architecture consists of a network of neurons, each of which applies an affine transformation to its input followed by a non-linear activation function.
The number of layers in the network is referred to as the \emph{depth}, while the maximum number of neurons in any layer is known as the \emph{width} of the neural network.
As discussed in \secref{intro}, we also consider recurrent neural networks (RNN), which compute the probabilities $p(x_1,\dots, x_n)$ via the chain rule stated in Eq.~\eqref{eq:chianRule}.
An RNN consists of a cascade of cells, functioning as feedforward neural networks. 
The $k$th cell processes an input bit $x_k$ and a hidden memory state $h_k$. 
The number of neurons in each cell and the size of the memory determine the complexity of evaluating an RNN.

We now state our first result which applies the conditional independence property \defref{ApproximateConditionalIndependence} in a quantum state to construct its neural representation. The proof is given in \appref{CNNfromCMI}.

\begin{theorem}[Neural quantum state from conditional independence]\label{thm:CMI_line_2d}
    Consider $n$ qubits arranged on a $D_{\Lambda}$-dimensional lattice $\Lambda$ with a quantum state $\ket{\psi}=\sum_{x \in \{0,1\}^n} \psi(x) \ket{x}$.
    Suppose the measurement distribution $p(x) = |\psi(x)|^2$ in this basis satisfies the approximate conditional independence property given in \defref{ApproximateConditionalIndependence} with respect to geometrically contiguous regions shown in \fig{CMI_line_2d}. Then, there exists a feedforward neural network with 
    \begin{align}
        \text{depth}  =\mathcal{O}\left(\log\log(n)\right)  \text{ \ and width} = n^{\mathcal{O}(\log^{D_{\Lambda}-1}(n))}\label{eq:depthWidth2}
    \end{align}
        that computes a function $q(x)$ such that $\sum_{x}|q(x)-p(x)|\leq 1/\poly(n)$.
This can also be achieved using a recurrent neural network with $D_{\Lambda}$-dimensional geometry,  a memory size of $\mathcal{O}(\log^{D_{\Lambda}}(n))$, and  feedforward cells with width and depth given by \eqref{eq:depthWidth2}.
\end{theorem}

While \thmref{CMI_line_2d} guarantees the existence of a shallow neural network representation for magnitudes $|\psi(x)|$ of states with short-range conditional correlations, it does not imply that statistically learning such a representation is easy. 
This can be illustrated with quantum phase states of the form $\ket{\psi} = \sum_{x \in \{0,1\}^n} \frac{1}{\sqrt{2^n}} e^{i g(x)} \ket{x}$, which exhibit no conditional correlations. 
However, for a random choice of phases $g(x)$, these states can be highly entangled and indistinguishable from Haar random states \cite{ji2018pseudorandom, brakerski2019pseudo}.

The proof of \thmref{CMI_line_2d} follows strategies similar to those employed in several previous works \cite{Brandao_gibbs_preparing, napp2019efficient2d, BeneWatts2024MIE, hastings2016signfree, Poulin_sdp}.
The decay of \emph{quantum} conditional mutual information is used in \cite{Brandao_gibbs_preparing} to prepare the Gibbs state of a quantum state.
Conditions similar to \defref{ApproximateConditionalIndependence} form the basis for the classical algorithms in  \cite{napp2019efficient2d, BeneWatts2024MIE} for simulating the measurement distribution of shallow random quantum circuits, as well as the ``coherent Gibbs state'' representation of non-negative wavefunctions in \cite{hastings2016signfree}. 
We emphasize that our use of conditional independence is limited to geometrically contiguous regions $\mathsf{A}, \mathsf{B}, \mathsf{C} \subset [n]$, while \cite{napp2019efficient2d, Brandao_gibbs_preparing} consider regions distributed non-locally on the lattice. 
There are examples of quantum states such as cluster states \cite{Haah2016SpuriousEntanglement} where conditional independence holds for geometrically local regions but not for arbitrary subsets of qubits.

In light of \thmref{CMI_line_2d}, we seek to identify processes in quantum systems that lead to short-range conditional correlations, as described in \defref{ApproximateConditionalIndependence}.
We discuss two such mechanisms in what follows: (1) sufficiently low amount of entanglement in the state prior to subsystem measurements, and (2) non-negative amplitudes $\psi(x)\geq 0$ corresponding to a sign-free (also known as stoquastic) state $\ket{\psi}= \sum_{x\in\{0,1\}^n} \psi(x) \ket{x}$.
Additionally, in \secref{long-rangeCMI}, we demonstrate how the choice of measurement basis can alter the range of conditional correlations.

\subsection{Sufficiently low entanglement}\label{sec:lowEntanglement}
\begin{figure}
    \centering
    \includegraphics[width=\columnwidth]{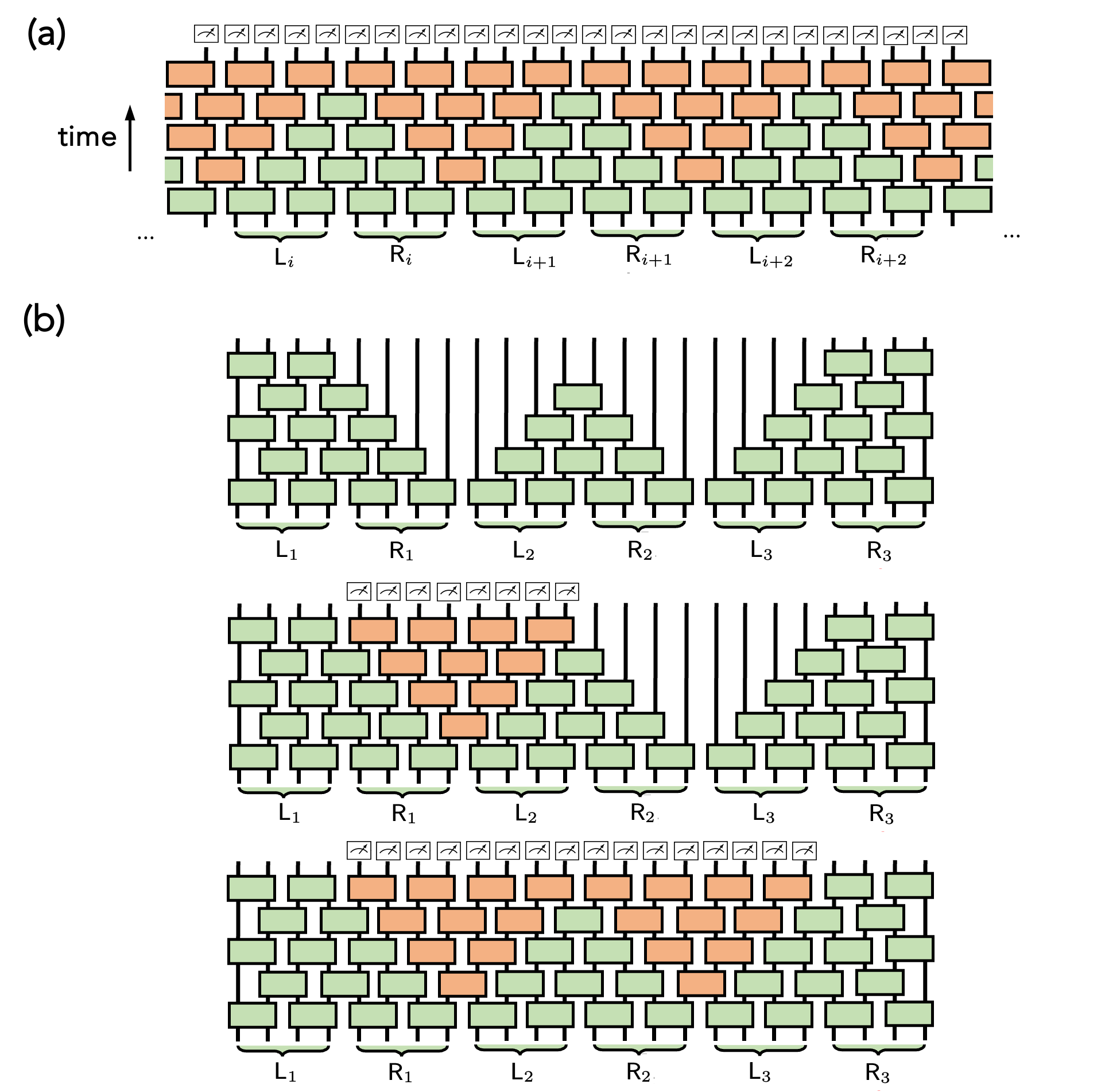}
    \caption{\textbf{Conditional correlations in shallow quantum circuits via entanglement swapping.} \textbf{(a)} The decomposition of a one-dimensional shallow circuit with the brickwork architecture into backward lightcones in green and forward lightcones in orange.
    \textbf{(b)} The conditional correlations generated between subsystems $\mathsf{L_1}$ and $\mathsf{R_3}$ due to computational basis measurements on subsystems $\mathsf{R_1}$, $\mathsf{L_2}$, $\mathsf{R_2}$, and $\mathsf{L_3}$ can be understood as two rounds of entanglement swapping measurements. 
    }

\label{fig:lightconesShallow}
\end{figure} 

The amount and structure of entanglement in a quantum state affect the conditional correlations in its measurement distribution.
This is most evident when the state $\ket{\psi}$ is generated by a low-depth quantum circuit. 
When this shallow circuit has a one-dimensional geometry, the study of conditional correlations effectively reduces to analyzing an \emph{entanglement swapping} setup.
In its simplest form, as defined in \defref{entanglementSwapping}, entanglement swapping refers to the formation of a maximally entangled state between the subsystems $\mathsf{L}_i$ and $\mathsf{R}_{i+1}$ of a four-qubit state $\ket{\mathrm{EPR}}_{\mathsf{L}_i\mathsf{R}_i} \otimes \ket{\mathrm{EPR}}_{\mathsf{L}_{i+1}\mathsf{R}_{i+1}}$ when subsystems $\mathsf{R}_i$ and $\mathsf{L}_{i+1}$ are jointly measured in the Bell basis. 
Such a Bell basis measurement is equivalent to applying an entangling gate followed by single-qubit measurements.
Now consider a one-dimensional circuit of constant depth $D$ with the brickwork architecture consisting of alternating layers of two-qubit gates as in \fig{lightconesShallow}(a). 
The effect of computational-basis measurements on the output state of this depth-$D$ circuit can be understood as consecutive applications of entanglement swapping measurements.
To illustrate this, we first partition the gates into disjoint groups: 'backward lightcones' in green and 'forward lightcones' in orange, as shown in \fig{lightconesShallow}.
We then replace $\ket{\mathrm{EPR}} \otimes \ket{\mathrm{EPR}}$ in the previous setup with the state
\begin{align}
    \ket{a_i}_{\mathsf{L}_i\mathsf{R}_i}\otimes\ket{a_{i+1}}_{\mathsf{L}_{i+1}\mathsf{R}_{i+1}} = \lightconex
\end{align} 
prepared by applying the gates in the backward lightcones on $\ket{0}_{\mathsf{L}_i \mathsf{R}_i} \otimes \ket{0}_{\mathsf{L}_{i+1}\mathsf{R}_{i+1}}$ with $\mathsf{L_i}$, $\mathsf{R}_i$, $\mathsf{L}_{i+1}$, and $\mathsf{R}_{i+1}$being subsystems of $D-1$ qubits. 
Finally, instead of Bell basis measurements, we perform entangling measurements by applying the gates in the forward lightcone on subsystems $\mathsf{R}_i$ and $\mathsf{L}_{i+1}$, followed by computational basis measurements on these subsystems.
The post-measurement state is given by
\begin{align}
    \lightconeee
\end{align} 
and can exhibit entanglement between disjoint subsystems $\mathsf{L}_i$ and $\mathsf{R}_{i+1}$ analogous to entanglement swapping with EPR pairs.
Measurements throughout the circuit induce multiple rounds of entanglement swapping, ultimately allowing distant parts of the circuit to become entangled.
This process is illustrated in \fig{lightconesShallow}(b) for a small circuit consisting of six subsystems.

As we will see in \appref{swappingprelim}, by carefully choosing the gates in the circuit, one can always implement a perfect entanglement swapping in which the initial states and joint measurements are maximally entangled. 
This leads to the formation of long-range measurement-induced entanglement between distant subsystems. 
However, for \emph{typical} circuits where the gates are not fine-tuned and instead \emph{randomly} chosen, this is no longer the case. 
For such random sufficiently-shallow circuits, we expect the states and joint measurements to only be \emph{partially} entangled. In \appref{CMIShallowCkt}, we show that in this case, each round of measurements reduces the post-measurement entanglement by a constant factor, leading to an overall exponential decay in measurement-induced entanglement and conditional correlations between distant subsystems. 
As the following theorem states, this decay ensures that conditional correlations between regions of length $\Omega(\log(n))$ are negligibly small.

\begin{theorem}[Conditional independence in random shallow 1D circuits]\label{thm:CMIdecayRandom1D}
       Consider the family of random depth-$D$ brickwork quantum circuits acting on a one-dimensional chain of qudits with local dimension $d$. 
       Fix $L \geq c_0 \cdot (4d)^{3D} \cdot \log(n \cdot D\log(d))$ for some sufficiently large constant $c_0$ and define contiguous regions $\mathsf{A}_1\cup \dots\cup\mathsf{A}_{n'}$ from left to right with $n' = n/L$ such that the size of each region is $|\mathsf{A}_i| = L \geq \Omega(\log(n))$ and $\dist(\mathsf{A}_i-\mathsf{A}_j)\geq \Omega(\log(n))$ for $|i-j|>1$.
    With probability $1 - 1/\poly(n)$ over the random choice of circuits, the measurement distribution of these random circuits satisfies the conditional~independence~property
     \begin{align}
    I(\mathsf{A}_i:\mathsf{A}_j|\mathsf{A}_{i+1},\dots, \mathsf{A}_{j-1}) \leq e^{-\Omega{\left(\dist(\mathsf{A}_i, \mathsf{A}_j)\right)}}\leq 1/\poly(n).\nonumber
     \end{align}
for any $1\leq i<j\leq n'$ with $|i-j|>1$.
\end{theorem}

Going beyond the 1D architecture, random shallow quantum circuits with a 2D geometry have also been conjectured and numerically shown to exhibit short-range conditional independence \cite{napp2019efficient2d, bao2021finite, BeneWatts2024MIE, Liu2022GraphState} for depths smaller than a critical constant value.
Similar behavior has been observed in random 2D tensor networks, known as projected entangled pair states (PEPS) \cite{Gonzalez2024CMIinPEPS}. 
When the bond dimension of these states is sufficiently small, the boundary entanglement of contracted PEPS follows an area law, indicating that the amplitudes of the state can be efficiently determined via tensor contraction methods.
In all these cases, however, a phase transition occurs beyond a small critical depth or bond dimension, where the conditional correlations become long-range.
We defer a more detailed discussion of this transition to \secref{long-rangeCMI}.

The properties of structured, physically relevant states may differ from those generated by random ensembles of quantum circuits and tensor networks \cite{Gonzalez2024CMIinPEPS}. To further support our findings, we next investigate the correlations in a class of physically motivated, sign-free local Hamiltonians.
Following this, in \secref{numerics}, we perform a numerical investigation of conditional correlations in several well-studied spin systems.

\begin{figure*}[t!]
    \centering
\includegraphics[width=0.8\textwidth]{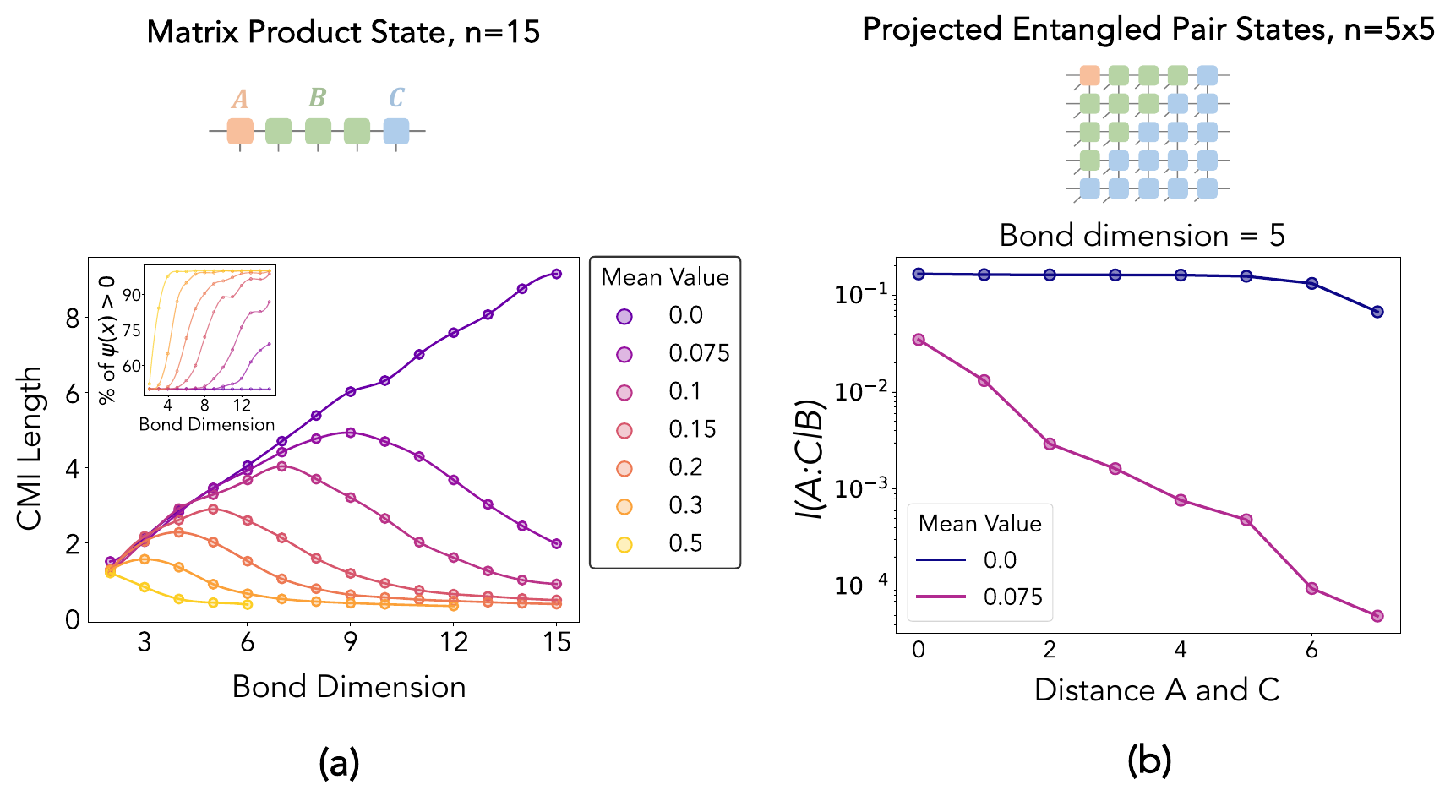}
\caption{\textbf{Conditional independence vs. sign structure in tensor networks:} \textbf{(a)} The decay of conditional mutual information (CMI)  \eqref{eq:ConditionalIndependence} in random MPS over $n=15$ qubits is demonstrated. 
The MPS tensor entries are real numbers drawn independently from the normal distribution $\mathcal{N}(\mu, 1)$ with mean $\mu$. 
For $\mu \simeq 0$, where the resulting amplitudes likely have varying signs $\psi(x)/|\psi(x)|$, increasing the bond dimension leads to longer-range conditional correlations. 
For larger $\mu\gg 0$, where more amplitudes are positive, correlations remain short-range, eventually with shorter correlation lengths at larger bond dimensions due to concentration effects.
\textbf{(b)} A similar behavior is observed in random PEPS on a $5 \times 5$ qubit grid. For a fixed bond dimension of $5$, the CMI decays exponentially when $\mu \gg 0$. However, for small $\mu \simeq 0$, the CMI does not decay, resulting in long-range conditional correlations. The MPS and PEPS plots show the average results over $200$ and $10$ samples per bond dimension, respectively.}
    \label{fig:SignFreeTN}
\end{figure*}

\begin{figure*}[t!]
    \centering
\includegraphics[width=\textwidth]{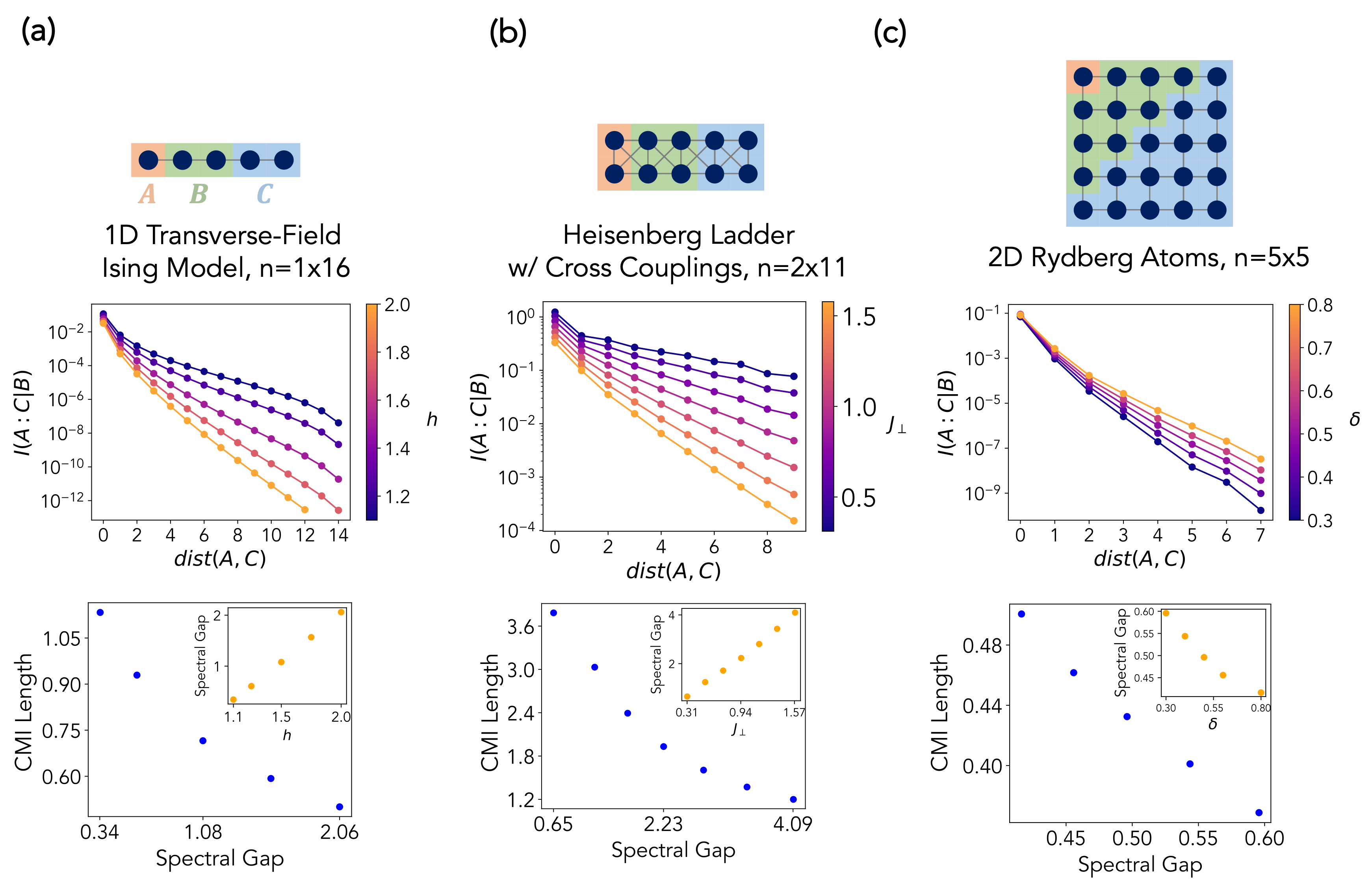}
    \caption{\textbf{Conditional independence in prototypical spin systems:} We observe the decay of conditional mutual information \eqref{eq:ConditionalIndependence} as a function of distance $\dist(\mathsf{A}, \mathsf{C})$ in the disordered phase of three gapped Hamiltonians with a unique ground state.
    \textbf{(a)} 1D transverse-field Ising model $H_{\text{Ising}} = J \sum_{i\in[n]} X_i X_{i+1} + h \sum_i Z_i$ with $J=1$ and $\mathsf{A}$ a single qubit. 
    \textbf{(b)} Heisenberg ladder with cross interactions:  $H_{\text{ladder}}= J_{\parallel} \sum_{i\in[n-1],j\in[2]} S_{i,j}\cdot S_{i+1,j} + J_{\perp}\sum_{i\in[n]} S_{i,1}\cdot S_{i,2} + J_{\times} \sum_{i\in[n-1]} (S_{i,1}\cdot S_{i+1,2} +  S_{i,2} \cdot S_{i+1,1})$ with $J_{\parallel} = 1$, $J_{\times} =  0.1$, and $\mathsf{A}$ consists of two qubits.
    \textbf{(c)} Rydberg atoms on a 2D lattice: $H_{\text{Rydberg}} = \sum_{i<j} \frac{C}{4\norm{r_i-r_j}^6}(\iden+Z_i)(\iden+Z_j)-\frac{\delta}{2} \sum_{i=1}^n (\iden+Z_i)-\frac{\Omega}{2}\sum_{i=1}^n X_i$ with $\Omega = 1.0$, $R_b = 1.2$, $C = \Omega \cdot R_b^6$, detuning parameter $\delta$, and $\mathsf{A}$ being a single qubit.}
\label{fig:CMIDecay}
\end{figure*}

\subsection{Non-negative amplitudes}\label{sec:nonNegativAmp}
Another property of a many-body quantum state that can affect conditional independence in its measurement distribution is the sign structure of its amplitudes.
A wide class of commonly studied local Hamiltonians enjoy the feature that their off-diagonal terms are non-positive. That is, in a local basis $\{\ket{x}: x\in \{0,1\}^n\}$, we have
\begin{align}
    \bra{x}H\ket{y}\leq 0\quad  \text{for}\quad  x\neq y.
\end{align}
Such Hamiltonians are called sign-free or ``stoquastic''.
The (unique) ground state of a sign-free Hamiltonian is a quantum state $\ket{\psi}=\sum_{x} \psi(x) \ket{x}$ with non-negative amplitudes $\psi(x)\geq 0$.

In the absence of topological order, the non-negativity of the amplitudes can constrain the amount of post-measurement entanglement that can be formed in the state.
This, for instance, may be seen in the entanglement swapping setup introduced earlier. 
Although the initial state $\ket{\mathrm{EPR}}\otimes \ket{\mathrm{EPR}}$ is non-negative, the Bell basis measurements performed subsequently are not, as they involve  basis states $(\ket{00}-\ket{11})/\sqrt{2}$ and $(\ket{01} - \ket{10})/\sqrt{2}$. 
In fact, the only sign-free measurement basis in this setting is the product basis $\{\ket{00},\ket{01},\ket{10}, \ket{11}\}$ which does not produce any entanglement. 

More generally, it has been shown \cite{hastings2016signfree} that when an $n$-qubit state $\ket{\psi}$ prepared by a constant depth 1D quantum circuit 
is non-negative, then the conditional correlations between two subsystems $\mathsf{A}$ and $\mathsf{C}$ upon measuring the rest of the system $\mathsf{B} = [n] \setminus{\mathsf{A}} \cup \mathsf{C}$ decay super-polynomially with the distance. 

Non-negative stabilizer states (also known as CSS codes) are another family of states whose conditional correlations have been studied.
Depending on their geometry, these states can host long-range topological order (e.g., the toric code). 
The limitations of neural networks in efficiently learning these states has been explored before in \cite{huang2020predicting}.
However, when such a state is the \emph{unique} ground state of a \emph{local} (stabilizer) Hamiltonian, it exhibits short-range conditional correlations \cite{lin2022probingsign}. 
It is also known that in the ground state of a stoquastic \emph{frustration-free} Hamiltonian, the ratio between amplitudes $\frac{\braket{y}{\psi}}{\braket{x}{\psi}}$ can be efficiently computed when $\braket{x}{\psi} > 0$ and $\bra{y}H\ket{x} < 0$ \cite{Bravyi2009FFStoquastic}.
This property has been further extended to the class of `magic ratio' Hamiltonians in~\cite{Bravyi2022MCMC}.

In addition to these cases, we also numerically investigate the link between the sign structure of the amplitudes and conditional correlations. 
In \fig{SignFreeTN}, we consider two families of tensor network states: matrix product states (MPS) and projected entangled pair states (PEPS). 
The amplitudes of these states are obtained by contracting the network of tensors. 
The sign structure in such tensor networks and its connection to efficient contraction algorithms has been recently studied in \cite{chen2024sign, jiang2024positive}.
Here, we investigate this connection from the viewpoint of conditional correlations.
The (real) entries in each tensor are shifted Gaussian variables $\mathcal{N(\mu, \sigma)}$ with mean $\mu \geq 0$ and variance $\sigma^2 = 1$.
By tuning the mean parameter $\mu$, we can change the sign structure of the corresponding amplitudes $\psi(x)$.
When $\mu$ becomes comparable to $\sigma$, the distribution of tensor entries shifts sufficiently into the positive region, resulting in amplitudes that are overwhelmingly positive,  
while choosing $\mu \simeq 0$ generates states with varying signs.
For a fixed mean value $\mu > 0$, the percentage of strictly positive amplitudes $\psi(x) > 0$ increases with the bond dimension $r$ due to the concentration effects. 
When $\mu \simeq 0$, increasing the bond dimension---thereby increasing entanglement---causes the conditional correlations to decay more slowly. 
This effect is particularly pronounced in random PEPS, where, beyond a bond dimension of $2$, the CMI does not decay with high probability when $\mu \simeq 0$.
Conversely, in the $\mu \gg 0$ regime, increasing the bond dimension eventually leads to shorter-range conditional correlations, as the majority of amplitudes become positive.
These results are illustrated in \fig{SignFreeTN}.
As a result of \thmref{CMI_line_2d} such states with short-range conditional correlations admit an efficient neural network representation.

Next, we show an intriguing relation between the conditional independence and the entanglement entropy area law for sign-free quantum states.
We note that the conditional independence studied so far is an information-theoretic property of the probability distribution obtained from measuring the state. 
The area law, however, is a bound satisfied by the von Neumann entropy of the reduced states. 
The sign-free assumption allows us to relate these classical and quantum information theoretic~notions. 
\begin{theorem}[Approximate conditional independence implies area law]\label{thm:signfreeArealaw}
    Consider an $n$-qudit sign-free state $\ket{\psi}\in \bbC^d\otimes \cdots \otimes \bbC^d$ on a lattice. Suppose the measurement distribution of this state satisfies the approximate conditional independence property \eqref{eq:ConditionalIndependence} for a given tripartition  $\mathsf{A} \cup \mathsf{B} \cup \mathsf{C} = [n]$ of the lattice such that $\mathsf{A} \cap \mathsf{C} = \emptyset$.
    The von Neumann entropy of the reduced state $\rho_{\mathsf{A}}$ in region $\mathsf{A}$ is upper bounded by $$S(\rho_{\mathsf{A}}) = \mathcal{O}(|\mathsf{B}|),$$
    which (up to $\log(|\mathsf{A}|)$ factors) implies an area law $S(\rho_{\mathsf{A}})= \mathcal{O}(|\partial \mathsf{A}|\cdot \log(|\partial\mathsf{A}|))$. 
\end{theorem}
The proof of this theorem is stated in \appref{areaLaw}.
The additional $\log(|\partial \mathsf{A}|)$ factor in the area law arises because controlling the error in applying the conditional independence property requires the distance between regions $\mathsf{A}$ and $\mathsf{C}$ (and thus the width of region $\mathsf{B}$) to satisfy $\operatorname{dist}(\mathsf{A}, \mathsf{C}) = \mathcal{O}(\log(|\mathsf{A}|))$. As a result, the size of region $\mathsf{B}$ is bounded by $|\mathsf{B}| = \mathcal{O}(|\partial \mathsf{A}| \cdot \log(|\mathsf{A}|))$.
The area law scaling of entanglement entropy in gapped ground states has been associated with the existence of efficient tensor network representations for these states \cite{AradRigorousRG, Arad2013AreaLaw, anshu2019AreaLaw2D, AradFrustrationFreeAreaLaw}. Consequently, \thmref{signfreeArealaw} establishes a link between accurate representations of quantum states using neural networks and tensor networks.

\subsection{Numerical simulations for prototypical spin systems}\label{sec:numerics}
In \fig{CMIDecay}, we examine the decay rate of conditional mutual information \eqref{eq:ConditionalIndependence} for some typical local Hamiltonians of interest: 1D transverse-field Ising model (TFIM), Heisenberg ladder Hamiltonian, and Rydberg atoms on a 2D lattice. 
For each case, a choice of regions $\mathsf{A}$, $\mathsf{B}$, and $\mathsf{C}$ is depicted in \fig{CMIDecay}, with region $B$ made larger in a natural way to change $\dist(\mathsf{A}, \mathsf{C})$.
Among these, the 1D TFIM and the Rydberg Hamiltonian are sign-free. 
All three Hamiltonians are gapped, have a unique ground state, and are studied in their disordered phase.
By tuning a parameter of each Hamiltonian, we vary the spectral gap and observe that the conditional mutual information decays exponentially with distance. However, the decay rate slows as the spectral gap decreases and the system approaches the phase transition point.

\section{Long-range conditional correlations}\label{sec:long-rangeCMI}

\subsection{Rotated cluster state}\label{sec:clusterStateCMI}
In this section, we consider quantum states whose measurement distribution exhibits long-range conditional correlations. 
The presence of these correlations is not necessarily an obstacle to the existence of efficient neural representations for the state. 
However, we show that such long-range correlations can affect the computational performance of conventional ML algorithms that aim to variationally find a neural representation of the state. 
To see this, we consider an $n$-qubit state on a one-dimensional chain, which is the ground state of
    \begin{align}
    H_{\mathrm{ES}} = &- \sum _{k=2}^{n-2} X_{k-1}Z_k X_{k+1}\nonumber\\
    &- Z_1 X_2 - X_{n-1} X_n -  X_{n-2} Z_{n-1} Z_n.\label{eq:rotated_clusterH}
    \end{align}
Up to Hadamard gates applied on all but the $n$-th qubit, this Hamiltonian is equivalent to the Hamiltonian for the 1D cluster state. The ground state of $H_{\mathrm{ES}}$ can be prepared using a depth-2 quantum circuit shown in \fig{ClusterState} that corresponds to the entanglement swapping (ES) setup in \fig{lightconesShallow}. Specifically, a series of EPR pairs are first created and then transformed into the Bell basis. 
Consequently, the measurement distribution of this state in the computational basis exhibits long-range conditional correlations. 
However, these correlations change when a local basis transformation, given by a $y$-rotation on each qubit, is applied. 
As stated in the following theorem, the resulting state---referred to as the \textit{rotated cluster state}---exhibits conditional correlations that decay at a rate determined by the rotation angle.

\begin{theorem}[Decay of CMI in rotated cluster state]\label{thm:decayCMIClusterState}

Consider the ground state of the Hamiltonian in \eqref{eq:rotated_clusterH} for an even $n$.
Apply a $y$-rotation $R_y(\theta)=\cos(\theta)\cdot \iden+ i \sin(\theta) \cdot Y$ for $\theta \in [0,\pi/2]$
to each qubit and proceed to measure all qubits in the standard $Z$ basis. 
Let $\mathsf{A}$, $\mathsf{C}$, and $\mathsf{B}$ represent the first qubit, the last qubit, and the remaining qubits, respectively.
The conditional mutual information $I(\mathsf{A}:\mathsf{C}|\mathsf{B})$ of the resulting measurement distribution satisfies
\begin{align}
I(\mathsf{A}:\mathsf{C}|\mathsf{B}) \leq \cos(\theta)^{n-2}. \label{eq:CMIDecayCluster_intro}
\end{align}
\end{theorem}
The proof of this theorem is provided in \secref{CMIChainQubits}.
While \eqref{eq:CMIDecayCluster_intro} provides a provable upper bound on conditional correlations, \fig{ClusterState}(b) presents numerically simulated values of $I(\mathsf{A}:\mathsf{C}|\mathsf{B})$ for varying angles $\theta$, confirming that the conditional correlation length increases as~$\theta \rightarrow 0$.

\begin{figure*}[t!]
    \centering
\includegraphics[width=\textwidth]{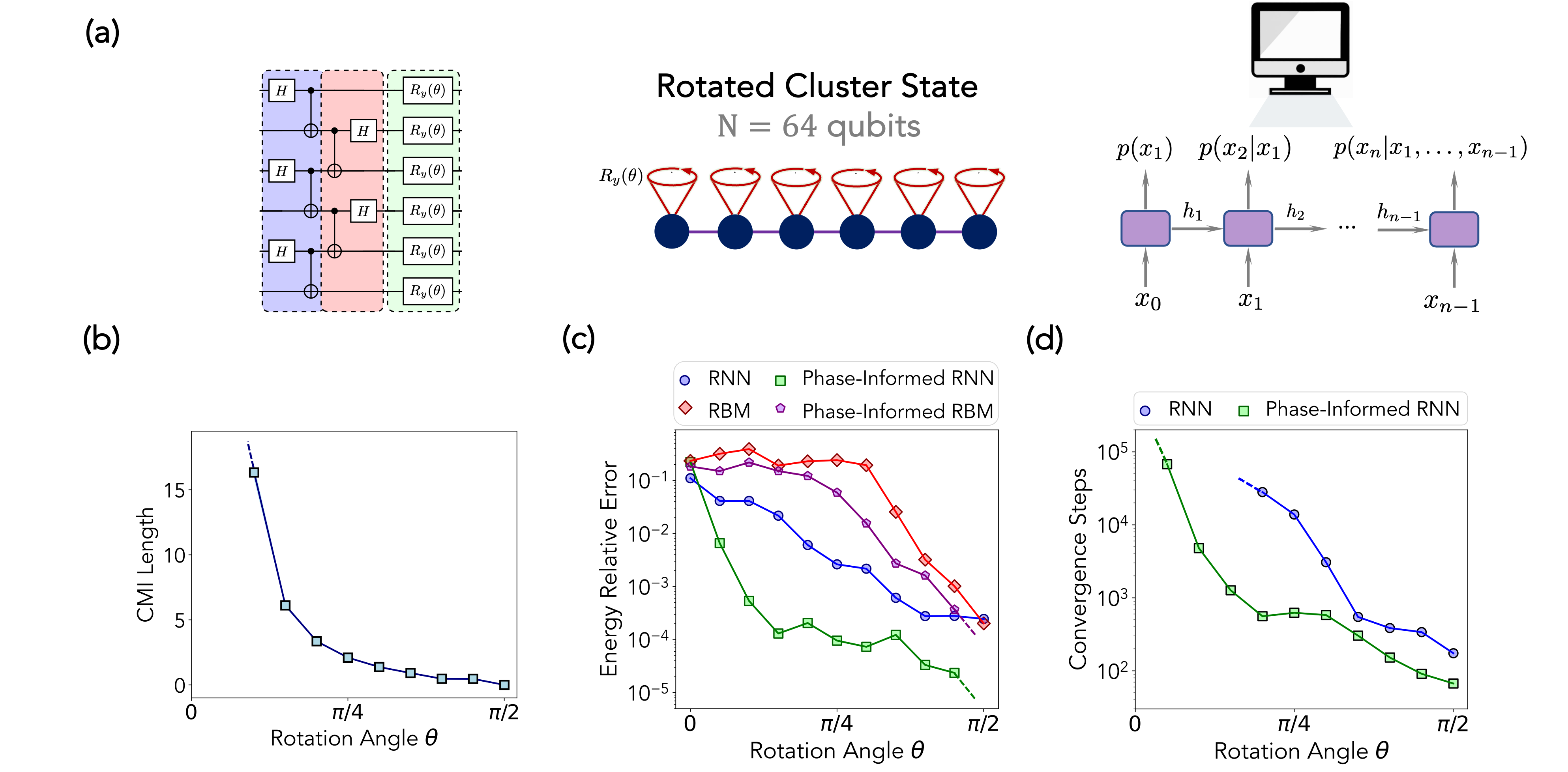}
\caption{\textbf{Effect of conditional correlations on Variational Monte Carlo:} \textbf{(a)} The rotated 1D cluster state is the ground state of the Hamiltonian $R_y^\dagger(\theta) H_{\mathrm{ES}} R_y(\theta)$, where $H_{\mathrm{ES}}$ is defined in Eq.~\eqref{eq:rotated_clusterH} and each qubit is rotated around the $y$-axis by an angle $\theta$. 
The performance of the variational Monte Carlo (VMC) with neural quantum states in estimating the energy of this rotated state with $64$ qubits is analyzed as a function of $\theta$. 
The true ground state energy is $E_0 = -64$.
\textbf{(b)} As $\theta$ increases from $0$ to $\pi/2$, the CMI length $\xi$ \eqref{eq:ConditionalIndependence} monotonically decreases, indicating that the state becomes conditionally independent.
For clarity, the data points corresponding to $\theta \in \{0, \pi/20\}$, where $\xi$ diverges, are not shown.
\textbf{(c)} The relative error in estimating the energy as a function of $\theta$ is shown for the RNN and RBM architectures, with and without complex phases in the amplitudes given as input to optimization. Overall, we observe a degradation in the efficacy of VMC as $\theta$ approaches $0$, where long-range conditional correlations are present.
The negligible error for $\theta = 0$ in phase-informed cases is not shown.
Focusing only on qualitative behavior, we did not optimize the relative performance between the RBM and RNN architectures. 
\textbf{(d)} The number of training steps required for the RNN ansatz to converge to mean energy $E = -63$ is shown. The omitted data points indicate non-convergence within $100,000$ training steps.
The number of training steps for the RBM architecture is $10,000$. 
}
\label{fig:ClusterState}
\end{figure*}

The rotated cluster state admits an efficient neural network representation.
This follows from the fact that this state has a matrix product representation with bond dimension $r=2$, which can be directly mapped to a recurrent neural network. 
An explicit construction of such a mapping is given in \cite{Wu2023RNNfromMPS} and a general correspondence between tensor network states and neural quantum states can be found in~\cite{NNrepresentationTNSharir}. 
Starting from a matrix product state with bond dimension $r$, the resulting RNN consists of hidden memories $h_k \in \mathbb{C}^r$ and cells that output conditional probabilities and complex phases by computing a weighted norm and affine map of $h_k$.

Despite the existence of an efficient RNN representation for the rotated 1D cluster state, we show that variationally finding this representation for \emph{small} rotation angles with long-range conditional correlations could be computationally difficult.
We employ Variational Monte Carlo (VMC) algorithm reviewed in \secref{VMC} for this task. 
This algorithm solves the optimization problem $\min_w \bra{\psi_w}H\ket{\psi_w}$ over a set of trial states, which in this case is a family of neural quantum states $\ket{\psi}_{w}$, each with a set of weights denoted by $w$.
Granted we can query the amplitudes $\braket{x}{\psi_w}$ and sample from $\abs{\braket{x}{\psi_w}}^2$, the energy $\bra{\psi_w}H\ket{\psi_w}$ can be efficiently evaluated. 
Additionally, if the neural quantum states are expressive enough with a favorable optimization landscape, the estimated energy $\min_w \bra{\psi_w}H\ket{\psi_w}$ can closely approximate the ground state energy of the local Hamiltonian~$H$.

We apply VMC to the rotated cluster state Hamiltonian \eqref{eq:rotated_clusterH} over $64$ qubits with two choices of neural networks as our variational ansatz: Recurrent Neural Networks (RNN) \cite{Hibat2020RNN} and Restricted Boltzmann Machines (RBM) \cite{carleo2017solving}.
Our numerical findings are summarized in \fig{ClusterState} and described in more detail in Appendix~\ref{sec:detailsNumerics}.
We observe that, in general and for both architectures, as $\theta\to 0$ and longer-range conditional correlations form, the performance of VMC degrades almost monotonically. 
This suggests that conventional optimization schemes may struggle to find the neural representation of simple classes of states with long-range conditional correlations.
Similar difficulties in learning long-term dependencies in recurrent networks using gradient descent are well-known in machine~learning~\cite{bengio1993problem}.

The longer-range correlations for smaller values of $\theta$ arise due to changes in both the magnitudes and phases of the probability amplitudes $\psi(x)$. 
As discussed in \secref{nonNegativAmp}, the emergence of a complex phase structure is likely necessary for the development of these longer-range correlations.
To what extent does the performance degradation of VMC, observed in \fig{ClusterState} for the rotated cluster state, result from changes in the complex phases $\psi(x)/|\psi(x)|$ versus the magnitudes $|\psi(x)|$? 
Indeed, complex phases are known to pose significant challenges for training neural quantum states with VMC or in quantum tomography applications \cite{parkKastoryano2020expressive, Carrasquilla2019Gutzwiller,Castelnovo2020signproblem, Bukov2021nonStoquastic}.

To address this, we first exactly determine the phases $\psi(x)/|\psi(x)|$ in the rotated cluster state. 
We then provide these phases as input to VMC, training the neural network solely to learn the magnitudes $|\psi(x)|$ of the state. 
The results shown in \fig{ClusterState}(c)-(d) indicate that the performance of VMC improves and achieves a more accurate approximation of the ground state energy except in the vicinity of $\theta = 0$, although the relative error and training steps required still increase as $\theta$ decreases.
In the vicinity of $\theta = 0$, VMC does not converge to the ground state even when the phases are known, underscoring the challenge of training the magnitudes $|\psi(x)|$ in regimes with long-range conditional correlations.
While having the phases improves VMC's performance, determining these phases or finding a local change of basis to remove them in general systems is computationally challenging~\cite{hangleiter2020easing}.

\subsection{Higher dimensions}\label{sec:HigherDimensions}
The 1D rotated cluster state serves as an example of a state with long-range conditional correlations, where the conditional probabilities $p(x_k \mid x_1, \dots, x_{k-1})$ are simple functions of their inputs and can be efficiently represented using neural networks.
Such correlations have also been observed in higher-dimensional variants of this state, such as 2D graph states \cite{Liu2022GraphState}.

Beyond these structured states, how prevalent are long-range conditional correlations in physically relevant quantum states, and do their conditional probabilities $p(x_k|x_1,\dots, x_{k-1})$ easily admit efficient neural representation?
Here we present evidence from a series of recent works that suggest certain random ensembles of states possess long-range conditional correlations with conditional probabilities $p(x_k|x_1,\dots, x_{k-1})$ 
which are complex functions not easily representable with neural networks.

Consider states generated by random ensembles of shallow geometrically-local quantum circuits or tensor networks with small bond dimension in two- or higher-dimensional lattices.
As discussed in \secref{lowEntanglement}, such states are believed to undergo a phase transition beyond a certain critical depth $d^*$ or bond dimension $r^*$ where the conditional correlations transition from short-range to long-range with high probability.
In these examples, the critical depth $d^*$ and the bond dimension $r^*$ are small constants independent of the number of qubits.
For instance, numerical simulations in 2D random circuits with the brickwork architecture suggest $d^*=6$ \cite{BeneWatts2024MIE, bao2021finite, napp2019efficient2d}, while for random PEPS, the critical bond dimension may be as small as $r^*=2$ \cite{levy2021entanglement, Gonzalez2024CMIinPEPS}. 
Although a rigorous proof of these phase transitions is still lacking, recent work \cite{BeneWatts2024MIE} shows that random 2D shallow quantum circuits satisfying a notion of anti-concentration---believed to occur at depths $\geq \Omega(\log(n))$ for the brickwork architecture---exhibit long-range conditional correlations.
This can, in fact, be unconditionally proven for a depth-$2$ ``coarse-grained'' circuit architecture composed of random gates acting on $\mathcal{O}(\log(n))$ qubits.
This  was further improved in recent work \cite{McGinley2024MIE}, which proves that certain random quantum circuit architectures generate long-range conditional correlations at constant depth and fixed qudit dimension.

The presence of long-range conditional correlations in the output states of random 2D shallow random circuits or tensor network states poses a challenge to accurately represent these states with neural networks. 
Here, we explore evidence for this from previous works.
Mapping the tensor network representation of these states to a neural network, as we previously did for the 1D cluster state, does not immediately apply here due to the difficulty of contracting the corresponding 2D tensor networks \cite{napp2019efficient2d, Gonzalez2024CMIinPEPS}. 
Moreover, the intricate patterns of multipartite measurement-induced entanglement, which non-locally appear everywhere in such randomly generated states \cite{BeneWatts2024MIE, bao2021finite, napp2019efficient2d}, suggest that the conditional probabilities $p(x_k|x_1,\dots, x_{k-1})$ are likely complex functions of their inputs $x_1,\dots, x_{k}$. 
Focusing on recurrent neural networks, this suggests that deeper feedforward neural networks in RNN cells are required to compute these complex dependencies. Additionally, the dimension of the hidden state vectors $h_k$ must scale as $\Omega(n)$ with the number of qubits to retain the information of previous variables $x_1, \dots, x_{k-1}$ for $k \in [n]$.

The complexity introduced by these long-range entanglement patterns has been further investigated in \cite{BeneWatts2024MIE}, where the authors provide formal evidence that classical probabilistic circuits whose depth grows sub-logarithmically with $n$ cannot simulate the output distribution $p(x_1,\dots, x_n)$ of random 2D shallow quantum circuits given a description of the quantum circuit as input. 
In fact, it is conjectured \cite{napp2019efficient2d} that this limitation extends to any polynomial-time classical algorithm.
There are also complexity theory barriers to the very precise computation of the output probabilities $p(x_1,\dots, x_n)$ of random shallow circuits.
In particular, it is known that achieving additive errors of $2^{-\Omega(n^2)}$ is \#$\mathsf{P}$-hard \cite{movassagh2023hardness,Bouland2022Frontier}.
Finally, the worst-case complexity of estimating the ground state energy of local Hamiltonians, whose ground state amplitudes admit a succinct representation (e.g., via classical neural networks) is shown to be $\mathrm{MA}$-complete \cite{jiang2023local}, where $\mathrm{MA}$ is a randomized extension of the complexity class $\mathrm{NP}$.

\section{Outlook}

Our work provides evidence that the performance of neural network representation of quantum states depends on the amount of conditional correlations in the state's measurement distribution.
Our results raise several intriguing open questions, some of which are briefly highlighted below.

\paragraph{Probability amplitudes with non-trivial phases:} We examined the expressive power of classical neural networks in capturing the magnitudes $|\psi(x)|$ of probability amplitudes, leaving the complex phases $\psi(x)/|\psi(x)|$ mostly unaddressed. 
What structures might these complex phases assume in physically relevant states, such as the ground states of local Hamiltonians?
It is known that not all these complex phases can be removed by a local change of basis, and when present, they degrade the performance of neural quantum states in learning and simulation tasks \cite{parkKastoryano2020expressive, Carrasquilla2019Gutzwiller,Castelnovo2020signproblem, Bukov2021nonStoquastic}.
In connection with this, we might ask whether the unique ground state of any gapped, sign-free local Hamiltonian exhibits short-range conditional correlations, or if there exist concrete counterexamples.
Resolving this question is tied to understanding the performance of Variational Monte Carlo (VMC). 
VMC avoids the sign problem by sampling directly from a parametrized wavefunction, rather than from high-dimensional path integrals or density matrices, as in quantum Monte Carlo (QMC). However, while the sign problem is not apparent in VMC’s construction, it might still implicitly encounter similar challenges due to limitations in efficiently representing the wavefunction.

\paragraph{Gapless phases and power-law interactions:}
Our analysis of conditional correlations focused on gapped local Hamiltonians. 
What is the nature of these correlations near critical points, in gapless phases, or when interactions are not strictly local but instead decay with distance according to a power law?

\paragraph{Practical performance with long-range conditional correlations:} In two- and higher-dimensional lattices, states prepared by random quantum circuits or tensor networks likely exhibit long-range conditional correlations, even at shallow depths or small bond dimensions.
While this suggests that precise neural representations of such states may be inefficient, the impact on practical applications, which often rely on coarser approximations and pertain to more structured states, remains unclear.
Can the performance of neural quantum states in representing long-range correlations be enhanced by increasing network parameters or designing tailored architectures, thereby improving on the methods in \fig{ClusterState}?

\paragraph{Other neural network architectures:} Modern architectures such as Receptance Weighted Key Value (RWKV) \cite{peng2023rwkv} and Transformer Quantum State (TQS) \cite{sprague2024variational} are equipped with a self-attention mechanism designed to better capture long-range correlations in speech or language modeling. 
Thus, one might hope that using them to solve the ground state problem with long-range correlations could lead to improved performance.
Our preliminary numerical simulations show that, with a limited number of parameters, these networks exhibit similar performance degradation in the presence of long-range conditional correlations. A more thorough investigation of these architectures is left for future work.

\subsection*{Code and Data Availability:}

The code and data for conducting the numerical experiments and for generating the figures in this work are openly available on GitHub \url{https://github.com/xiaotai-yang/NQS_cmi}.

\vspace{0.5em}

\subsection*{Acknowledgments:}
The authors thank Adam Bene Watts, Fernando Brand\~ao, Sergey Bravyi, David Gosset, Hsin-Yuan Huang, and Yinchen Liu for helpful discussions.
MS is grateful to John Wright for many discussions regarding the analysis of the entanglement swapping in \secref{CMIin1D}, and Tomotaka Kuwahara for suggesting the proof of \lemref{BoundonEntropy}.
MS thanks John Napp for insightful discussions in the early stage of this project, especially regarding \secref{areaLaw} and the relevance of \cite{hastings2016signfree,napp2019efficient2d}.
The computations presented in this work were conducted at the Resnick High Performance Computing Center, a facility supported by the Resnick Sustainability Institute at the California Institute of Technology, Illinois Campus Cluster Program at the University of Illinois at Urbana-Champaign, along with resources generously provided by Ying-Jer Kao at National Taiwan University.
MS is supported by AWS Quantum Postdoctoral Scholarship and funding from the National Science Foundation NSF CAREER award CCF-2048204.
JP acknowledges support from the U.S. Department of Energy Office of Science, Office of Advanced Scientific Computing Research (DE-NA0003525, DE-SC0020290), the U.S. Department of Energy, Office of Science, National Quantum Information Science Research Centers, Quantum Systems Accelerator, and the National Science Foundation (PHY-1733907). 
The Institute for Quantum Information and Matter is an NSF Physics Frontiers Center.

\bibliographystyle{plain}
\bibliography{main}
\onecolumngrid
\appendix
\newpage

\setcounter{section}{0}
\renewcommand*{\thesection}{\Alph{section}}

\section{Preliminaries}

\subsection{Architectures of neural quantum states}\label{sec:NeuralNetworks}

Neural networks can succinctly represent an $n$-qubit quantum state 
$\ket{\psi} = \sum_{x \in \{0,1\}^n} \psi(x) \ket{x}$
in the computational basis $\{\ket{x} : x \in \{0,1\}^n\}$ by encoding the complex amplitudes $\psi(x) \in \mathbb{C}$.
We denote the quantum state parametrized by neural network weights $w$ as $\ket{\psi_w}$.
When the context is clear, we may omit the dependence on $w$.
This, for instance,  can be achieved by specifying the squared-magnitude $p(x)$ and the phase $g(x)$ parts of each amplitude $\braket{x}{\psi} = \sqrt{p(x)} e^{i g(x)}$ or the real $\operatorname{Re}(\psi(x))$ and imaginary $\operatorname{Im}(\psi(x))$ components. 
Alternatively, the network may directly output $\psi(x)$ by allowing complex-valued neurons or using architectures such as complex-valued restricted Boltzmann machines \cite{carleo2017solving}. 
We primarily focus on the case of polar representation in terms of functions $p(x)$ and $g(x)$ using real weights $w$ for a clearer exposition.

\vspace{\baselineskip}

\noindent \textbf{Feedforward neural networks:}
The first neural network architecture we consider is the feedforward network, shown in \fig{FNNRNN}.
Such a network consists of a series of layers including input, hidden, and output layers.
The depth of the neural network refers to the number of layers $L$ after the input layer.
Each layer $i \in [L]$ consists of $n_i$ neurons which are connected to the following layer with certain weights. 
The maximum number of neurons over all layers $W = \max_{i \in [L]} n_i$ is known as the width of the~network.
The neuron located at position $j \in [n_i]$ within layer $i \in [L]$ operates by transforming the input vector $z \in \mathbb{R}^{n_{i-1}}$, through a function consisting of two main components:
an \emph{affine} transformation characterized by $\langle w_{ij} | z \rangle + b_{ij}$, where the vector $w_{ij} \in \mathbb{R}^{n_{i-1}}$ and the scalar $b_{ij} \in \mathbb{R}$ represent the weight and bias parameters, respectively;
and also, a nonlinear transformation by the activation function $\sigma (\cdot )$. 
Common choices for $\sigma(\cdot)$ include the rectified linear unit (ReLU), defined as $z \rightarrow \max(0, z)$, and the sigmoid function, given by $z \rightarrow \frac{1}{1 + e^{-z}}$.
Together, these elements enable the neuron to perform complex mappings from inputs to outputs, contributing to the neural network's ability to model nonlinear~relationships.

\vspace{\baselineskip}

\noindent \textbf{Recurrent neural networks:}
A probability distribution $p(x_1,\dots, x_n)$ can be expressed using its conditional distributions by 
\begin{align}
    p(x_1,\dots, x_n) = p(x_1)p(x_2|x_1)\cdots p(x_n|x_1,\dots,x_{n-1}).\label{eq:conditinalForm}
\end{align}
A recurrent neural network, shown in \fig{FNNRNN}, consists of $n$ recurrent cells.
The cell $i\in [n]$ takes as input a hidden variable $h_{i-1} \in \mathbb{R}^{n_{i-1}}$ and the binary value $x_{i-1} \in \{0,1\}$ (with $h_0$ and $x_0$ some fixed constant values).
The output of the cell is an updated hidden variable $h_{i} \in \mathbb{R}^{n_{i}}$ given by $h_{i} = f_h(x_{i-1}, h_{i-1})$ and a binary probability distribution given by $f_p(h_i) \in \mathbb{R}^2$.
Here, $f_h(\cdot)$ and $f_p(\cdot)$ are functions composing an affine and a non-linear activation maps similar to the neurons in the feedforward neural network.

The distributions generated by each recurrent cell is used to specify the conditional distribution $p(x_i|x_1,\dots, x_{i-1}) = f_p(h_i)$.
This along with the conditional form in \eqref{eq:conditinalForm} can be used to directly sample from the distribution $p(x_1,\dots, x_n)$, a feature known as the \emph{autoregressive} property.
This is achieved by first sampling $\bm{x_1}\sim p(x_1)$ using the output of the first recurrent cell in the network. 
Next, the sampled $\bm{x_1}$ is given as input to the second recurrent cell whose output value is used to sample $\bm{x_2} \sim p(x_2|x_1)$. 
Continuing this way, we end by sampling $\bm{x_n} \sim p(x_n|x_1,\dots, x_{n-1})$.
The conditional distributions in each step can be efficiently evaluated and a binary value is sampled consequently.
This renders the overall sampling $\bm{x_1},\dots, \bm{x_n} \sim p(x_1,\dots, x_n)$ efficient, contrasting with the exponential cost of sampling a general multivariate distribution.
Such networks define recurrent neural network wavefunctions by expressing the measurement distribution $p(x)$ of a quantum state $\sum_{x \in \{0,1\}^n} \sqrt{p(x)} e^{i g(x)}$ \cite{Hibat2020RNN, Hibat2023RnnTopology}. 
Extensions of this architecture for generating the phase components $g(x)$ have been also considered \cite{Hibat2020RNN}.

\vspace{\baselineskip}
\noindent \textbf{Restricted Boltzmann machines:}
In this case, the network features a visible and a single hidden layer of neurons. 
The visible layer has $n$ nodes taking values $x = (x_1,\dots,x_n)\in \{0,1\}^n$ that specify the computational basis as before. 
The $n_h$ neurons in the hidden layer are binary-valued and fully connected to those in the visible layer by complex weights $\{w_{ij}\}_{i\in [n], j\in [n_h]}$.
The visible and hidden neurons also admit complex weights $\{a_i\}_{i \in [n]}$ and $\{b_j\}_{j \in [n_h]}$. 
The expressivity of this architecture can be controlled by adjusting the number of hidden neurons $n_h$.
The output of the network with weights $w= \cup_{ij} \{w_{ij}, a_i, b_j\}$ is the un-normalized complex amplitudes $\psi_w(x)$ given by
\begin{align}
    \psi_w(x) &= \sum_{y \in \{0,1\}^n} e^{\sum_{i,j} w_{ij} x_i y_j +  \sum_i a_i x_i + \sum_j b_j y_j} \nonumber\\
    &= \prod_{j\in[n_h]}(1+e^{b_j+ \sum_i w_{ij}x_i})  \cdot  e^{\sum_i a_i x_i}\label{eq:RBMExpression}
\end{align}
We see from \eqref{eq:RBMExpression} that for a bounded number of hidden neurons $n_h = \poly(n)$, the output of the network can be efficiently evaluated.
This is in contrast to the deeper version of these networks, featuring multiple hidden layers, where computing their output may become computationally intractable. 
Despite their simple structures, restricted Boltzmann machines can efficiently express a variety of quantum states \cite{carleo2017solving, gao2017efficient, Gao2019RBM} and are conventionally used in the Monte Carlo simulations of ground state properties \cite{carleo2017solving, melko2019restricted}.

\subsection{Review of variational Monte Carlo}\label{sec:VMC}

Neural quantum states offer a computationally powerful model that can be used to find the ground state properties of quantum systems. 
The underlying idea is that the ground state of an interacting quantum system described by a Hamiltonian $H$ can be approximately found by solving the variational~problem 
$\min_w \ \bra{\psi_w} H \ket{\psi_w}$,
where the energy of a family of neural quantum states $\ket{\psi_w}$ is minimized over their set of weights~$w$.
It is not hard to see that an equivalent expression for this minimization is given by
\begin{align}
    \min_w \ \bra{\psi_w} H \ket{\psi_w} = \min_w \E_{\bm{x}\sim |\psi_w(x)|^2}\underbrace{\left(\sum_{x'}\bra{\bm{x}}H|x'\rangle\cdot \frac{\psi_w(x')}{\psi_w(\bm{x})}\right)}_{=: \mathcal{E}_{\mathrm{local}}(\bm{x})}.\label{eq:vmcOptimization1}
\end{align}
The resulting optimization algorithm is known as the variational Monte Carlo (VMC) and has been applied in a series of recent works to find precise estimates of the ground state properties. 
For the objective function in minimizaiton \eqref{eq:vmcOptimization1} to be computable in $\poly(n)$ time, one needs 
\begin{itemize}
    \item[(1)] the term  $\mathcal{E}_{\mathrm{local}}(\bm{x})$, known as the \emph{local energy}, to be bounded and can be efficiently evaluated, 
    \item[(2)] an algorithm that efficiently draws samples $\bm{x}$ from the measurement distribution~$|\psi_w(x)|^2$.
\end{itemize}
The first condition is satisfied when $H$ is a local (or more generally a sparse) Hamiltonian and the amplitudes $\psi_w(x)$ can be computed in $\poly(n)$ time up to a normalization factor, a property that is satisfied for neural quantum states.
Assuming the efficient sampling in condition (2) can be satisfied as well, the objective function is estimated by taking the empirical average of many samples $\bm{x_1},\dots,\bm{x_T}$ drawn independently from $|\psi_w(x)|^2$ using Markov chain sampling algorithms such as Metropolis–Hastings \cite{Carleo2019CNN, Bravyi2023rapidlymixingmarkov, Bravyi2022MCMC} or autoregressive methods such as those used in recurrent neural networks \cite{Hibat2020RNN}.
One can estimate $\bra{\psi}H\ket{\psi}$ up to an error $\epsilon$ with probability $\geq 1-\delta$ using a total of $T \leq  \mathcal{O}(\bra{\psi}H^2\ket{\psi}/\epsilon^2)$ samples $\bm{x} \sim |\braket{x}{\psi}|^2$.
This follows from standard concentration results and the fact that the variance of the real and imaginary parts of local energy can be bounded by $ \bra{\psi_w}H^2\ket{\psi_w} \leq \norm{H^2}$ (refer e.g. to \cite{huang2024certifying} for a detailed derivation).

This leads to a stochastic optimization problem that can be solved using methods such as stochastic reconfiguration (or natural gradient descent) \cite{SorellaStochReconfig1998, carleo2017solving}.
The performance of variational Monte Carlo in finding the ground state is contingent upon the optimization landscape \eqref{eq:vmcOptimization1} over the set of weights $w$. 
Crucially, for this approach to succeed, there must exist a set of weights $w^*$ such that the neural quantum state $\ket{\psi_{w^*}}$ approximately represents the true ground state. 
Addressing this issue of \emph{expressivity} is the primary objective of this work.

\section{Decay of CMI Implies Efficient Representations}\label{sec:CMIimpliesRep
}

\subsection{Shallow feedforward neural network (FNN) representation}\label{sec:CNNfromCMI}
In this section, we show that the magnitudes $|\psi(x)|$ of the ground state amplitudes admit an approximate representation using shallow feedforward neural nets.
We start with deriving a factorized form for the measurement distribution $p(x) = |\psi(x)|^2$ assuming it satisfies the conditional independence property.

\begin{theorem}\label{thm:factorizaiton}
Consider the setup of \thmref{CMI_line_2d} and suppose the measurement distribution $p(x_1,\dots,x_n)$ satisfies the approximate conditional independence property $$I(\mathsf{A}:\mathsf{C}|\mathsf{B})\leq n^{c_0} \cdot e^{-\dist(\mathsf{A},\mathsf{C})/\xi}$$
as in \defref{ApproximateConditionalIndependence} with some constant $c_0$ and the CMI length $\xi$.
Fix parameters $M=\mathcal{O}\left(\frac{n}{\ell_0^{D_{\Lambda}} \cdot \log^{D_{\Lambda}}(n)}\right)$ and $s=\mathcal{O}(\ell_0^{D_{\Lambda}}\cdot \log^{D_{\Lambda}}(n))$. Then the distribution $p(x_1,\dots,x_n)$ in its support can be approximately factored in its support as
\begin{align}
    \sum_{x\in\{0,1\}^n}\left\lvert p(x)-\frac{\prod_{k=1}^M p(x_{s_k})}{\prod_{k=1}^{M-1} p(x_{s'_k})}\right\rvert \leq n^{-(\ell_0/\xi-c_0-1)}.\label{eq:factorizationTV}
\end{align}
Here $p(x_{s_k})$ and $p(x_{s'_k})$ are marginal distributions over variables  $x_{s_k}$ and $x_{s'_k}$ supported on $s_k,s'_k\subseteq [n]$ such that $\max_{k \in[M]} |s_k| \leq s$ and $\max_{k\in[M-1]} |s'_k| \leq s$.
\end{theorem}
\begin{proof}
   We start with an equivalent formulation of the conditional independence property. 
   For any three subsets $\mathsf{A}$,$\mathsf{B}$, and $\mathsf{C})$ such that $\dist(\mathsf{A},\mathsf{C})=2 \ell_0 \cdot \log(n)$ and $p(x_{\mathsf{B}}) > 0$ we have
         \begin{align}      \sum_{x_{\mathsf{B}}} p(x_{\mathsf{B}} )\sum_{x_{\mathsf{A}} , x_{\mathsf{C}}} \left|\frac{p(x_{\mathsf{A}} , x_{\mathsf{B}} , x_{\mathsf{C}} )}{p(x_{\mathsf{B}})} - \frac{p(x_{\mathsf{A}} , x_{\mathsf{B}})}{p(x_{\mathsf{B}} )}\cdot \frac{p(x_{\mathsf{B}} , x_{\mathsf{C}})}{p(x_{\mathsf{B}} )}\right| &=\E_{\mathsf{B}}\norm{p_{{\mathsf{A}} ,{\mathsf{C}} |{\mathsf{B}}}-p_{{\mathsf{A}} |{\mathsf{B}}}\otimes p_{{\mathsf{C}}|{\mathsf{B}}}}_1 \nonumber\\
     &\leq  \sqrt{2 \cdot I(\mathsf{A}:\mathsf{C}|\mathsf{B})}\nonumber\\
         &\leq \poly(|\mathsf{A}|, |\mathsf{C}|)\cdot n^{-\ell_0/\xi}\nonumber\\
         &\leq n^{-(\ell_0/\xi-c_0)},
              \label{eq:decompositionOfDistribution}
      \end{align}
      where the expectation is over $\bm{x}_B\sim p_{\mathsf{B}}(x_{\mathsf{B}})$ and $c_0$ is a suitable constant.
      Next, we are going to apply this bound repeatedly to factorize the measurement distribution $p(x)$.
    \paragraph{One-dimensional chain:} We first consider the case when the qubits are arranged on an open chain with $D_{\Lambda}=1$.
    Divide the chain to contiguous sets $\mathsf{A}_1,\mathsf{A}_2,\dots, \mathsf{A}_M$ where each $|\mathsf{A}_i|=2\ell_0 \cdot \log (n)$. 
    Replacing $\mathsf{A}= \mathsf{A}_1$, $\mathsf{B} = \mathsf{A}_2$ and $\mathsf{C} = \mathsf{A}_3,\dots,\mathsf{A}_M$ in bound \eqref{eq:decompositionOfDistribution} yields:
    \begin{align}
       \Norm{p_{\mathsf{A}_1,\mathsf{A}_2,\dots,\mathsf{A}_M}-p_{\mathsf{A}_1,\mathsf{A}_2}\cdot p^{-1}_{\mathsf{A}_2}\cdot p_{\mathsf{A}_2,\dots,\mathsf{A}_M}}_1\leq n^{-(\ell_0/\xi-c_0)}\label{eq:decomp1}
    \end{align}
    Similarly, we can apply \eqref{eq:decompositionOfDistribution} to the marginal distribution $p_{\mathsf{A}_2,\dots,\mathsf{A}_M}$.
    Continuing in this manner, in the $k$'th step, we factorize $p_{\mathsf{A}_k,\dots,\mathsf{A}_M}$ to
    \begin{align}
        \Norm{p_{\mathsf{A}_k,\dots,\mathsf{A}_M} - p_{\mathsf{A}_{k},\mathsf{A}_{k+1}}\cdot p^{-1}_{\mathsf{A}_{k+1}}\cdot p_{\mathsf{A}_{k+1},\dots,\mathsf{A}_{M}}}_1\leq n^{-(\ell_0/\xi-c_0)}\label{eq:decomp2}
    \end{align}
    Combining these two bounds \eqref{eq:decomp1} and \eqref{eq:decomp2}, we get
          
        \begin{align}
        \NORM{p_{\mathsf{A}_1,\dots,\mathsf{A}_M}-\prod_{k=1}^{M-2} \frac{p_{\mathsf{A}_k,\mathsf{A}_{k+1}}}{p_{\mathsf{A}_{k+1}}} \cdot p_{\mathsf{A}_M,\mathsf{A}_{M-1}}}_1\leq M\cdot n^{-\ell_0/\xi}\leq n^{-(\ell_0/\xi-c_0-1)}
        \end{align}

\paragraph{$\bm{D_{\Lambda}}$-dimensional lattice:} We now move the more general case of a $D_{\Lambda}$-dimensional lattice $\Lambda$, which for simplicity, we assume is the cubic lattice $\mathbb{Z}^{D_{\Lambda}}$. Suppose we divide the lattice into a collection of cubic subsets $\cup_{k=1}^M \mathsf{A}_k$, where each $\mathsf{A}_k$ has side-length $\ell_0 \cdot \log(n)$ and the total number of $\mathsf{A}_k$ is $M=\lceil \frac{n}{\ell_0^{D_{\Lambda}} \cdot \log^{D_{\Lambda}}(n)}\rceil$. For each $\mathsf{A}_k$, we define its neighborhood $\partial \mathsf{A}_k \subset \cup_{j=1}^M \mathsf{A}_j$ such that the qubits in $\partial A_k$ shield $\mathsf{A}_k$ from $\Lambda\setminus (\mathsf{A}_k \cup \partial \mathsf{A}_k)$.

We proceed in a similar way as in the one-dimensional case:
Define $\mathsf{C}_1:=\Lambda \setminus \{\mathsf{A}_1 \cup \partial \mathsf{A}_1\}$ and $\mathsf{B}_1:=\partial \mathsf{A}_1$. Consider the tripartition of $\Lambda$ into $\mathsf{A}_1$, $\mathsf{B}_1$, and $\mathsf{C}_1$.
See \fig{CMI_line_2d} for a demonstration of the $D_{\Lambda}=2$ case.
From bound \eqref{eq:decompositionOfDistribution} we~get:
\begin{align}
    \norm{p_{\mathsf{A}_1,\dots,\mathsf{A}_M}-p_{\mathsf{A}_1,\mathsf{B}_1}\cdot p^{-1}_{\mathsf{B}_1}\cdot p_{\mathsf{B}_1,\mathsf{C}_1} }_1\leq n^{-(\ell_0/\xi-c_0)}\label{eq:factorize1}
\end{align}
We then repeat this argument to factorize $p_{\mathsf{B}_1,\mathsf{C}_1}$.
In the $k$'th step of this process, we want to factorize $p_{\mathsf{B}_{k-1},\mathsf{C}_{k-1}}$. Consider the partitioning of set $\mathsf{B}_{k-1}\cup \mathsf{C}_{k-1}$ into $\mathsf{A}_k$, $\mathsf{B}_k=\partial \mathsf{A}_k \cap \mathsf{C}_{k-1}$, and $\mathsf{C}_k=\mathsf{C}_{k-1}\setminus \{\mathsf{A}_k \cup \mathsf{B}_k\}$. Applying bound \eqref{eq:decompositionOfDistribution} implies
\begin{align}
    \norm{p_{\mathsf{B}_{k-1},\mathsf{C}_{k-1}}-p_{\mathsf{A}_k,\mathsf{B}_k} \cdot p^{-1}_{\mathsf{B}_k}\cdot p_{\mathsf{B}_k,\mathsf{C}_k}}_1\leq n^{-(\ell_0/\xi-c_0)}\label{eq:factorize2}
\end{align}

Overall, we find
        \begin{align}
        \NORM{p_{\mathsf{A}_1,\dots,\mathsf{A}_M}-\prod_{k=1}^{M-2} \frac{p_{\mathsf{A}_k,\mathsf{B}_k}}{p_{\mathsf{B}_k}} \cdot p_{\mathsf{A}_M, \mathsf{A}_{M-1}}}_1&\leq M\cdot n^{-(\ell_0/\xi-c_0)}\leq n^{-(\ell_0/\xi-c_0-1)}.
        \end{align}

\end{proof}

\begin{remark}
    The proof of \thmref{factorizaiton} for the one-dimensional systems, only requires a weaker version of the approximate conditional independence \eqref{eq:ConditionalIndependence} where the sets $\mathsf{A}$, $\mathsf{B}$, and $\mathsf{C}$ are known to partition the whole chain, $\mathsf{A}\cup \mathsf{B} \cup \mathsf{C}=[n]$ (compared to applying \eqref{eq:ConditionalIndependence} to subsets $\mathsf{A}\cup \mathsf{B} \cup \mathsf{C}\subset[n]$).
    This is a consequence of the data processing inequality which implies that if the approximate conditional independence holds for a triparition  $\mathsf{A}$, $\mathsf{B}$, and $\mathsf{C}$ where $\mathsf{A} = \mathsf{A}_1 \cup \mathsf{A}_2$, then this property also holds for the subsets $\mathsf{A}_2$, $\mathsf{B}$, and $\mathsf{C}$. 
\end{remark}
We show next that \thmref{factorizaiton} can be applied to construct a shallow neural network representation for the magnitude $|\psi(x)|$ of the wavefunction amplitudes.

\begin{theorem}[Restatement of \thmref{CMI_line_2d}]\label{thm:FNNrepresentationRestatement}
    Consider a local Hamiltonian on $n$ qubits with the ground state $\ket{\psi}=\sum_{x \in \{0,1\}^n} \psi(x) \ket{x}$ defined on a $D_{\Lambda}$-dimensional lattice $\Lambda$ with open boundary conditions.
    Suppose the measurement distribution $p(x) = |\psi(x)|^2$ in this basis satisfies the approximate conditional independence property given in \defref{ApproximateConditionalIndependence} with respect to geometrically contiguous regions shown in \fig{CMI_line_2d}. Then, there exists a feedforward neural network with 
    \begin{align}
        \text{depth\ \ }  \mathcal{O}\left(\log\log(n)\right)  \text{ \ and width\ \ }  n^{\mathcal{O}(\log^{D_{\Lambda}-1}(n))}\label{eq:depthWidth}
    \end{align}
    that computes a function $q(x)$ such that $\sum_{x}|q(x)-p(x)|\leq 1/\poly(n)$.
    This can also be achieved using a recurrent neural network with $D_{\Lambda}$-dimensional geometry,  a memory size of $\mathcal{O}(\log^{D_{\Lambda}}(n))$, and  feedforward cells with width and depth given by \eqref{eq:depthWidth}.
\end{theorem}
We postpone the proof of the second part to \secref{RNNfromCMI} regarding the RNN representation to the next section and focus only on the feedforward representation in what follows.

\begin{proof}
We give an explicit construction of a feedforward neural network that computes the factorization in \thmref{factorizaiton}.
\begin{enumerate}
    \item Any function $q(x_1,\dots, x_m): \{0,1\}^m \mapsto \mathbb{R}$ can be expressed by $$q(x_1,\dots,x_m)=\sum_{(y_1,\dots,y_m)\in \{0,1\}^m} q_{(y_1,\dots,y_m)} \cdot \prod_{i=1}^m z_{i,y_i}(x_i),$$ where we define $z_{i,0}(x_i)=1- x_i$ and $z_{i,1}(x_i)=x_i$ for each $i\in[m]$. Since the variables $z_{i,y_i}\in\{0,1\}$, we can multiply them exactly via the ReLU function $\sigma(x) = \max(0, x)$. 
    To see this, consider the sawtooth function~\cite{telgarsky2015representation, yarotsky2017error}: $$g(x):=2\cdot \sigma(x)-4\cdot \sigma(x-\frac{1}{2})+2\cdot \sigma(x-1).$$  
    Define $h(x):=x-\frac{1}{4}\cdot g(x)$. This function satisfies $h(0)=0$, $h(1/2)=1/4$, and $h(1)=1$. The connection between $h(x)$ and $x \cdot y$ for two Boolean variables is due to $$x\cdot y=2\cdot h(\frac{x+y}{2})-\frac{x+y}{2}.$$
    Multiple repetition of this network allows us to compute $q(x_1,\dots,x_m)$ in depth $\mathcal{O}(\log(m))$ and width $\mathcal{O}(2^m)$ with the coefficients $q_{(y_1,\dots,y_m)}$ encoded in the weights of the network. 
    \item Next, we apply the above construction to evaluate the expression \eqref{eq:factorizationTV} in \thmref{factorizaiton}.
    That is, we can get
    \begin{align}
    \log \left(\frac{\prod_{k=1}^M p(x_{s_k})}{\prod_{k=1}^{M-1} p(x_{s'_k})}\right) = \sum_{k=1}^M \log(p(x_{s_k})) - \sum_{k=1}^{M-1}\log(p(x_{x'_k})).
    \end{align}
    by computing $\log(p(x_{\mathsf{s}_k}))$
        for $k \in[M]$ and $\log( p(x_{\mathsf{s}'_k}))$
        for $k \in [M-1]$ in parallel in depth $\mathcal{O}(\log(s))$ and width $\mathcal{O}(M\cdot 2^{s})$, where for a $D_{\Lambda}$-dimensional lattice $s = \mathcal{O}(\log^{D_{\Lambda}}{n})$ and $|\mathsf{s}_k|, |\mathsf{s}'_{k'}|\leq s$.
        We assume the weights of the neurons are are encoded in $\mathcal{O}(\log(n))$ bits. 
        Since $\log(p(x_{\mathsf{s}_k}))$ and $\log( p(x_{\mathsf{s}'_k}))$ are given by the weights of the neurons, we truncate the probability values that are  $< 1/(2^n \poly(n))$.
        Moreover, we can accommodate the cases when the value of a (possibly truncated) probability $p(x_{\mathsf{s}_k})$ is zero by including a constant positive offset in place of $\log(p(x_{\mathsf{s}_k}))$ and $\log( p(x_{\mathsf{s}'_k}))$.
        The presence of such an offset can be detected by a threshold gate, in which case the output probability computed by the network is considered~$0$. 
        \item The last layer of the network adds the above values and outputs a function that, due to the combined errors in step 2, is within $1/\poly(n)$ additive error of $\prod_{k=1}^M p(x_{s_k})/\prod_{k=1}^{M-1} p(x_{s'_k})$ in $\ell_1$ distance. 
         As a result of \thmref{factorizaiton}, we obtain a feedforward neural network with 
    \begin{align}
        \text{depth\ \ }  \mathcal{O}\left(\log\log(n)\right)  \text{ \ and width\ \ }  n^{\mathcal{O}(\log^{D_{\Lambda}-1}(n))}\label{eq:depthWidth3}
    \end{align}
    that computes a function $q(x_1,\dots, x_n)$ such that 
    $$\sum_{x_1,\dots, x_n}\left |p(x_1,\dots, x_n) -q(x_1,\dots, x_n)\right| \leq \frac{1}{\poly(n)}.$$
\end{enumerate}

\end{proof}

\subsection{Shallow recurrent neural network (RNN) representation}\label{sec:RNNfromCMI}
The conditional independence property in the state $\ket{\psi}$ also yields a recurrent neural network representations of $|\psi(x)|$ with $\mathcal{O}(\log^{D_{\Lambda}}(n))$ memory in a $D$-dimensional lattice.
As we discuss, the RNN architecture also incorporates a $D_{\Lambda}$-dimensional design \cite{Hibat2020RNN}.
To show this, analogous to the proof of \thmref{FNNrepresentationRestatement}, we start with a one-dimensional chain. 

\paragraph{One-dimensional chain:} As before, we divide the chain into contiguous sets $\mathsf{A}_1,\mathsf{A}_2,\dots, \mathsf{A}_M$ where each $|\mathsf{A}_j|=2\ell_0 \cdot \log (n)$. 
We now show that one can factorize the overall probability $p_{\mathsf{A}_1,\cdots,\mathsf{A}_M}$ into the product $\prod_{j=1}^M p_{\mathsf{A}_j|\mathsf{A}_{j-1}}$ with $\mathsf{A}_0 = \emptyset$.
This follows from
\begin{align}
    &\NORM{\prod_{j=1}^M p_{\mathsf{A}_j|\mathsf{A}_{< j}}-\prod_{j=1}^M p_{\mathsf{A}_j|\mathsf{A}_{j-1}}}_1 \nonumber\\
    &\leq \sum_{k=1}^{M-1}\NORM{\prod_{j=1}^k p_{\mathsf{A}_j|\mathsf{A}_{< j}}\cdot (p_{\mathsf{A}_{k+1}|\mathsf{A}_{\leq k}}-p_{\mathsf{A}_{k+1}\mathsf{A}_k})\cdot \prod_{j=k}^{M-2} p_{\mathsf{A}_{j+2}|\mathsf{A}_{j+1}}}_1\nonumber\\
    &\leq \sum_{k=1}^{M-1}\NORM{\prod_{j=1}^k p_{\mathsf{A}_j|\mathsf{A}_{< j}}\cdot (p_{\mathsf{A}_{k+1}|\mathsf{A}_{\leq k}}-p_{\mathsf{A}_{k+1}|\mathsf{A}_k})}_1\nonumber\\
    &=\sum_{k=1}^{M-1}\NORM{p_{\mathsf{A}_k}\cdot p_{\mathsf{A}_{k-1},\cdots,\mathsf{A}_1|\mathsf{A}_k}\cdot \left(\frac{p_{\mathsf{A}_{k+1},\mathsf{A}_{k-1},\dots, \mathsf{A}_1|\mathsf{A}_k}}{p_{\mathsf{A}_{k-1},\cdots, \mathsf{A}_1|\mathsf{A}_k}}-p_{\mathsf{A}_{k+1}|\mathsf{A}_k}\right)}_1\nonumber\\
    & = \sum_{k=1}^{M-1}\underbrace{\E_{\bm{\mathsf{A}_k}}\norm{p_{\mathsf{A}_{k+1},\mathsf{A}_{k-1},\dots, \mathsf{A}_1|\bm{\mathsf{A}_k}}-p_{\mathsf{A}_{k+1}|\bm{\mathsf{A}_k}}\cdot p_{\mathsf{A}_{k-1},\cdots, \mathsf{A}_1|\bm{\mathsf{A}_k}}}_1}_{\leq n^{-(\ell_0/\xi -c_0)}}\leq n^{-(\ell_0/\xi-c_0-1)}.\label{eq:factorizationRNN}
\end{align}

The last line is a consequence of the approximate conditional independence in \defref{ApproximateConditionalIndependence} after recognizing  $\mathsf{A}=\cup_{j=1}^{k-1}\mathsf{A}_j$, $\mathsf{B}=\mathsf{A}_k$, and $\mathsf{C}= \cup_{j=k+1}^{n}\mathsf{A}_j$ and choosing a suitable constant $c_0$ such that $\poly(\mathsf{A}, \mathsf{C}) \leq n^{c_0}$.

As shown in \fig{FNNRNN}, the recurrent neural network consists of a cascade of cells each computing the probability of a variable conditioned on the value of the variables in the previous cells.
Define $x^{(k)}_1,\cdots, x^{(k)}_{M'}$ to be the set of variables in each subset $\mathsf{A}_k$, where $M':= 2\ell_0 \log(n)$. 
We now focus on a subset of cells that computes the conditional probabilities $$p(x^{(k)}_1,\cdots, x^{(k)}_{M'}|x^{(k-1)}_1,\cdots, x^{(k-1)}_{M'}) = \prod_{i=1}^{M'} p(x^{(k)}_i|x^{(k)}_{< i}, x^{(k-1)}_{\leq M'})$$
for some~$k \in [M]$.
The cell that computes $p(x^{(k)}_i|x^{(k-1)}_{<i}, x^{(k-1)}_{\leq M'})$ 
takes as input $x^{(k)}_{i-1}$ and a hidden memory state $h^{(k)}_{i-1}$.
It then outputs an updated memory 
\begin{align}
    h^{(k)}_i = f_h(x^{(k)}_{i-1}, h^{(k)}_{i-1})\label{eq:hidden1}
\end{align}
and the conditional probability $$p(x^{(k)}_i|x^{(k)}_{< i}, x^{(k-1)}_{\leq M'}) = f_p(h^{(k)}_i).$$
The updated memory can be simply set $h^{(k)}_i = (x^{(k)}_{< i}, x^{(k-1)}_{\leq M'})$ which involves $\mathcal{O}(\log(n))$ many variables compared to $\Omega(n)$ variables needed for general $n$-variable distributions. 

The function $f_p(h^{(k)}_i)$ is constructed similar to the proof of \thmref{FNNrepresentationRestatement} which yields a shallow feedforward neural network of depth $\mathcal{O}(\log\log(n))$ and width $\poly(n)$ that computes $p(x^{(k)}_i|x^{(k)}_{< i}, x^{(k-1)}_{\leq M'})$ up to $\frac{1}{\poly(n)}$ multiplicative error. 
An argument similar to the one leading to \eqref{eq:factorizationRNN} shows the overall factorized distribution $ \prod_{k=1}^M\prod_{i=1}^{M'} p(x^{(k)}_i|x^{(k)}_{< i}, x^{(k-1)}_{\leq M'})$ can be computed with this recurrent architecture within $\frac{1}{\poly(n)}$ error in $\ell_1$ distance.
Combined with the error bound in \eqref{eq:factorizationRNN}, this yields a shallow recurrent neural network that computes and samples from $p(x_1,\dots,x_n)$ up to $\frac{1}{\poly(n)}$ error in $\ell_1$ distance. 

\paragraph{$\bm{D_{\Lambda}}$-dimensional lattice:}
The one-dimensional construction above can be generalized to the $D_{\Lambda}$-dimensional lattice $\Lambda = \mathbb{Z}^{D_{\Lambda}}$ by a proper choice of regions $\mathsf{A}_k$ to condition on in each RNN cell.
We first consider the $D_{\Lambda}=2$ lattice and partition it into regions $\cup_{k=1}^M \mathsf{A}_k$ with a `snake' ordering shown in \fig{RNNDecomposition} and a 2D RNN architecture such as the one considered in \cite{Hibat2020RNN}. 
The $D_{\Lambda}>2$ case follows from a similar argument. 
Analogous to the proof of \thmref{factorizaiton}, we consider the neighborhood $\partial \mathsf{A}_{k}$ of each subset $\mathsf{A}_k$.
In the proof of \eqref{eq:factorizationRNN} for the one-dimensional case, it suffices to factorize the distribution into a product of conditional distributions $p(\mathsf{A}_k|\mathsf{A}_{k-1})$ since $\mathsf{A}_{k-1}$ shields $\mathsf{A}_k$ from variables in $\mathsf{A}_{\leq k-2}$. 
In the two-dimensional lattice, the factorization is in terms of distributions $p(\mathsf{A}_k|\partial\mathsf{A}_k \cap \mathsf{A}_{<k})$ since we now need $\partial\mathsf{A}_k \cap \mathsf{A}_{<k}$ to shield $\mathsf{A}_k$ from the previous subsets and be able to to apply the conditional independence property in \defref{ApproximateConditionalIndependence}.
Since the number of regions in $\partial\mathsf{A}_k \cap \mathsf{A}_{<k}$ is a constant, the remainder of the proof is analogous to the one-dimensional case. 
We note that the ordering in which the distribution $p(x_1,\dots, x_n)$ is factorized is the `reverse' of the construction in \thmref{factorizaiton}.
The regions $\partial\mathsf{A}_k \cap \mathsf{A}_{<k}$ have a two-dimensional geometry. 
However, conventional 1D RNNs, as shown in \fig{FNNRNN}(b), map 2D regions to 1D configurations, often disrupting spatial locality and requiring $\poly(n)$-sized hidden memory states to capture the resulting non-local information. 
To better align the RNN architecture with 2D geometry, the authors in \cite{Hibat2020RNN} proposes a design where hidden memory states propagate vertically as well as horizontally (see their Fig. 5(b)). 
In other words, instead of following \eqref{eq:hidden1}, the updated memory state $h^{(k)}_{ij}$ in the RNN cell at row $i$ and column $j$ within the $k$'th region $\mathsf{A}_k$ now satisfies
\begin{align}
    h_{ij}^{(k)} = f_h(x_{i-1,j}^{(k)}, h_{i-1,j}^{(k)}, x_{i,j-1}^{(k)}, h_{i,j-1}^{(k)}).\nonumber
\end{align}
Extensions of this architecture to higher dimensions follow naturally. These generalized RNN wavefunctions allow the computation of conditional probabilities $p(\mathsf{A}_k|\partial\mathsf{A}_k \cap \mathsf{A}_{<k})$ using $\mathcal{O}(\log^{D_{\Lambda}}(n))$ dimensional hidden memory.

\begin{figure}
    \centering
    \includegraphics[width=0.20\textwidth]{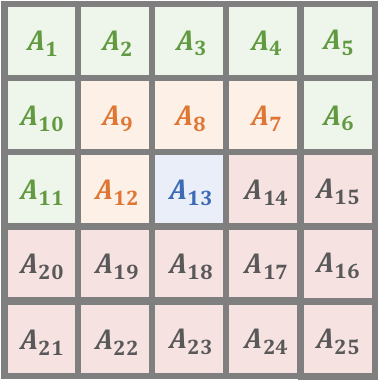}
    \caption{In the construction of the recurrent neural network, the lattice is divided into subsets $\{A_i\}_{i=1}^M$. 
    The chain rule is then applied to factorize the $p(\mathsf{A}_1,\dots, \mathsf{A}_M)$ into a product of conditional distributions $p(A_k|A_{<k} \cap \partial A_k)$ in the `snake' order shown.
    As an example, the case of $k=13$ with the  conditional probability $p(\mathcolor{blue}{\mathsf{A}_{13}}|\mathcolor{orange}{\mathsf{A}_{7}}, \mathcolor{orange}{\mathsf{A}_{8}}, \mathcolor{orange}{\mathsf{A}_{9}}, \mathcolor{orange}{\mathsf{A}_{12}})$ is highlighted.}
    \label{fig:RNNDecomposition}
\end{figure}

\section{From conditional independence to area law}\label{sec:areaLaw}

Here, we establish \thmref{signfreeArealaw}, which shows that approximate conditional independence in non-negative quantum states leads to an entanglement area law for such states.
\begin{theorem}[Approximate conditional independence implies area law]
    Consider an $n$-qudit sign-free state $\ket{\psi}\in \bbC^d\otimes \cdots \otimes \bbC^d$ on a lattice. Suppose the measurement distribution of this state satisfies the approximate conditional independence property \eqref{eq:ConditionalIndependence} for a given tripartition  $\mathsf{A} \cup \mathsf{B} \cup \mathsf{C} = [n]$ of the lattice such that $\mathsf{A} \cap \mathsf{C} = \emptyset$.
    The von Neumann entropy of the reduced state $\rho_{\mathsf{A}}$ in region $\mathsf{A}$ is upper bounded by $$S(\rho_{\mathsf{A}})\leq \mathcal{O}(|\mathsf{B}|),$$
    which up to $\log(|\mathsf{A}|)$ factors implies an area law $S(\rho_{\mathsf{A}})\leq \mathcal{O}(|\partial \mathsf{A}|\cdot \log(|\partial\mathsf{A}|)$.
\end{theorem}

\begin{proof}
   Consider the following state $\ket{\phi}$ defined by
\begin{align}
    \ket{\phi}:=\sum_{x_{\mathsf{A}}, x_{\mathsf{B}}, x_{\mathsf{C}}}\sqrt{p(x_{\mathsf{B}}) \cdot p(x_{\mathsf{A}}|x_{\mathsf{B}})\cdot p(x_{\mathsf{C}}|x_{\mathsf{B}})}\cdot \ket{x_{\mathsf{A}},x_{\mathsf{B}},x_{\mathsf{C}}},
\end{align}
where the sum is over $x_{\mathsf{A}} \in [d]^{|\mathsf{A}|}$, $x_{\mathsf{B}} \in[d]^{|\mathsf{B}|}$, and $x_{\mathsf{C}} \in [d]^{|\mathsf{C}|}$. If the ground state $\ket{\psi}$ satisfies the conditional independence \emph{exactly}, then it would be equal to the state $\ket{\phi}$. The proof of area law includes two steps:
\begin{enumerate}
    \item First, we show that the entropies of the reduced states $$\rho_{\mathsf{A}}=\Tr_{\mathsf{BC}}(\ketbra{\psi}{\psi})\quad \text{and}\quad \sigma_{\mathsf{A}}=\Tr_{\mathsf{BC}}(\ketbra{\phi}{\phi})$$ are close to each other.
    \item We then prove that the entropy of the reduced state $\sigma_{\mathsf{A}}$ is bounded by $S(\sigma_{\mathsf{A}})\leq \mathcal{O}(|\mathsf{B}|)$.
\end{enumerate}
By combining these two steps, we arrive at the area law for the entanglement entropy $S(\rho_{\mathsf{A}})$.

\paragraph{Step 1:} We begin by bounding the  overlap $\braket{\phi}{\psi}$ 

\begin{align}
    \braket{\phi}{\psi}&= \E_{\bm{x}_{\mathsf{B}}}\left[\sum_{x_{\mathsf{A}},x_\mathsf{C}}  \sqrt{p(x_{\mathsf{A}},x_{\mathsf{C}}|\bm{x}_{\mathsf{B}})}\cdot \sqrt{p(x_A|\bm{x}_B)\cdot p(x_C|\bm{x}_B)}\right]\nonumber\\
    &=\E_{\bm{x}_{\mathsf{B}}} F^{1/2}\left(p_{{\mathsf{A}},{\mathsf{C}}|\bm{x}_{\mathsf{B}}},p_{{\mathsf{A}}|\bm{x}_\mathsf{B}}\otimes p_{\mathsf{C}|\bm{x}_\mathsf{B}}\right) && \text{(by definition of fidelity between distributions)}\nonumber \\
    &\geq 1-\frac{1}{2}\cdot \E_{\bm{x}_{\mathsf{B}}} \norm{p_{{\mathsf{A}},{\mathsf{C}}|\bm{x}_{\mathsf{B}}}-p_{{\mathsf{A}}|\bm{x}_{\mathsf{B}}}\otimes p_{{\mathsf{C}}|\bm{x}_{\mathsf{B}}}}_1 && \text{(using $F^{1/2}(p,q)\geq 1-1/2\cdot\norm{p-q}_1$)}\nonumber\\
    &\geq 1-\E_{\bm{x}_{\mathsf{B}}} \sqrt{\frac{1}{2} D_{\mathrm{KL}}(p_{{\mathsf{A}}, {\mathsf{C}}|\bm{x}_{\mathsf{B}}}, p_{{\mathsf{A}}|\bm{x}_{\mathsf{B}}}\otimes p_{{\mathsf{C}}|\bm{x}_{\mathsf{B}}})} &&\text{(Pinsker's inequality)} \nonumber\\
    &\geq 1-\sqrt{ \frac{1}{2}\E_{\bm{x}_B} D_{\mathrm{KL}}(p_{{\mathsf{A}},{\mathsf{C}}|\bm{x}_{\mathsf{B}}},p_{X_{\mathsf{A}}|\bm{x}_{\mathsf{B}}}\otimes p_{{\mathsf{C}}|\bm{x}_{\mathsf{B}}})} &&\text{(Jensen's inequality)}\nonumber\\
    &=1-\sqrt{\frac{1}{2} I(\mathsf{A},\mathsf{C}|\mathsf{B})} &&\text{(definition of CMI)} \nonumber\\
    &\geq 1-e^{-\Omega(\dist(\mathsf{A},\mathsf{C}))}.
\end{align}
We can now turn this lower bound on the overlap $\norm{\rho_{\mathsf{A}}-\sigma_{\mathsf{A}}}_1$ to an upper bound on the $\ell_1$-distance between the reduced states $\rho_{\mathsf{A}}=\Tr_{{\mathsf{BC}}}(\ketbra{\psi}{\psi})$ and $\sigma_{\mathsf{A}}=\Tr_{{\mathsf{BC}}}(\ketbra{\phi}{\phi})$:
\begin{equation}
    \norm{\rho_{\mathsf{A}}-\sigma_{\mathsf{A}}}_1 \leq \norm{\ketbra{\psi}{\psi}-\ketbra{\phi}{\phi}}_1 \leq 2\cdot \sqrt{1-|\braket{\phi}{\psi}|^2}\leq e^{-\Omega(\dist({\mathsf{A}},{\mathsf{C}}))}.\label{eq:DistanceOfRhoSigma}
\end{equation}
It follows from the Fannes' inequality that when $\norm{\rho_{\mathsf{A}}-\sigma_{\mathsf{A}}}_1\leq 1/e$, the difference between the entropies $|S(\rho_{\mathsf{A}})-S(\sigma_{\mathsf{A}})|$ is bounded by
\begin{equation}
    |S(\rho_{\mathsf{A}})-S(\sigma_{\mathsf{A}})|\leq \log(d^{|{\mathsf{A}}|}) \cdot \norm{\rho_{\mathsf{A}}-\sigma_{\mathsf{A}}}_1-\frac{1}{\ln(2)} \norm{\rho_{\mathsf{A}}-\sigma_{\mathsf{A}}}_1 \cdot \log\norm{\rho_{\mathsf{A}}-\sigma_{\mathsf{A}}}_1.\label{eq:fannes}
\end{equation}
Let $\dist({\mathsf{A}},{\mathsf{C}})=\mathcal{O}\left(\log(|{\mathsf{A}}|\cdot\log d)\right)$, then by combining the two bounds \eqref{eq:DistanceOfRhoSigma} and \eqref{eq:fannes}, we get
\begin{equation}
    |S(\rho_{\mathsf{A}})-S(\sigma_{\mathsf{A}})| \leq \frac{1}{\poly(|{\mathsf{A}}|)}\label{eq:boundDifferenceOfEntropies}.
\end{equation}

\paragraph{Step 2:} Next, we upper bound the entropy $S(\sigma_{\mathsf{A}})$ which entails finding the following expression for the reduced state $\sigma_{\mathsf{A}}$:
\begin{align}
    \sigma_{\mathsf{A}}&=\E_{\bm{x}_{\mathsf{B}}} \ketbra{\alpha_{\mathsf{A}}^{\bm{x}_{\mathsf{B}}}}{\alpha_{\mathsf{A}}^{\bm{x}_{\mathsf{B}}}},
\end{align}
in terms of the pure states 
\begin{align}
\ket{\alpha_{\mathsf{A}}^{\bm{x}_{\mathsf{B}}}}:=\sum_{x_{\mathsf{A}}}\sqrt{p(x_{\mathsf{A}}|\bm{x}_{\mathsf{B}})}\cdot \ket{x_{\mathsf{A}}}.
\end{align}
This expression for $\sigma_{\mathsf{A}}$ can be derived by a direct calculation:
\begin{align}
    \sigma_{\mathsf{A}}&=\sum_{x_{\mathsf{A}},y_{\mathsf{A}},x_{\mathsf{C}}} \E_{\bm{x}_{\mathsf{B}}}\left[\sqrt{p(x_{\mathsf{A}}|\bm{x}_{\mathsf{B}})\cdot p(x_{\mathsf{C}}|\bm{x}_{\mathsf{B}})}\cdot \sqrt{p(y_{\mathsf{A}}|\bm{x}_{\mathsf{B}})\cdot p(x_{\mathsf{C}}|\bm{x}_{\mathsf{B}})}\right]\cdot \ketbra{x_{\mathsf{A}}}{y_{\mathsf{A}}}\nonumber\\
    &=\sum_{x_{\mathsf{A}},y_{\mathsf{A}}} \E_{\bm{x}_{\mathsf{B}}}\left[\sqrt{p(x_{\mathsf{A}}|\bm{x}_{\mathsf{B}})}\cdot \sqrt{p(y_{\mathsf{A}}|\bm{x}_{\mathsf{B}})}\cdot \sum_{x_{\mathsf{C}}} p(x_{\mathsf{C}}|\bm{x}_{\mathsf{B}})\right]\cdot \ketbra{x_{\mathsf{A}}}{y_{\mathsf{A}}}\nonumber\\
    &=\E_{\bm{x}_{\mathsf{B}}} \ketbra{\alpha_{\mathsf{A}}^{\bm{x}_{\mathsf{B}}}}{\alpha_{\mathsf{A}}^{\bm{x}_{\mathsf{B}}}},
\end{align}
The von Neumann entropy of this mixture $\sigma_{\mathsf{A}}$ can be bounded by 
\begin{align}
    S(\sigma_{\mathsf{A}})&=S\left(\E_{\bm{x}_{\mathsf{B}}} \ketbra{\alpha_{\mathsf{A}}^{\bm{x}_{\mathsf{B}}}}{\alpha_{\mathsf{A}}^{\bm{x}_{\mathsf{B}}}}\right)\nonumber\\
    &\leq H(\mathsf{B})+\E_{\bm{x}_{\mathsf{B}}} S(\ketbra{\alpha_{\mathsf{A}}^{\bm{x}_{\mathsf{B}}}}{\alpha_{\mathsf{A}}^{\bm{x}_{\mathsf{B}}}}) \nonumber\\
    &\leq |{\mathsf{B}}| \cdot \log d.
\end{align}
This bound on $S(\sigma_{\mathsf{A}})$ along with the choice of $\dist({\mathsf{A}},{\mathsf{C}})=\mathcal{O}\left(\log(|{\mathsf{A}}|\cdot \log d)\right)$ that leads to the bound in  \eqref{eq:boundDifferenceOfEntropies} imply 
\begin{align}
    S(\rho_{\mathsf{A}})&\leq S(\sigma_{\mathsf{A}})+\frac{1}{\poly(|{\mathsf{A}}|)} \nonumber\\
    &\leq \mathcal{O}\left(|\partial {\mathsf{A}}|\cdot \log d \cdot \dist({\mathsf{A}},{\mathsf{C}})\right)\nonumber\\
    &\leq \mathcal{O}\left(|\partial {\mathsf{A}}| \cdot \log d \cdot \log \left(|{\mathsf{A}}| \cdot \log d\right)\right).
\end{align}
\end{proof}

\section{Conditional independence in one-dimensional systems}\label{sec:CMIin1D}

Here, we establish several results concerning the decay of conditional mutual information for one-dimensional gapped Hamiltonians.
We begin by formally defining entanglement swapping in \secref{swappingprelim}. 
Next, we prove the optimality of the Bell basis for entanglement swapping with qubits, as presented in \secref{OptimalEntanglementSwapping}, and extend this result to the rotated cluster state in \secref{CMIChainQubits}. 
Finally, we analyze shallow random quantum circuits in 1D, demonstrating that they exhibit short-range conditional correlations in \secref{CMIShallowCkt}.

\subsection{Preliminaries on measures of entanglement}

\begin{definition}[Concurrence of two-qubit states]\label{def:Concurrence}
Given a pure two-qubit state $\ket{\phi}_{\mathsf{AB}} = a \ket{0 0} + b \ket{0 1} + c \ket{1 0} + d \ket{1 1}$, we define its concurrence $C(\phi_{\mathsf{AB}})$ by 
\begin{align}
    C(\phi_{\mathsf{AB}}) = 2\cdot |a d - b c|.
\end{align}
We can equivalently express the concurrence in terms of the reduced state $\rho_{\mathsf{A}}$ on subsystem $\mathsf{A}$ as 
\begin{align}
    C(\phi_{\mathsf{AB}}) = \sqrt{2\cdot \left(\Tr(\rho_{\mathsf{A}})^2 - \Tr(\rho_{\mathsf{A}}^2)\right)}.\label{eq:concurrenceExpression}
\end{align}
More generally, the concurrence of a two-qubit mixed state $\rho$ is defined by
    \begin{align}
        C(\rho_{\mathsf{AB}}) = \inf \sum_{i}p_i \cdot C\left((\phi_i)_{\mathsf{AB}}\right)\label{eq:concurrenceMixed}
    \end{align}
    where the infimum is over all mixtures of pures states $\{p_i, (\phi_i)_{\mathsf{AB}}\}$ such that $\rho_{\mathsf{AB}} = \sum_{i} p_i \cdot (\phi_i)_{\mathsf{AB}}$.
\end{definition}
\begin{fact}[Closed form formula for concurrence, cf. \cite{Woo98Formation}]\label{fact:concurrenceFormula}
    Given a two-qubit state $\rho$, define 
    \begin{align}
        \widetilde{\rho} = (Y \otimes Y) \cdot \rho^* \cdot (Y \otimes Y)
    \end{align}
    where $\rho^*$ is the complex conjugate of $\rho$ in the computational basis. 
    Let $\alpha_1 \geq \alpha_2 \geq \alpha_3 \geq\alpha_4$ be the eigenvalues of the Hermitian operator $\sqrt{\sqrt{\rho} \cdot \widetilde{\rho}\cdot \sqrt{\rho}}$.
    The concurrence of $\rho$ satisfies 
    \begin{align}
        C(\rho_{\mathsf{AB}}) = \max\{0, \alpha_1 - \alpha_2 - \alpha_3 - \alpha_4\}.
    \end{align}
\end{fact}

\begin{fact}[Entanglement entropy in terms of concurrence, cf. \cite{Wootters97PairQubits}]\label{fact:EntropyConcurrence}
    Given a two-qubit state $\ket{\phi}_{\mathsf{AB}}$, let $\phi_{\mathsf{A}}$ be the reduced state on subsystem $\mathsf{A}$.
    The entanglement entropy $S(\phi_{\mathsf{A}})$ admits the following expression in terms of concurrence $C(\phi_{\mathsf{AB}})$:
    \begin{align}
        S(\phi_{\mathsf{A}}) = h\left( \frac{1 + \sqrt{1- C(\phi_{\mathsf{AB}})^2}}{2}\right).
    \end{align}
    Here $h(x)$ is the binary Shannon entropy defined by
    \begin{align}
        h(x) = - x \log(x) - (1-x)\log(1-x).\label{eq:binaryEntropy}
    \end{align} 
\end{fact}

\begin{definition}[Entanglement of formation]\label{def:EntanglementFormation}
    The Entanglement of formation for a state $\rho_{\mathsf{AB}}$ is defined~by
    \begin{align}
        E_F(\rho_{\mathsf{AB}}) = \inf \sum_i p_i \cdot S\left((\phi_i)_{\mathsf{A}}\right).
    \end{align}
    The infimum in this expression is over all mixtures of pure states $\{p_i , (\phi_i)_{\mathsf{AB}}\}$ such that $\rho_{\mathsf{AB}} = \sum_i p_i \cdot (\phi_i)_{\mathsf{AB}}$.
\end{definition}
\begin{fact}[Entanglement of formation in terms of concurrence, cf. \cite{Woo98Formation}]\label{fact:EntanglementFormation}
    The entanglement of formation of a two-qubit state $\rho_{\mathsf{AB}}$ satisfies
    \begin{align}
        E_{F}(\rho_{\mathsf{AB}}) = h\left( \frac{1 + \sqrt{1- C(\rho_{\mathsf{AB}})^2}}{2}\right).
    \end{align}
    where $h(\cdot)$ is the binary Shannon entropy in \eqref{eq:binaryEntropy} and the concurrence $C(\rho_{\mathsf{AB}})$ is defined by the expression~\eqref{eq:concurrenceMixed} in \defref{Concurrence}.
\end{fact}
\subsection{Entanglement swapping preliminaries}\label{sec:swappingprelim}

\begin{definition}[Bell basis]\label{def:bell-basis}
Let $d \geq 2$ be an integer,
and write $\omega = e^{2 \pi i / d}$.
The \emph{generalized Bell basis} is the set which contains a vector $\Bell{a}{b}$ in $(\mathbb{C}^d)^{\ot 2}$ for each $a, b \in \{0, \ldots, d-1\}$ defined as
\begin{equation*}
\Bell{a}{b} = \frac{1}{\sqrt{d}} \cdot \sum_{x = 0}^{d-1} \omega^{a \cdot x} \ket{x, x+b},
\end{equation*}
where addition is modulo~$d$.
\end{definition}

By Definition~\ref{def:bell-basis},
the~$d$ Schmidt coefficients of $\Bell{a}{b}$ are all~$1/\sqrt{d}$.
Hence, all of the Bell basis states are maximally entangled.
A particular example of such states is the EPR state in $\mathbb{C}^d$ given by
\begin{equation*}
\epr{d} := \Bell{0}{0} = \frac{1}{\sqrt{d}} \cdot \sum_{x=0}^{d-1} \ket{x,x}.
\end{equation*}
Before proceeding to introducing the entanglement swapping setup, we state and prove the following proposition which gives a decomposition of the four-party state $\ket{\psi} = \ket{a}_{\mathsf{AB}}\otimes \ket{b}_{\mathsf{CD}}$  in terms of the Bell basis.
\begin{proposition}[General state in Bell basis]\label{prop:general-state-in-bell-basis}
Consider the following two-qubit states $\ket{a}_{\mathsf{AB}}$ and $\ket{b}_{\mathsf{CD}}$:
\begin{equation}\label{eq:general-psi}
\ket{a}_{\mathsf{AB}} \otimes \ket{b}_{\mathsf{CD}}
= \Big(\sum_{i=0}^{d-1} a_i \cdot \ket{i, i}_{\mathsf{AB}}\Big)
	\ot \Big(\sum_{j=0}^{d-1} b_j \cdot \ket{j, j}_{\mathsf{CD}}\Big).
\end{equation}
Then, we have
\begin{equation*}
\ket{a}_{\mathsf{AB}}\otimes \ket{b}_{\mathsf{CD}}
= \frac{1}{\sqrt{d}} \cdot \sum_{i, j=0}^{d-1} \Bell{i}{j}_{\mathsf{BC}}
	\ot \Big(\sum_{k=0}^{d-1} \omega^{-i \cdot k}\cdot a_k\cdot  b_{k + j} \cdot \ket{k, k+j} \Big)_{\mathsf{AD}}.
\end{equation*}
In particular, it holds that
\begin{equation}
\epr{d}_{\mathsf{A}\mathsf{B}}
\otimes \epr{d}_{\mathsf{C}\mathsf{D}}
= \frac{1}{d}\cdot \sum_{a,b=0}^{d-1} \Bell{a}{b}_{\mathsf{A}\mathsf{D}}
	\otimes \Bell{-a}{b}_{\mathsf{B}\mathsf{C}}.\label{eq:EPRinBellBasis}
\end{equation}
\end{proposition}
\begin{proof}
We begin by expanding the right-hand side.
\begin{align*}
&\frac{1}{\sqrt{d}} \cdot \sum_{i, j} \Bell{i}{j}_{\mathsf{BC}}
	\ot \Big(\sum_{k} \omega^{-i \cdot k}\cdot a_k\cdot  b_{k + j} \cdot \ket{k, k+j} \Big)_{\mathsf{AD}}\\
={}&\frac{1}{\sqrt{d}} \cdot \sum_{i, j} \Big(\frac{1}{\sqrt{d}} \cdot \sum_{\ell=0}^{d-1} \omega^{i \cdot \ell} \ket{\ell, \ell+j}\Big)
	\ot \Big(\sum_{k} \omega^{-i \cdot k}\cdot a_k\cdot  b_{k + j} \cdot \ket{k, k+j} \Big)\\
={}&\frac{1}{d} \cdot \sum_{k, \ell} \sum_{i, j}  \omega^{i \cdot (\ell - k)} \cdot a_k \cdot b_{k+j}\cdot
\ket{\ell, \ell+j}\otimes \ket{k, k+j} \\
={}& \sum_{k} \sum_{j}   a_k \cdot b_{k+j}\cdot  \ket{k, k+j}\otimes \ket{k, k+j} \\
={}& \sum_{k} \sum_{\ell}   a_k \cdot b_{\ell}\cdot  \ket{k, \ell}_{\mathsf{AC}}\otimes \ket{k, \ell}_{\mathsf{BD}} \\
={}& \ket{a}_{\mathsf{AB}}\otimes \ket{b}_{\mathsf{CD}}.
\end{align*}
This completes the proof.
\end{proof}
To understand equation \eqref{eq:EPRinBellBasis} intuitively, suppose we are given $\epr{d}_{\mathsf{A}\mathsf{B}}
\otimes \epr{d}_{\mathsf{C}\mathsf{D}}$
and we measure subsystems $\mathsf{B}$ and~$\mathsf{C}$
in the Bell basis.
By \propref{general-state-in-bell-basis},
if we receive outcome $\Bell{a}{b}$,
then the state in subsystems~$\mathsf{A}$ and~$\mathsf{D}$
collapses to $\Bell{-a}{b}$.
Thus, we have exchanged a pair of maximally entangled states between subsystems~$\mathsf{A}$ and~$\mathsf{B}$
and between~$\mathsf{C}$ and~$\mathsf{D}$
with a pair between~$\mathsf{A}$ and~$\mathsf{D}$
and~$\mathsf{B}$ and~$\mathsf{C}$.
This is the protocol known as \emph{entanglement swapping}.

In this work, we consider two means of generalizing the entanglement swapping protocol.
First, we allow the starting state to be any bipartite state of the form $\ket{a}_{\mathsf{AB}} \otimes \ket{b}_{\mathsf{CD}}$,
where $\ket{a}_{\mathsf{AB}}, \ket{b}_{\mathsf{CD}} \in (\mathbb{C}^{d})^2$,
rather than just a tensor product of two EPR states.
Second, we allow the measurement of the~$\mathsf{B}$ and~$\mathsf{C}$ subsystems
to be any basis measurement $\ket{u_1}, \ldots, \ket{u_{d^2}}$ in $(\mathbb{C}^d)^{\otimes 2}$,
rather than just the Bell basis measurement.
More precisely, we consider the following protocol:

\begin{definition}[Entanglement swapping]\label{def:entanglementSwapping}
A general entanglement swapping protocol refers to the following procedure
\begin{enumerate}
    \item \emph{Input:}  $\ket{\psi} = \ket{a}_{\mathsf{AB}} \otimes \ket{b}_{\mathsf{CD}}$ as in equation~\eqref{eq:general-psi}, and the unitary $U = \sum_{i\in[d^2]} \ketbra{u_i}{i}$ which specifies an orthonormal basis of $(\mathbb{C}^d)^{\otimes 2}$.
    \item \emph{Procedure}: measure subsystems $\mathsf{BC}$ of~$\ket{\psi}$ in the $U$ basis.
    Let $p_i$ be the probability of observing outcome~$i$, and~$\rho_i = (\rho_i)_{\mathsf{A}}$ be the reduced density matrix on subsystems~$\mathsf{A}$ given this outcome.
    \item \emph{Output}: the pair $(p, \rho)$, where $p = (p_1, \ldots, p_{d^2})$ and $\rho = (\rho_1, \ldots, \rho_{d^2})$.
    We write $\boldsymbol{\rho}$ for the random variable which is equal to~$\rho_i$ with probability~$p_i$.
\end{enumerate}
\end{definition}

\subsection{Reduction from shallow circuits to entanglement swapping}\label{sec:ReductionShallow}

In this section, we focus on ground states that can be prepared with one-dimensional shallow quantum circuits. 
We will see that for such states, the study of the measurement-induced correlations and conditional independence reduces to a question about the performance of successive applications of the entanglement swapping protocol defined in \defref{entanglementSwapping}.

We primarily consider the brickwork architecture of 1D quantum circuits, although our results apply more broadly to other geometrically-local circuits. 
In the brickwork architecture, as shown in \fig{lightconesShallow}, the quantum circuit is constructed using alternating layers of two-quit gates. 
The odd-numbered layers consist of two-qudit gates acting on pairs of qudits $(2i, 2i+1)$, while the even-numbered layers include gates acting on pairs of qudits $(2i-1, 2i)$, 
where $i$ ranges from $1$~to~$\lfloor \frac{n-1}{2} \rfloor$.

Given a circuit of depth $D$, we decompose the qudits and the gates as follows:
Assume the number of qudits $n = 2n'(D-1)$ for some integer $n'\geq 1$.
Divide the set of qudits into contiguous subsets of size $D-1$.
We denote the odd subsets by $\mathsf{L}_i$ and the even subsets by $\mathsf{R}_i$ for $i\in \{1,\dots,n'\}$.
The gates are partitioned into backward lightcones $V_1\otimes \cdots \otimes V_{n'}$ and forward lightcones $U_1 \otimes \cdots \otimes U_{n'-1}$ with \emph{triangular} shapes as shown in \fig{lightconesShallow}.
The two-qudit gates in $V_i$ act on the qudits in subsystems $\mathsf{L}_i$ and $\mathsf{R}_i$,  while the gates in $U_i$ act on the qudits in subsystems $\mathsf{R}_i$ and $\mathsf{L}_{i+1}$. 
Finally, we define the entangled state $\ket{a}_{\mathsf{L}_i\mathsf{R}_i} = V_i \cdot \ket{0}_{\mathsf{L}_i\mathsf{}_i}$ for any pair of subsets $(\mathsf{L}_i, \mathsf{R}_i)$.

Using this decomposition, we can now easily interpret the distribution obtained by measuring the output state of this one-dimensional shallow circuit using the entanglement swapping setup of \defref{entanglementSwapping}.
Initially multiple pairs of entangled states $\otimes_{i=1}^{n'} \ket{a}_{\mathsf{L}_i,\mathsf{R}_i}$ are prepared. 
Subsequently, each pair of registers $(\mathsf{R}_i,\mathsf{L}_{i+1})$ is measured in the orthogonal basis $\{\ket{u^i_j}_{\mathsf{R}_i\mathsf{L_{i+1}}}\}$ specified by the unitary $U_i = \sum_{j=1}^{d^{2(D-1)}} \ketbra{u^i_j}{j}$, starting from $i = 1$ and proceeding to the end of the circuit at~$i = n'-1$. 
Hence, to understand the nature of correlations that exist in the measurement distribution of a one-dimensional shallow circuit, we can study the states produced as the result of the corresponding entanglement swapping setup.

This reduction applies more generally to any state $\ket{\psi}$ with the \emph{zero-correlation length} property.
That is, if we can partition the qudits into contiguous groups $\mathsf{T}_1,\dots,\mathsf{T}_{n'}$  of size $\Theta(D)$, for some constant $D$, such that 
\begin{align}
    \rho_{\mathsf{T}_i, \mathsf{T}_j} = \rho_{\mathsf{T}_i}\otimes \rho_{\mathsf{T}_j} \text{ for any } |i-j|>1.\nonumber
\end{align}
We will later consider an example of entanglement swapping in such states with a zero correlation length in the context of one-dimensional sign-free states \cite{hastings2016signfree}.
Using Uhlmann's theorem, we see that there is an isometry $U_i$ acting on each group of qudits $\mathsf{T}_i$ that maps the Hilbert space into a tensor product of the form $\mathsf{R}_i \otimes \mathsf{L}_{i+1}$ such that the resulting state is not entangled across subsystems $\mathsf{R}_i$ and $\mathsf{L}_{i+1}$.
By repeatedly applying such isometries $U_i$ across the one-dimensional chain, we obtain a state of the form $\otimes_i \ket{a}_{\mathsf{L}_i\mathsf{R}_i}$ similar to the decomposition we obtained earlier for the shallow quantum circuit.
Depending on the Schmidt rank of the reduced state $\rho_{\mathsf{T}_1,\dots, \mathsf{T}_{i}}$, the resulting tensor product space $\mathsf{R}_i \otimes \mathsf{L}_{i+1}$ may have a smaller dimension than the original space and in general lacks the geometrically-local structure we observe in the reduction starting from shallow circuits.

Having obtained this reduction to the entanglement swapping protocol, we proceed with our objective of analyzing the correlations that arise from measuring interacting quantum systems on a one-dimensional chain.  
We consider contiguous subsets $(\mathsf{A},\mathsf{B}, \mathsf{C})$ where $\mathsf{A}$ and $\mathsf{C}$ correspond respectively to $\mathsf{L}_i$ and $\mathsf{A}_j$ for some $1\leq i<j\leq n'$, and $\mathsf{B}$ corresponds to the qudits between $\mathsf{L}_i$ and~$\mathsf{R}_j$.

Given the zero-correlation length property, when the remaining qubits in $[n]\setminus \mathsf{A}\cup \mathsf{B} \cup \mathsf{C}$ are measured in the computational basis, only the state of subsystems $(\mathsf{L}_i, \mathsf{R}_i)$ and $(\mathsf{L}_j, \mathsf{R}_j)$ changes (from $\ket{a}_{\mathsf{L}_i\mathsf{R}_i}$ and $\ket{a}_{\mathsf{L}_j\mathsf{R}_j}$), leaving the state of the remaining qudits in subset $\mathsf{B}\setminus \mathsf{R}_i \cup \mathsf{L}_j$ undisturbed. 
Therefore, in future sections, we often analyze the correlations between $\mathsf{L}_1$ and $\mathsf{R}_{n'}$ at the two ends of the chain, with the understanding that a similar analysis applies to subsystems $\mathsf{L}_i$ and $\mathsf{R}_j$ for $1<i<j<n'$ when the states $\ket{a}_{\mathsf{L}_i\mathsf{R_i}}$ and $\ket{a}_{\mathsf{L}_j\mathsf{R}_j}$ are possibly changed.

We study the correlations induced by entanglement swapping and in particular the notion of conditional independence in \defref{ApproximateConditionalIndependence} using the \emph{classical} conditional mutual information $I(\mathsf{A}:\mathsf{C}|\mathsf{B})$ in the measurement distribution. 
We find it more convenient to analyze an upper bound on $I(\mathsf{A}:\mathsf{C}|\mathsf{B})$ given in terms of the average von-Neumann entropy of the subsystem $\mathsf{A}$.
It follows \cite{meurice2024experimental} from Holevo's theorem that
\begin{align}
    I(\mathsf{A}:\mathsf{C}|\mathsf{B}) &:= \E_{\bm{x}_{\mathsf{B}}}D_{\text{KL}}(p_{{\mathsf{A}},{\mathsf{C}}|\bm{x}_{\mathsf{B}}} \parallel p_{{\mathsf{A}}|\bm{x}_{\mathsf{B}}}\otimes p_{{\mathsf{B}}|\bm{x}_{\mathsf{B}}})\nonumber\\
    & \leq \E_{\bm{x}_{\mathsf{B}}} S(\rho_{\mathsf{A}|{\bm{x}_{\mathsf{B}}}}) \label{eq:CMIfromES}
\end{align}
Therefore, to prove the conditional independence property \eqref{eq:ConditionalIndependence}, it suffices to find an upper bound on the average entropy of the post-measurement state  $\E_{\bm{x}_{\mathsf{B}}} S(\rho_{\mathsf{A}|{\bm{x}_{\mathsf{B}}}})$.

\subsection{Optimal entanglement swapping for qubits}\label{sec:OptimalEntanglementSwapping}

We start with when entanglement swapping measurement is performed on qubits.
Our main result is a tight upper bound on the expected entropy $\E S(\bm{\rho}_{\mathsf{A}})$ of the state $\bm{\rho}_{\mathsf{A}}$ generated in this~process.

Prior to stating this result, we gather some definitions and notations. 
We let $\mathcal{H}_{\mathsf{ABCD}}$ denote a Hilbert space
containing four subsystems $\mathsf{A}$, $\mathsf{B}$, $\mathsf{C}$, and $\mathsf{D}$
of dimension $d = 2$ each.
Let $\ket{u_1}, \ket{u_2}, \ket{u_3}, \ket{u_4}$ be an orthonormal basis of $\mathcal{H}_{\mathsf{BC}}$,
where
\begin{equation*}
\ket{u_i}_{\mathsf{BC}} = u_{00}^i \ket{00} + u_{01}^i \ket{01} + u_{10}^i \ket{10} + u_{11}^i \ket{11}.
\end{equation*}
The (unnormalized) post-measurement state
after receiving outcome~$i$ is given by
\begin{align}
\ket{\psi_i}_{\mathsf{AD}}
&= \bra{u_i}_{\mathsf{BC}}
			\cdot \ket{a}_{\mathsf{AB}} \otimes \ket{b}_{\mathsf{CD}}\nonumber\\
&= \Bigg(\sum_{x \in \{0, 1\}^2} (u^i_{x})^\dagger \bra{x}_{\mathsf{BC}}\Bigg)
	\cdot \Bigg(\sum_{y \in \{0, 1\}^2} a_{y_1} b_{y_2} \ket{y_1 y_1}_{\mathsf{AB}} \ket{y_2 y_2}_{\mathsf{CD}}\Bigg)
= \sum_{x \in \{0, 1\}^2} (u_{x}^i)^\dagger a_{x_1} b_{x_2}\cdot \ket{x}. \label{eq:post-measurement-unnormalized}
\end{align}
The probability of receiving outcome~$i$ is
\begin{equation}
p_i = \braket{\psi_i}{\psi_i} = \sum_{x \in \{0, 1\}^2} |u_x^i|^2 \cdot |a_{x_1}|^2 |b_{x_2}|^2.\label{eq:MeasurementProbability}
\end{equation}
Hence, the normalized post-measurement state on subsystem~$\mathsf{A}$ is
\begin{equation}
\rho^i_{\mathsf{A}} = \frac{1}{p_i} \cdot (\ket{\psi_i}\bra{\psi_i})_{\mathsf{A}}.\label{eq:postMeasurementState_2}
\end{equation}
The main result of this section is stated as follows. 
\begin{theorem}[Upper bound on binary entanglement swapping]\label{thm:BinaryEntanglementSwapping}
Following an entanglement swapping measurement in an arbitrary basis $U$ on the state $\ket{\psi}$ defined by
\begin{equation}
\ket{a}_{\mathsf{AB}} \otimes \ket{b}_{\mathsf{CD}}
= \Big(a_0 \ket{00} + a_1 \ket{11}\Big)
	\ot \Big(b_0 \ket{00} + b_1 \ket{11}\Big)\label{eq:initialStateQubits},
\end{equation}
the average entropy $S(\bm{\rho}_{\mathsf{A}})$ of the output state is bounded by 
\begin{align}
    \E S(\brho_{\mathsf{A}}) \leq  p_0 \cdot S(\rho_0) + p_1 \cdot S(\rho_1) \label{eq:EntropyBoundSwapping}
\end{align}
where $p_0 = |a_0|^2 |b_0|^2 + |a_1|^2 |b_1|^2$, $p_1 =|a_0|^2 |b_1|^2 + |a_1|^2 |b_0|^2$, and we have
\begin{align*}
\rho^0_{\mathsf{A}} = \frac{1}{p_0} \cdot
 \begin{pmatrix}
|a_0|^2 |b_0|^2 & 0\\
0 & |a_1|^2 |b_1|^2
\end{pmatrix},
\quad \rho^1_{\mathsf{A}} = \frac{1}{p_1} \cdot 
\begin{pmatrix}
|a_0|^2 |b_1|^2 & 0\\
0 & |a_1|^2 |b_0|^2
\end{pmatrix}.
\end{align*}
Moreover, the equality is achieved when the measurement is performed in the Bell basis as in \defref{bell-basis}.
\end{theorem}
To give a detailed proof of  \thmref{BinaryEntanglementSwapping}, we consider the
experiment of measuring $\ket{a} \otimes \ket{b}$
in the basis $\ket{x}_{\mathsf{AD}} \otimes \ket{u_i}_{\mathsf{BC}}$ over $x \in \{0, 1\}^2$ and $i \in \{1, 2, 3, 4\}$.
If we do so, then by Equation~\eqref{eq:post-measurement-unnormalized}
we receive outcome $(x, i)$ with probability $|u^i_{x}|^2 \cdot |a_{x_1}|^2 |b_{x_2}|^2$.
This motivates studying the following random variable.

\begin{definition}[Measurement outcome random variable]
We define $(\bX, \bI)$ to be the random variable which is equal to $(x, i)$ with probability $|u^i_{x}|^2 \cdot |a_{x_1}|^2 |b_{x_2}|^2$.
\end{definition}

One property we will need of this random variable is whether the string~$\bX$ is even or odd.
This is formalized in the following definition.

\begin{definition}
We write $\mathrm{odd}(\bX)$ for the $\{0, 1\}$-valued random variable defined as
\begin{equation*}
\mathrm{odd}(\bX) := \bX_1 + \bX_2 \pmod{2}.
\end{equation*}
Equivalently, $\mathrm{odd}(\bX) = 1$ when $\bX_1 = \bX_2$
and~$0$ otherwise.
\end{definition}

Our key technical lemma relates the von Neumann entropy of the state~$\rho_i$
with the Shannon entropy of the random variable $(\bX, \bI)$.
It is stated as follows.

\begin{lemma}\label{lem:BoundonEntropy}
For each $i \in \{1, 2, 3, 4\}$,
\begin{equation}
S(\rho^i_{\mathsf{A}}) \leq H(\bX \mid \mathrm{odd}(\bX), \bI=i).\label{eq:BoundonEntropy}
\end{equation}
\end{lemma}

Prior to proving \lemref{BoundonEntropy},
we show how it can be used to complete the proof of \thmref{BinaryEntanglementSwapping}.
\begin{proof}[Proof of \thmref{BinaryEntanglementSwapping}]

We first specialize to the case that the entanglement swapping  measurement is performed in the Bell basis. 
\propref{general-state-in-bell-basis} for qubits $d = 2$, implies that
\begin{equation*}
\ket{a}_{\mathsf{AB}}\otimes \ket{b}_{\mathsf{CD}}
= \frac{1}{\sqrt{2}} \cdot \sum_{i, j \in \{0,1\}} \Bell{i}{j}_{\mathsf{BC}}
	\ot \Big(a_0\cdot  b_{j} \cdot \ket{0, j}
			+ (-1)^{i}\cdot a_1\cdot  b_{1 + j} \cdot \ket{1, 1+j} \Big)_{\mathsf{AD}},
\end{equation*}
and that $\boldsymbol{\rho}$ is distributed as follows:
\begin{align}
\rho^0_{\mathsf{A}} = \frac{1}{p_0} \cdot
 \begin{pmatrix}
|a_0|^2 |b_0|^2 & 0\\
0 & |a_1|^2 |b_1|^2
\end{pmatrix},
& \qquad \text{with probability $p_0 =|a_0|^2 |b_0|^2 + |a_1|^2 |b_1|^2$},\nonumber\\
\rho^1_{\mathsf{A}} = \frac{1}{p_1} \cdot
\begin{pmatrix}
|a_0|^2 |b_1|^2 & 0\\
0 & |a_1|^2 |b_0|^2
\end{pmatrix},
& \qquad \text{with probability $p_1 = |a_0|^2 |b_1|^2 + |a_1|^2 |b_0|^2$}.\label{eq:mixtureForBoundEntropy}
\end{align}
Hence, when measured in the Bell basis, we have $\bm{E}(\bm{\rho}_{\mathsf{A}}) = p_0 \cdot S(\rho_0) + p_1 \cdot S(\rho_1)$.
We next show that in an arbitrary basis, it holds that $\bm{E}(\bm{\rho}_{\mathsf{A}}) \leq p_0 \cdot S(\rho^0_{\mathsf{A}}) + p_1 \cdot S(\rho^1_{\mathsf{A}})$.
This follows from
\begin{align*}
\sum_{i=1}^4 p_i S(\rho^i_{\mathsf{A}})
&=\sum_{i=1}^4 \Pr[\bI = i] \cdot S(\rho^i_{\mathsf{A}}) \tag{by definition of~$\bI$}\\
&\leq \sum_{i=1}^4 \Pr[\bI = i] \cdot H(\bX \mid \mathrm{odd}(\bX), \bI=i)
		\tag{by \lemref{BoundonEntropy}}\\
&=H(\bX \mid \mathrm{odd}(\bX), \bI)\\
&\leq H(\bX \mid \mathrm{odd}(\bX)) \tag{because conditioning reduces entropy}\\
&= H(\bX) - H(\mathrm{odd}(\bX)).
\end{align*}
But
\begin{equation*}
H(\bX) = H(|a_0|^2, |a_1|^2) + H(|b_0|^2, |b_1|^2)
\end{equation*}
and
\begin{equation*}
H(\mathrm{odd}(\bX)) = H(|a_0|^2 |b_0|^2 + |a_1|^2 |b_1|^2, |a_0|^2 |b_1|^2 + |a_1|^2 |b_0|^2).
\end{equation*}
Hence, a direct calculation shows that $$H(\bX) - H(\mathrm{odd}(\bX)) = p_0\cdot S(\rho^0_{\mathsf{A}}) + p_1\cdot S(\rho^1_{\mathsf{A}}),$$
where probabilities $p_0$ and $p_1$, and the states $\rho_0$ and $\rho_1$ are given by \eqref{eq:mixtureForBoundEntropy}.
This leads to the claimed bound in \eqref{eq:EntropyBoundSwapping} and completes the proof of \thmref{BinaryEntanglementSwapping}.
\end{proof}
Finally, we state the proof of \lemref{BoundonEntropy}.
\begin{proof}[Proof of \lemref{BoundonEntropy}]
Let $\ket{\phi}_{\mathsf{AD}}=a \ket{00}+b \ket{11}+c \ket{01}+d \ket{10}$ be a normalized bipartite state.
This state may correspond to any of the post-measurement states $\ket{\psi_i}_{\mathsf{AD}}$ for $i \in \{1,2,3,4\}$.
The reduced state on subsystem $\mathsf{A}$ is given by
\begin{align}
    \rho=\begin{pmatrix}
    |a|^2+|c|^2 & a d^*+b^* c\\
    a^* d+b c^* & |b|^2+|d|^2
    \end{pmatrix}.\nonumber
\end{align}
As before, we consider the
experiment of measuring $\ket{\phi}_{\mathsf{AD}}$
in the basis $\{\ket{x}_{\mathsf{AD}}: x \in \{0, 1\}^2\}$ with the random variable $\bX = (\bX_1,\bX_2)$ denoting the measurement outcome. 
The following projective measurements correspond to obtaining $\mathrm{odd}(\bX) = 0$ and $\mathrm{odd}(\bX) = 1$:
\begin{align}
    &\mathrm{odd}(\bX) = 0, \quad \Pi_0=\ketbra{00}{00}+\ketbra{11}{11},\quad\nonumber\\
    &\mathrm{odd}(\bX) = 1, \quad \Pi_1=\ketbra{01}{01}+\ketbra{10}{10}.\nonumber
\end{align}
For each of these two measurements, the reduced post-measurement state and its measurement probability is
\begin{align}
    &\mathrm{odd}(\bX) = 0:\quad  \sigma_0=\begin{pmatrix}
\frac{|a|^2}{|a|^2+|b|^2}  & 0\\
0 & \frac{|b|^2}{|a|^2+|b|^2}
\end{pmatrix},\quad \textit{with probability}\quad q_0 = |a|^2 + |b|^2,\nonumber\\
 &\mathrm{odd}(\bX) = 1:\quad \sigma_1=\begin{pmatrix}
\frac{|c|^2}{|c|^2+|d|^2}  & 0\\
0 & \frac{|d|^2}{|c|^2+|d|^2}
\end{pmatrix},\quad \textit{with probability.}\quad q_1=|c|^2+|d|^2.\label{eq:statesforOddx}
\end{align}
Our goal is to prove the bound \eqref{eq:BoundonEntropy} which in this setup is equivalent to proving
\begin{align}
S(\rho)\leq q_0 \cdot S(\sigma_0) + q_1 \cdot S(\sigma_1).\label{eq:equivalentEntropybound}
\end{align} 
We prove this by interpreting the entropy term $S(\rho)$ as the entanglement of formation of a two-qubit mixed state $\tau_{\mathsf{AD}}$ and drawing on \defref{EntanglementFormation}.
Consider the following pair of pure states $$\ket{\phi_0}_{\mathsf{AD}} = \frac{1}{\sqrt{|a|^2+|b|^2}}\left(a\ket{00}+ b\ket{11}\right), \quad \ket{\phi_1}_{\mathsf{AD}} = \frac{1}{\sqrt{|c|^2 + |d|^2}}\left(c\ket{01} + d\ket{10}\right),$$
and their mixture $\tau_{\mathsf{AD}}$ given by 
\begin{align}
    \tau_{\mathsf{AD}} = p_0\cdot \ketbra{\phi_0}{\phi_0}_{\mathsf{AD}} + p_1\cdot \ketbra{\phi_1}{\phi_1}_{\mathsf{AD}}
    = \begin{pmatrix}
        |a|^2 & 0 & 0 & a b^* \\
        0 & |c|^2 & c d^* & 0 \\
        0 & c^*d & |d|^2 & 0 \\
        a^* b & 0 & 0 & |b|^2
    \end{pmatrix}, 
\end{align}
where as before $q_0 = |a|^2 + |b|^2$ and $q_1 = |c|^2 + |d|^2$.
The states $\sigma_0$ and $\sigma_1$ in \eqref{eq:statesforOddx} correspond to the reduced state of $\ket{\phi_1}_{\mathsf{AD}}$ and $\ket{\phi_2}_{\mathsf{AD}}$ on subsystem $\mathsf{A}$.
Hence, following \defref{EntanglementFormation}, the entanglement of formation $E_F(\tau_{\mathsf{AB}})$ satisfies
\begin{align}
    E_F(\tau_{\mathsf{AB}})\leq q_0\cdot  S(\sigma_0) + q_1 \cdot S(\sigma_1).\label{eq:EfBound}
\end{align}
Using \factref{EntanglementFormation}, the entanglement of formation $E_F(\tau_{\mathsf{AB}})$ can be explicitly calculated as
\begin{align}
    E_F(\tau_{\mathsf{AB}}) = h\left( \frac{1 + \sqrt{1- C(\tau_{\mathsf{AD}})^2}}{2}\right),
\end{align}
where the concurrence $C(\tau_{\mathsf{AD}})$ is given by $C(\tau_{\mathsf{AD}}) = 2 \cdot |a b - c d|$ as per \factref{concurrenceFormula}.
We now claim that the entropy $S(\rho)$ in \eqref{eq:equivalentEntropybound} is equal to $E_F(\tau_{\mathsf{AB}})$. 
This follows from the fact that
\begin{align}
        S(\rho) = h\left( \frac{1 + \sqrt{1- C(\phi_{\mathsf{AD}})^2}}{2}\right)
\end{align}
according to \factref{EntropyConcurrence}, and that $C(\phi_{\mathsf{AD}}) = 2 \cdot |a b -c d|$ using \factref{concurrenceFormula}.
Going back to \eqref{eq:EfBound}, we see that $S(\rho) \leq q_0 \cdot S(\sigma_0) + q_1 \cdot S(\sigma_1)$, which completes the proof. 
\end{proof}

Next, we give an alternative looser upper bound on the average entropy $\E S(\bm{\rho}_{\mathsf{A}})$ again by using the notion of concurrence.
The usefulness of this bound will become clear in the next section where we consider a contiguous series of entanglement swapping measurements performed on a chain of two-qubit entangled pairs.
A similar analysis for a pair of qubits has been considered before in \cite{Popp2005Localizable, Verstraete2004Localizable} in the context of localizable entanglement. 
 
We begin with a known general upper bound on the entropy $S(\rho)$ in terms of concurrence $C(\rho)$.

\begin{theorem}[Entropy bound in terms of concurrence, cf. \cite{Uhlmann2009entropy}]\label{thm:concurrence}
Let $\rho$ be a rank $d$ state. The following upper bound on the von Neumann entropy of $\rho$ holds:
\begin{align}
    S(\rho)\leq c_d \cdot C(\rho).
\end{align}
Here $c_d=\log(d) \cdot \sqrt{\frac{d}{2(d-1)}}$ and the concurrence is given by $C(\rho) = \sqrt{2\cdot (\Tr(\rho)^2-\Tr(\rho^2))}$.
In particular, in the case of a qubit $d=2$, we can bound the entropy $S(\rho)$ of a qubit by its concurrence $C(\rho)$ via
\begin{align}
    S(\rho) \leq C(\rho) \quad \text{for qubits}.\label{eq:EntropyBoundConcurrence}
\end{align}
\end{theorem}
By applying the bound \eqref{eq:EntropyBoundConcurrence}, we prove the following result.
\begin{proposition}[Alternative upper bound on average entropy ]\label{prop:qubit_bound}
Following an entanglement swapping measurement on qubit states $\ket{a}_{\mathsf{AB}} \otimes \ket{b}_{\mathsf{CD}}
= \Big(a_0 \ket{00} + a_1 \ket{11}\Big)
	\ot \Big(b_0 \ket{00} + b_1 \ket{11}\Big)$ as in \defref{entanglementSwapping}, 
the average entropy $\E S(\bm{\rho}_\mathsf{A})$ and concurrence $\E C(\bm{\rho}_\mathsf{A})$ of the state $\rho_{\mathsf{A}}$ satisfy
\begin{align}
    \E S(\bm{\rho}_\mathsf{A})\leq \E C(\bm{\rho}_\mathsf{A})\leq  4\cdot|a_0 a_1|\cdot |b_0 b_1|.\label{eq:f7}
\end{align}
\end{proposition}
The proof of this proposition entails  finding the concurrence of the post-measurement states which we detail first. 
Using our notation in equations  \eqref{eq:post-measurement-unnormalized}-\eqref{eq:postMeasurementState_2}, the un-normalized state on subsystems $\mathsf{AD}$ after measuring subsystems $\mathsf{BC}$ is given by 
\begin{align}
\ket{\psi_i}_{\mathsf{AD}} = 
 \sum_{x \in \{0, 1\}^2} (u_{x}^i)^\dagger a_{x_1} b_{x_2}\cdot \ket{x}, 
\end{align}
and the corresponding reduced state on subsystem $\mathsf{A}$ is denoted by $\tilde{\rho}^i_{\mathsf{A}} =  p_i \cdot \rho^i_{\mathsf{A}}$.
Given this and the definition of concurrence in equation \eqref{eq:concurrenceExpression}, we see that $C(\tilde{\rho}^i_{\mathsf{A}}) = p_i  \cdot C(\rho_i)$, and therefore, the average concurrence satisfies $\sum_i p_i \cdot C(\rho^i_{\mathsf{A}}) = \sum_i C(\tilde{\rho}^i_{\mathsf{A}})$.

This expression for this concurrence in the general case of qudits is stated in the next lemma.
\begin{lemma}[Post-measurement concurrence]\label{lem:qudit_expression}
The concurrence of the post-measurement state $C(\tilde{\rho}^i_{\mathsf{A}})$ is given by
\begin{align}
    C(\tilde{\rho}^i_{\mathsf{A}})= \left(\sum_{x_1<x'_1,x_2<x'_2} |u^i_{x_1 x_2} u^i_{x'_1 x'_2}-u^i_{x_1 x'_2} u^i_{x'_1 x_2}|^2 |a_{x_1}|^2 |a_{x'_1}|^2 |b_{x_2}|^2 |b_{x'_2}|^2\right)^{1/2},\label{eq:f2}
\end{align}
\end{lemma}
\begin{proof}
A direct calculation shows that 
\begin{align}
    \Tr[\tilde{\rho}^i_{\mathsf{A}}]=\sum_{x_1,x'_1\in\{0,1\}} |u^i_{x_1x'_1}|^2 |a_{x_1}|^2 |b_{x'_1}|^2\nonumber
\end{align}
and 
\begin{align}
    \Tr[(\tilde{\rho}^i_{\mathsf{A}})^2]=\sum_{x,x'\in\{0,1\}^2} u^i_{x_1x_2}u^{i*}_{x'_1 x_2} u^i_{x'_1 x'_2} u^{i*}_{x_1 x'_2} |a_{x_1}|^2 |b_{x_2}|^2 |a_{x'_1}|^2 |b_{x'_2}|^2.\nonumber
\end{align}
Hence, 
\begin{align}
    C(\tilde{\rho}^i_{\mathsf{A}})^2&=2 \cdot \sum_{x_1,x_2,x'_1,x'_2}\left(|u^i_{x_1x_2}|^2 |u^i_{x'_1,x'_2}|^2 -u^i_{x_1,x_2}u^{i*}_{x'_1x_2}u^i_{x'_1x'_2}u^{i*}_{x_1x'_2}\right)\cdot  |a_{x_1}|^2 |a_{x'_1}|^2 |b_{x_2}|^2 |b_{x'_2}|^2\nonumber\\
    &= 4\cdot \sum_{x_1<x'_1,x_2<x'_2} |u^i_{x_1 x_2} u^i_{x'_1 x'_2}-u^i_{x_1 x'_2} u^i_{x'_1 x_2}|^2 |a_{x_1}|^2 |a_{x'_1}|^2 |b_{x_2}|^2 |b_{x'_2}|^2.
\end{align}
\end{proof}
Next we show how in the case of qubits ($d=2$), the bound \eqref{eq:f7} on the average entropy can be deduced from \eqref{eq:f2}:

\begin{proof}[Proof of \propref{qubit_bound}]
In this case, the bounds \eqref{eq:EntropyBoundConcurrence} and $\eqref{eq:f2}$ imply 
\begin{align}
    \sum_{i=1}^4 p_i \cdot S(\rho^i_{\mathsf{A}})\leq  \sum_{i=1}^4 p_i \cdot C(\rho^i_{\mathsf{A}}) &\leq 2 \cdot |a_0a_1|\cdot |b_0b_1|\cdot \sum_{i=1}^4 |u^i_{00} u^i_{11}-u^i_{10} u^i_{01}|\label{eq:abBoundConcurrence}\\
    &\leq |a_0a_1|\cdot |b_0||b_1| \cdot \sum_{i=1}^4 \left((u_{00}^i)^2+(u_{11}^i)^2+(u_{01}^i)^2+(u_{10}^i)^2\right)\nonumber\\
    &= 4 \cdot |a_0a_1|\cdot |b_0b_1|.\nonumber
\end{align}
We note that the second inequality is tight when subsystems $\mathsf{CD}$ are measured in the Bell basis.
\end{proof}
\subsection{Cluster state}\label{sec:CMIChainQubits}
We next apply the previous lemma along a one-dimensional chain to show that the conditional correlations between the two end points of the chain decay exponentially with the length of the~chain.
\begin{theorem}[Decay of CMI for a chain of qubits]\label{thm:ChainQubitsCMI}
Consider the state $\ot_{k=1}^{n}\ket{a}_{\mathsf{L}_k\mathsf{R}_k}$ where $\ket{a}_{\mathsf{L}_k\mathsf{R}_k}=a_{00}\ket{00}_{\mathsf{L}_k\mathsf{R}_k}+a_{11}\ket{11}_{\mathsf{L}_k\mathsf{R}_k}$. 
Suppose we measure each pair of registers $(\mathsf{R}_j, \mathsf{L}_{j+1})$ for $j\in [n-1]$ in an arbitrary four dimensional orthonormal basis.
Let $\bm{\rho}_{1}$ be the post-measurement state on qubit $\mathsf{L}_1$.
It holds~that
\begin{align}
\E S(\bm{\rho}_{1})  \leq |2 a_0 a_1|^n. \label{eq:f6}
\end{align}
Since $a_0a_1\leq \frac{1}{2}$, this implies an exponentially decaying bound $\E S(\bm{\rho}_{1}) \leq e^{-\Omega(n)}$ in terms of $n$ unless $a_0=a_1=1/\sqrt{2}$, i.e. the state $\ket{a}$ is a maximally entangled state.
\end{theorem}
\begin{proof}
We assume that the joint measurements on registers $(\mathsf{R}_j,\mathsf{L}_{j+1})$ are performed sequentially starting from $(\mathsf{R}_1, \mathsf{L}_{2})$.
The measurement outcomes in each of these measurements are denoted by random variables $\bm{i}_1,\dots, \bm{i}_{n-1} \in \{0,1\}^2$.
Given these outcomes, we let $\bm{\rho}^{(\bm{i}_1\,\dots, \bm{i}_{k})}_1$ be the post-measurement state on subsystem $\mathsf{A}$ upon the first $k$ measurements.
Our goal is to upper bound $\E_{\bm{i}_1,\dots, \bm{i}_{n-1}} S(\bm{\rho}_1^{(\bm{i}_1\,\dots, \bm{i}_{n-1})})$.
This can be achieved by finding an upper bound on the expected concurrence $\E_{\bm{i}_1,\dots, \bm{i}_{n-1}} C(\bm{\rho}^{(\bm{i}_1\,\dots, \bm{i}_{n-1})}_1)$, and applying the inequality \eqref{eq:EntropyBoundConcurrence} in \thmref{concurrence}.

Consider the last measurement performed on subsystems $\mathsf{R_{n-1} L_n}$, and suppose the previous measurements have returned the values $i_1,\dots, i_{n-2}$ with a corresponding post-measurement state $\rho^{(i_1,\dots, i_{n-2})}$ on subsystems $\mathsf{L_1}$. 
Using \eqref{eq:f7}, we have that 
\begin{align}
    \E_{\bm{i}_{n-1}} C(\bm{\rho}_1^{(i_1,\dots,i_{n-2},\bm{i}_{n-1})})\leq 2\cdot |a_0 a_1| \cdot C(\rho_1^{(i_1,\dots,i_{n-2})}).\label{eq:boundExpectedConcurrence}
\end{align}
This implies the following bound on the expected concurrence over the measurement outcomes on subsystems $\mathsf{R}_{n-2}\mathsf{L}_{n-1}$
\begin{align}
    \E_{\bm{i}_{n-2}, \bm{i}_{n-1}} C(\bm{\rho}_1^{(i_1,\dots,i_{n-2},\bm{i}_{n-1})}) &\leq |2 a_0 a_1| \cdot  \E_{\bm{i}_{n-2}} C(\bm{\rho}_1^{(i_1,\dots,\bm{i}_{n-2})})\nonumber\\
    &\leq |2 a_0 a_1|^2 \cdot C(\rho_1^{(i_1,\dots, i_{n-3})}), 
\end{align}
where the last inequality via another application of \eqref{eq:f7}.
Iterating in this manner over all the $n-1$ measurements yields the claimed  bound \eqref{eq:f6}.
\end{proof}

We now apply this bound to the rotated 1D cluster state which is the ground state of the Hamiltonian in \eqref{eq:rotated_clusterH}. 
In the standard basis, the rotated cluster state is the ground state of $H = - \sum _{k=2}^{n-1} X_{k-1}Z_k X_{k+1}- Z_1 X_2 - X_{n-1} X_n -  X_{n-2} Z_{n-1} Z_n.$
This state can be prepared by first producing a series of EPR pairs on qubits $i, i+1$ for $i\in [n-1]$. 
Then a Bell-basis unitary is applied on qubits $2i, 2i+1$ for $i=n/2-1$. 

 Hence, due to the measurement-induced entanglement, the measurement distribution in the $Z$-basis exhibits long-range conditional correlations. In contrast, measuring in the $X$-basis results in a uniform distribution that lacks any correlations. More generally, measuring in a rotated basis induces an exponential decay of correlations, as formally stated in the following corollary:

\begin{corollary}[Decay of CMI in cluster state, restatement of \thmref{decayCMIClusterState}]\label{cor:decayCMIClusterState}
Consider the rotated $n$-qubit cluster state given by the circuit in \fig{ClusterState} for an even $n$.
This circuit prepares $n/2$ EPR pairs followed by a Bell basis rotation. We then apply a $y$-rotation \footnote{This is the same as changing $X$ and $Z$ Pauli operators in the cluster state Hamiltonian to $X \rightarrow \cos(2\theta) \cdot X  - \sin(2\theta)\cdot Z,\quad
Z \rightarrow \cos(2\theta) \cdot Z + \sin(2\theta)\cdot X$.} given by
\begin{align}
    R_y(\theta) = \begin{pmatrix}
    \cos(\theta/2) && \sin(\theta/2)\\
    -\sin(\theta/2) &&\cos(\theta/2)
\end{pmatrix} = \cos(\theta/2)\cdot \iden+ i \sin(\theta/2) \cdot Y, \quad \quad \theta \in [0, \pi/2]
\end{align}
to each qubit and proceed to measure all but the first and the last qubits in the standard $Z$ basis. 
The post-measurement state $\bm{\rho}_{1}$ on the first qubit satisfies 
\begin{align}
\E S(\bm{\rho}_{1})  \leq \cos(\theta)^{n-2}. \label{eq:CMIDecayCluster}
\end{align}
\end{corollary}
\begin{proof}
We first bound the expected entropy after one application of the entanglement swapping measurement as in \eqref{eq:abBoundConcurrence}.
One can directly calculate that in the $y$-rotated basis, we have
\begin{align}
    |u^i_{00} u^i_{11}-u^i_{10} u^i_{01}| = \frac{1}{2}\cdot \cos(\theta)^2.\nonumber
\end{align}
Applying this to  \eqref{eq:abBoundConcurrence} yields
\begin{align}
    \sum_{i=1}^4 p_i S(\rho^i_{1})\leq  \sum_{i=1}^4 p_i C(\rho^i_{1}) &\leq 2 \cdot |a_0a_1|\cdot |b_0b_1|\cdot \sum_{i=1}^4 |u^i_{00} u^i_{11}-u^i_{10} u^i_{01}|\nonumber\\
    &= \cos(\theta)^2\cdot 2|a_0 a_1| \cdot 2|b_0 b_1|\nonumber
\end{align}
The remainder of the proof follows similar to that of \thmref{ChainQubitsCMI} (note the different number of qubits used in \thmref{ChainQubitsCMI}).
In particular, the cluster state satisfies $a_0 = a_1 = \frac{1}{\sqrt{2}}$, and hence, the bound in equation \eqref{eq:boundExpectedConcurrence} changes to
\begin{align}
    \E_{\bm{i}_{n-1}} C(\bm{\rho}_1^{(i_1,\dots,i_{n-2},\bm{i}_{n-1})})\leq \cos(\theta)^2\cdot C(\rho_1^{(i_1,\dots,i_{n-2})}).\nonumber
\end{align}
A successive application of this bound yields the claimed bound in \eqref{eq:CMIDecayCluster}.
\end{proof}

\subsection{Conditional independence  
for random shallow 1D circuits}\label{sec:CMIShallowCkt}
Our study of conditional independence in \secref{CMIChainQubits} and \secref{OptimalEntanglementSwapping} has focused on states that can be prepared by depth-$2$ quantum circuits. 
The first layer of gates in the circuit prepares entangled pairs of qubits.
The second layer, followed by measurements in the computational basis, forms the two-qubit measurements that cause the entanglement swapping. 

We presented two mechanisms for the emergence of conditionally-independent correlations after measurements.
The first, as stated in \thmref{ChainQubitsCMI}, occurs when partially entangled states are prepared in the first layer.
The second, as stated in \corref{decayCMIClusterState}, occurs when the two-qubit measurements are not maximally entangling. 

In this section, we extend this analysis to the case of depth-$D$ one-dimensional quantum circuits, for any constant~$D$.
We show that random shallow quantum circuits exhibit conditional independence due to both of the mechanisms stated above.

\subsubsection{Setup}
We slightly adjust the notation previously established in \defref{entanglementSwapping} for clarity.
Initially, we consider two pairs of entangled qudits (dimension $d_0$) in the state $\ket{\psi} = \ket{a}_{\mathsf{AB}} \otimes \ket{b}_{\mathsf{CD}}$. 
We then apply a rotation to this state, resulting in $U_{\mathsf{BC}} \otimes \iden_{\mathsf{AD}} \cdot \ket{a}_{\mathsf{AB}} \ket{b}_{\mathsf{CD}}$. 
Subsequently, we measure the $\mathsf{BC}$ subsystem in the computational basis. 
The probability of obtaining outcome $i \in [d_0]$ after measuring subsystem $\mathsf{B}$ but before measuring subsystem $\mathsf{C}$ is represented by $p_i$.
Conditioned on this outcome, the post-measurement state on $\mathsf{ABCD}$ is $\ket{\psi_i}$ with $\rho^i_{\mathsf{A}}$ denoting the reduced state on subsystem $\mathsf{A}$.

The Schmidt decomposition of the state $\ket{a}$ is given by $\ket{a}_{\mathsf{AB}} = \sum_{k=1}^{d_0} \sqrt{\lambda_k} \ket{k}_{\mathsf{A}} \ket{k}_{\mathsf{B}}$.
Now consider an experiment where the state $\ket{k}_{\mathsf{B}}\ket{b}_{\mathsf{CD}}$ is first rotated by $U_{\mathsf{BC}} \otimes \iden_{\mathsf{D}}$ and then its subsystem $\mathsf{B}$ is measured in the computational basis and the outcome $i \in [d_0]$ is observed.
The state $\ket{\phi_{i,k}}_{\mathsf{CD}}$ which describes the resulting post-measurement state can be expressed as
\begin{align}
    \ket{\phi_{i,k}}_{\mathsf{CD}} = \frac{1}{q_{i,k}^{1/2}} (\bra{i}_{\mathsf{B}}\otimes \iden_{\mathsf{CD}})\cdot (U_{\mathsf{BC}}\otimes \iden_{\mathsf{D}})\cdot \ket{k}_{\mathsf{B}} \ket{b}_{\mathsf{CD}}, \label{eq:postMeasurementState}
\end{align}
where the normalization factor $q_{i,k}$ is defined by 
\begin{align}
    q_{i,k} = \Tr[(\ketbra{i}{i}_{\mathsf{B}}\otimes \iden_{\mathsf{CD}})\cdot (U_{\mathsf{BC}}\otimes \iden_{\mathsf{D}})\cdot (\ketbra{k}{k}_{\mathsf{B}}\otimes \ketbra{b}{b}_{\mathsf{CD}})\cdot(U_{\mathsf{BC}}^\dagger\otimes \iden_{\mathsf{AD}})].\nonumber
\end{align}
The post-measurment state $\ket{\psi_i}$ resulting form measuring $\ket{a}_{\mathrm{AB}}\otimes \ket{b}_{\mathrm{CD}}$ admits the following~expression 
\begin{align}
    \ket{\psi_i} = \frac{1}{\sqrt{p_i}}\sum_k \sqrt{\lambda_k q_{i,k}}\cdot |k\rangle_{\mathsf{A}} \ket{i}_{\mathsf{B}} |\phi_{i,k}\rangle_{\mathsf{CD}}, \text{ with }  p_i = \sum_k \lambda_k q_{i,k}.\label{eq:PostMeasurementPsi_1}
\end{align}
There is also a simpler interpretation to the above, in terms of a random variable $\bm{X}$ which is set to $k\in [d_0]$ with probability $\lambda_k$, and the measurement outcome $\bm{I}$:
    \begin{equation}
    \ket{\psi_i} = \sum_k \sqrt{\Pr[\bm{X} = k|\bm{I} = i]}\cdot \ket{k}_{\mathsf{A}}\ket{i}_{\mathsf{B}}\ket{\phi_{i,k}}_{\mathsf{CD}}.\label{eq:alternativePostMeasurementState}
\end{equation}

\subsubsection{Imperfect swapping from the top Schmidt coefficient}

Our proof strategy is inspired by that of \cite{napp2019efficient2d}, which studies random depth-2 quantum circuits. 
We proceed by establishing a lower bound for the operator norm of the post-measurement reduced state.
This approach demonstrates that the entanglement swapping measurements, resulting from the application of random shallow circuits, are `imperfect' in a well-defined sense.

\begin{definition}[Imperfect entanglement swapping]\label{def:def-imperfect}
    An entanglement swapping protocol on subsystems $\mathsf{ABCD}$ with post-measurement stated given by \eqref{eq:postMeasurementState} is called $\epsilon$-imperfect if $$\min_{i, k, k'}|\braket{\phi_{i, k}}{\phi_{i, k'}}|^2 = \epsilon>0.$$
\end{definition}
\begin{theorem}[Bounding the operator norm, cf. \cite{napp2019efficient2d}]\label{thm:entswap-claim}
    Following the application of an $\epsilon$-imperfect entanglement swapping, the average top Schmidt coefficient of the reduced state on subsystem $\mathsf{A}$ is bounded from below by
    \begin{align}
        \E_{\bm{i}} \norm{\bm{\rho}_{\mathsf{A}}^{\bm{i}}}_\infty \geq (1-\epsilon)\cdot \norm{\rho_{\mathsf{A}}}_\infty  + \epsilon,\label{eq:boundExpectedTopSchmidt}
    \end{align}
where $\rho_{\mathsf{A}}$ is the state on subsystem $\mathsf{A}$ prior to the measurements. 
In particular, this implies that unless $\norm{\rho_{\mathsf{A}}}_\infty = 1$ we have $ \E_{\bm{i}} \norm{\bm{\rho}_{\mathsf{A}}^{\bm{i}}}_\infty > \norm{\rho_{\mathsf{A}}}_\infty$ for imperfect measurements.
\end{theorem}

\begin{proof}
    Since after measuring subsystem $\mathsf{B}$, the post-measurement state on $\mathsf{ABCD}$ is pure and classical on subsystem $\mathsf{B}$, it holds that $\norm{\rho^i_{\mathsf{A}}}_\infty = \norm{\rho_{\mathsf{CD}}^i}_\infty$ where $\rho^i_{\mathsf{CD}}$ is the reduced state on $\mathsf{CD}$ conditioned on observing outcome $i \in [d_0]$ upon measuring $\mathsf{B}$. 
    Using the expression \eqref{eq:alternativePostMeasurementState}, the reduced state $\rho^i_{\mathsf{CD}}$ satisfies 
    \begin{align}
    \norm{\rho^i_{\mathsf{CD}}}_\infty &\geq \max_k\ \bra{ \phi_{i, k}}\cdot \rho^i_{\mathsf{CD}}\cdot \ket{\phi_{i, k}} \nonumber \\ 
    &= \max_k \ (\Pr[\bm{X}=k|\bm{I}=i] + \sum_{k'\neq k}\Pr[X=k'|\bm{I}=i] \cdot |\braket{\phi_{i, k}}{\phi_{i, k'}}|^2 )\nonumber \\ 
    &\geq  \max_k \ (\Pr[\bm{X}=k|\bm{I}=i] + (1-\Pr[\bm{X}=k|\bm{I}=i]) \cdot \epsilon)
    \end{align} 
for any choice of $k$. 
In expectation over $i \in [d_0]$ and picking $k$ to be the index of the top eigenvalue of the pre-measurement reduced state $\rho$ on subsystem $\mathsf{A}$, we observe the claimed bound in \eqref{eq:boundExpectedTopSchmidt}.

\end{proof}

In our previous analysis, the measurement is only performed on subsystem $\mathsf{B}$. 
Let $\ket{\psi_{ij}}$ denote the reduced state on subsystem $\mathsf{A}$ conditioned on observing outcomes $i \in [d_0]$ and $j \in [d_0]$ after measuring $\mathsf{BC}$.
We next show that a convexity argument ensures that on average, the top Schmidt coefficient $\norm{\rho_{\mathsf{A}}^{(i,j)}}_{\infty}$ strictly increases. 
\begin{corollary}\label{cor:DecayByConvexity}
 Following the application of an $\epsilon$-imperfect entanglement swapping, the average top Schmidt coefficient of the reduced state on subsystem $\mathsf{A}$ is bounded from below by
    \begin{equation}
        \E_{\bm{i}, \bm{j}} \norm{\bm{\rho}_{\mathsf{A}}^{(\bm{i},\bm{j})}}_\infty \geq \E_{\bm{i}} \norm{\E_{\bm{j}} \bm{\rho}_{\mathsf{
        A}}^{(\bm{i},\bm{j})}}_{\infty} \geq (1-\epsilon)\cdot \norm{\rho_A}_\infty + \epsilon
    \end{equation}
\end{corollary}

We now discuss how Haar random unitaries and low-depth, random quantum circuits give rise to imperfect entanglement swapping. Namely, we argue that in expectation over the randomness of the unitaries, these protocols are $\epsilon$-imperfect for some constant $\epsilon$, a function of the local dimension (in the Haar random case) and circuit depth. 

In the upcoming analysis, it is convenient to use the following un-normalized version of the overlap between states $\ket{\phi_{i,k}}$ and $\ket{\phi_{i,k'}}$ for $k\neq k'$ which are defined in \eqref{eq:postMeasurementState}:
\begin{align}
    &f(U_{\mathsf{BC}}, i, k, k',\sigma_{\mathsf{C}}):= q_{i,k} \cdot q_{i,k'}\cdot |\braket{\phi_{i,k}}{\phi_{i,k'}}|^2\nonumber\\
    &= \Tr_{\mathsf{BC}, \mathsf{B'C'}}\bigg[(U^\dagger_{\mathsf{BC}}\otimes U^\dagger_{\mathsf{\mathsf{B'C'}}})(\ketbra{i}{i}_{\mathsf{B}}\otimes \ketbra{i}{i}_{\mathsf{B'}}\otimes \iden_{\mathsf{CC'}})(U_{\mathsf{BC}}\otimes U_{\mathsf{B'C'}})\cdot \big(\ketbra{k}{k'}_B \otimes \sigma_{\mathsf{C}}\otimes \ketbra{k'}{k}_{\mathsf{B'}} \otimes \sigma_{\mathsf{C'}} \big)\bigg].\label{eq-imperfect}
    \end{align}
where $$\sigma_{\mathsf{C}} = \Tr_{\mathsf{D}} \ketbra{b}{b}_{\mathsf{CD}}$$ is the reduced state on subsystem $\mathsf{C}$ of the original state $\ket{a}_{\mathsf{AB}}\otimes \ket{b}_{\mathsf{CD}}$ prior to the entanglement swapping measurement.
The notion of $\epsilon$-imperfectness for an entanglement swapping protocol in \defref{def-imperfect} relies on establishing that $\min_{i,k,k'} |\braket{\phi_{i,k}}{\phi_{i,k'}}|^2 \geq \epsilon$.
However, dealing with the normalization factors $q_{i,k}$ and $q_{i,k'}$, turns out to be challenging when averaging over random choices of unitaries $U$.
In the next lemma, we show that one can instead directly lower bound the un-normalized version of the overlaps $\min_{i, k, k'\in [d_0]} f(U, i, k, k', \sigma_{\mathsf{C}}) \geq \epsilon$.
This results in a looser $(d_0\cdot\epsilon)$-imperfectness, but is sufficient for our application of establishing conditional independence.
\begin{lemma} \label{lem:imperfect-expectation}
Suppose $\E_{\bm{U}} f(\bm{U}, i, k, k', \sigma_{\mathsf{C}}) \geq \epsilon$ for all $i \in [d_0]$ and $ k\neq k'\in [d_0]$.
Then, we have
\begin{align}
    \E_{\bm{U},\bm{i}}\norm{\bm{\rho}_{\mathsf{A}}^{\bm{i}}}_\infty \geq \norm{\rho_{\mathsf{A}}}_\infty \cdot (1- d_0\cdot \epsilon)+d_0\cdot \epsilon
\end{align}
\end{lemma}
\begin{proof}
Following the proof of \thmref{entswap-claim} and using the notation in \eqref{eq:postMeasurementState} and \eqref{eq:PostMeasurementPsi_1}, we have that for all $k\in [d_0]$
\begin{align}
   \E_{\bm{U}, \bm{i}} \norm{\bm{\rho}_{\mathsf{A}}^{\bm{i}}}_{\infty} &\geq \lambda_k + \sum_{k'\neq k}\E_{\bm{i}} \E_{\bm{U}}\bigg[  \cdot \frac{q_{\bm{i}, k'}\lambda_{k'}}{p_{\bm{i}}} \cdot |\braket{\phi_{\bm{i}, k}}{\phi_{\bm{i}, k'}}|^2\bigg] \nonumber \\ 
   &=\lambda_k + \sum_{k'\neq k} \lambda_{k'}\cdot \sum_i \E_{\bm{U}} \bigg[\frac{f(\bm{U}, i, k, k',\sigma_{\mathsf{C}})}{q_{i, k}}\bigg] \geq \lambda_k + \sum_{k'\neq k} \lambda_{k'} \cdot d_0\cdot \epsilon\nonumber \\
    &= \lambda_k + (1-\lambda_k)\cdot d_0\cdot \epsilon.
\end{align}
To arrive at the last inequality, we used the fact that $q_{i, k}\leq 1$ and by assumption $\E_{\bm{U}} f(\bm{U}, i, k, k', \sigma_{\mathsf{C}}) \geq \epsilon$ for all $i \in [d_0]$ and $ k\neq k'\in [d_0]$.
\end{proof}

By sequentially repeating the result of \lemref{imperfect-expectation}, we can establish a decay of conditional correlations along a chain of qudits that satisfy $\E_{\bm{U}} f(\bm{U}, i, k, k', \sigma_{\mathsf{C}}) \geq \epsilon$.
To this end, consider the state $\ot_{k=1}^{n}\ket{a_k}_{\mathsf{L}_k\mathsf{R}_k}$ where each subsystem has a local dimension $d_0$.
Suppose each pair of registers $(\mathsf{R}_j, \mathsf{L}_{j+1})$ for $j\in [n-1]$ are first rotated by a random unitary $\bm{U_j}$ and then measured in the computational basis.
We denote the outcomes by $(\bm{i}_1,\dots, \bm{i}_{n-1})$ for $\bm{i}_1,\dots, \bm{i}_{n-1} \in [d_0]^{\times 2}$.
Given a random set of unitaries $\bm{U_1},\dots,\bm{U_{n-1}}$, we let $\bm{\rho}_{1}^{(\bm{i}_1,\dots, \bm{i}_{j})}$ be the post-measurement state on subsystem~$\mathsf{L}_1$ upon the first $j$ measurements.
Our goal is to upper bound $\E_{\bm{U}_1,\dots, \bm{U}_{n-1}} \E_{\bm{i}_1,\dots, \bm{i}_{n-1}} S(\bm{\rho}_{1}^{(\bm{i}_1,\dots, \bm{i}_{n-1})})$.

\begin{proposition}[Exponential decay of average post-measurement entanglement]\label{prop:claim-chain}
     Suppose each of the entanglement swapping protocols on registers $(\mathsf{B_j}, \mathsf{A}_{j+1})$ for $j \in [n-1]$ satisfies $$\E_{\bm{U}_j} f(\bm{U}, i, k, k', \sigma_{\mathsf{A_{j+1}}}) \geq \epsilon$$ for all $i \in [d_0]$, $ k\neq k'\in [d_0]$, and reduced states $\sigma_{\mathsf{A}_{j+1}}$ for some constant $\epsilon > 0$.
     Then, the post measurement reduced state $\bm{\rho}^{(\bm{i_1},\dots,\bm{i_{n-1}})}_{1}$ on $\mathsf{L}_1$ satisfies
    \begin{align}
        \E_{\bm{U_1},\dots, \bm{U_{n-1}}}\E_{\bm{i}_1,\dots, \bm{i}_{n-1}} S(\bm{\rho}^{(\bm{i}_1,\dots, \bm{i}_{n-1})}_{1}) \leq e^{-\epsilon \cdot (n-1)} (1 + \epsilon\cdot (n-1) + \log d_0).
    \end{align}
\end{proposition}

\begin{proof}
    For a given set of unitaries $\bm{U_1},\dots, \bm{U_{n-1}}$, suppose we have obtained $(\bm{i}_1,\dots, \bm{i}_j)$ upon measuring the first $j$ pairs of subsystems. 
    By \corref{DecayByConvexity}, for all $j \in [n-2]$, the operator norm $\norm{\bm{\rho}_{1}^{(\bm{i_1},\dots, \bm{i_j}, \bm{i_{j+1}})}}_\infty$~satisfies
    \begin{align}
    \E_{\substack{\bm{U_1},\dots, \bm{U_j}\\\bm{i_1},\dots,\bm{i_j}}}\E_{\bm{U_{j+1}}, \bm{i_{j+1}}} \norm{\bm{\rho}_1^{(\bm{i_1},\dots, \bm{i_j}, \bm{i_{j+1}})}}_\infty \geq \E_{\substack{\bm{U_1},\dots, \bm{U_j}\\\bm{i_1},\dots,\bm{i_j}}}\norm{\bm{\rho}_1^{(\bm{i_1},\dots, \bm{i_j})}}_\infty (1-\epsilon) + \epsilon 
    \end{align}
and thereby by induction and using the simplified notation $\bm{I}=(\bm{i}_1,\dots, \bm{i}_{n-1})$ and $\bm{U}= (\bm{U_1},\dots, \bm{U}_{n-1})$, we have
    \begin{align}
        \E_{\bm{U}, \bm{I}} \norm{\bm{\rho}_1^{\bm{I}}}_\infty  &\geq \epsilon\cdot \left((1-\epsilon)^{n-2}+(1-\epsilon)^{n-3}+\dots+1\right)+ \|\rho_{1}\|_\infty\cdot (1-\epsilon)^{n-1}\nonumber\\
         &\geq 1-(1-\epsilon)^{n-1} \geq 1 - e^{-\epsilon (n-1)}.\label{eq:BoundOnLambda}
    \end{align}
    Let $\bm{\lambda}_{\bm{I}} = \norm{\bm{\rho}_1^{\bm{I}}}_\infty  $ be the largest eigenvalue of $\bm{\rho}_1^{\bm{I}}$ and define $\lambda:=\E_{\bm{U}, \bm{I}}[\bm{\lambda}_{\bm{I}}]$.
    We can upper bound the entropy of $\bm{\rho}_1^{\bm{I}}$ by distributing the remaining $1-\bm{\lambda}^{\bm{I}}$ mass evenly among the other $d_0-1$ eigenvalues. 
    We conclude, 
        \begin{align}
          \E_{\bm{U}, \bm{I}} S(\bm{\rho}_1^{\bm{I}}) &\leq  \E_{\bm{U},\bm{I}} S\big(\lambda_{\bm{I}}, \underbrace{\frac{1-\lambda_{\bm{I}}}{d_0-1},\dots, \frac{1-\lambda_{\bm{I}}}{d_0-1}}_{\times (d_0-1)}\big) &&\text{Schur concavity}\nonumber\\
          &\leq S\big(\lambda, \underbrace{\frac{1-\lambda}{d_0-1},\dots, \frac{1-\lambda}{d_0-1}}_{\times (d_0-1)}\big) && \text{concavity}\nonumber\\
          &\leq (1-\lambda)\log \frac{d_0-1}{1-\lambda} + \lambda\log \frac{1}{\lambda} \nonumber\\
          &\leq  e^{-\epsilon (n-1)} (1 + \epsilon (n-1) + \log d_0). && \text{using \eqref{eq:BoundOnLambda}}.
    \end{align}
\end{proof}

In the next lemma, we prove $\min_{i,k\neq k'}\E_{\bm{U}} f(\bm{U}, i, k, k', \sigma_{\mathsf{C}}) \geq \epsilon$ for two cases where the unitaries $\bm{U}_{\mathsf{BC}}$ are drawn from the Haar measure, or when these unitaries are constructed from random shallow quantum circuits.

\begin{lemma}\label{lem:UBCBounds}
    Fix $k\neq k'$, and consider the context of \lemref{imperfect-expectation}.
    \begin{enumerate}
        \item \textbf{Haar random:} if $\bm{U}_{\mathsf{BC}}$ is drawn from the Haar measure over $\mathsf{BC}$ with local dimensions $\dim{\mathcal{H}_{\mathsf{B}}} = \dim{\mathcal{H}_{\mathsf{C}}} = d_0$, then for all $i\in[q]$ and $\sigma_{\mathsf{C}}$, we have
     \begin{align}
        \E_{\bm{U}}f(\bm{U}, i, k, k', \sigma_{\mathsf{C}}) \geq \frac{1}{2d_0^4}.\label{eq:BoundEFHaar}
      \end{align}
        \item \textbf{Shallow circuit:} Suppose $\bm{U}_{\mathsf{BC}}$ is a triangular quantum circuit constructed in \secref{ReductionShallow} and \fig{lightconesShallow} consisting of $D$ levels of two-local gates acting on qudits of dimension $d$. 
        Then, for all $i\in[q]$ and $\sigma_{\mathsf{C}}$, we have
        \begin{align}
        \E_{\bm{U}}f(\bm{U}, i, k, k',\sigma_{\mathsf{C}})  \geq \left(\frac{1}{(d^2+1)(d+1)}\right)^{D-1}\gtrsim \frac{1}{d^{3D}}.\label{eq:BoundEFShallow}
        \end{align}   
    \end{enumerate}
\end{lemma}

The two bounds \eqref{eq:BoundEFHaar} and \eqref{eq:BoundEFShallow} can be compared by replacing $d_0 = d^{D-1}$. 
A corollary of these bounds is the following exponential decay of conditional mutual information in the output state of a random 1D shallow circuit. 

\begin{corollary}[Conditional independence for random 1D shallow circuits, restatement of \thmref{CMIdecayRandom1D}]
       Consider the family of random depth-$D$ brickwork quantum circuits acting on a one-dimensional chain of qudits with the local dimension $d$. 
       Fix $L \geq c_0 \cdot (4d)^{3D} \cdot \log(n \cdot D\log(d))$ for some sufficiently large constant $c_0$ and define contiguous regions $\mathsf{A}_1\cup \dots\cup\mathsf{A}_{n'}$ from left to right with $n' = n/L$ such that the size of each region is $|\mathsf{A}_i| = L \geq \Omega(\log(n))$ and $\dist(\mathsf{A}_i-\mathsf{A}_j)\geq \Omega(\log(n))$ for $|i-j|>1$.
    With probability $1 - 1/\poly(n)$ over the random choice of circuits, the measurement distribution of these random circuits satisfies the conditional~independence~property
     \begin{align}
    I(\mathsf{A}_i:\mathsf{A}_j|\mathsf{A}_{i+1},\dots, \mathsf{A}_{j-1}) \leq e^{-\Omega{\left(\dist(\mathsf{A}_i, \mathsf{A}_j)\right)}}\leq 1/\poly(n).
     \end{align}
for any $1\leq i<j\leq n'$ with $|i-j|>1$.
\end{corollary}
\begin{proof}
We consider the setup of \propref{claim-chain}. 
Given the stated choice of $L$, the distance $\dist(\mathsf{A_i}, \mathsf{A_j}) \geq \Omega(\log(n))$ for any $1\leq i<j\leq n'$.
We decompose the circuit into backward $\triangle$ and forward lightcones $\bigtriangledown$ as in \fig{lightconesShallow}, and think of each subsystem $\mathsf{A_i}$ as a region consisting of an integer number of backward lightcones, along with two half-backward lightcones at its right and left boundaries.

Suppose we proceed to measure subsystems $\mathsf{A}_{i+1},\dots, \mathsf{A}_{j-1}$. 
We know from \lemref{UBCBounds} that shallow 1D quantum circuits satisfy $\E_{\bm{U}} f(\bm{U}, i, k, k', \sigma) \geq \epsilon$ with $\epsilon \gtrsim (4d)^{-3D}$ for any choice of reduced state $\sigma$. 
Following a lightcone argument, we see that, the measurement probabilities and the resulting entanglement between the unmeasured subsystems does not change by applying the gates in any of the forward lightcones $\bigtriangledown$ supported in regions $\mathsf{A}_i$ and $\mathsf{A}_j$.
This along with the data processing inequality applied to subsystems $\mathsf{A}_1,\dots, \mathsf{A}_{i-1}$ and $\mathsf{A}_{j+1},\dots, \mathsf{A}_{n'}$, shows that it suffices to apply our entanglement swapping analysis only to a part of the circuit consisting of the middle subsystems $\mathsf{A}_{i+1}, \dots, \mathsf{A_j}-1$ along with the qubits in the right boundary of $\mathsf{A}_i$ and the left boundary of $\mathsf{A}_{j}$, independent of the measurement outcome in the rest of the circuit. 
Applying \eqref{eq:BoundOnLambda} and \eqref{eq:CMIfromES} implies that
$\E_{\bm{U}} I(\mathsf{A}_i:\mathsf{\mathsf{A}_j}|\mathsf{A}_{i+1},\dots, \mathsf{A}_{j-1}) \leq \frac{1}{n^{c_0-1}}$.
An application of the Markov inequality and a union bound over the choice of regions $\mathsf{A}_i$ and $\mathsf{A}_j$ concludes the proof.

\end{proof}

In the following section, we state the proof of \lemref{UBCBounds}.

\subsubsection{Haar Random Unitaries}\label{sec:HaarRandomUBC}

We use the following property of Haar random unitaries:

\begin{fact}[cf. \cite{anticoncentration}] \label{fact:haarfact}
    Let $V$ be a unitary acting on a $q$-dimensional space $\mathcal{H}$, and $O$ a Hermitian operator on the $q^2$-dimensional space $\mathcal{H} \otimes \mathcal{H}$.
    We have
        \begin{equation}
        \E_{\bm{V}}[\bm{V}^{\dagger} \otimes \bm{V}^{\dagger}\cdot  O \cdot \bm{V} \otimes \bm{V}] = \frac{\Tr(O)-q^{-1}\Tr(O \operatorname{SWAP})}{q^2-1} \cdot \iden + \frac{\Tr(O\operatorname{SWAP}) - q^{-1}\Tr(O)}{q^2-1} \cdot \operatorname{SWAP},\nonumber
    \end{equation}
    where the expectation over $\bm{V}$ is with respect to the Haar measure, and $\operatorname{SWAP}$ is the Swap operator such that $\operatorname{SWAP}\cdot \ket{i}_{\mathsf{B}} \ket{j}_{\mathsf{C}} = \ket{j}_{\mathsf{B}} \ket{i}_{\mathsf{C}}$ for any $i,j \in [q]$.
    This, in particular, implies the following equalities for bipartite spaces $\mathcal{H} = \mathsf{BC}$ and $\mathcal{H}=\mathsf{B'C'}$:
\begin{enumerate}[label=\roman*.]
\item $\E_{\bV} (\bV^{\dagger{\ot 2}}\cdot \iden \ot \iden \cdot \bV^{\ot 2})=\iden \ot \iden$,
\item $\E_{\bV} (\bV^{\dagger{\ot 2}}\cdot\mathrm{SWAP} \ot \mathrm{SWAP}\cdot \bV^{\ot 2})=\mathrm{SWAP} \ot \mathrm{SWAP}$,
\item $\E_{\bV} (\bV^{\dagger{\ot 2}}\cdot\mathrm{SWAP} \ot \iden \cdot\bV^{\ot 2})=\frac{q}{q^2+1}\cdot(\iden \ot \iden +\mathrm{SWAP}\ot \mathrm{SWAP})$,
\item $\E_{\bV} (\bV^{\dagger{\ot 2}}\cdot\iden \ot \mathrm{SWAP} \cdot \bV^{\ot 2})=\frac{q}{q^2+1}\cdot(\iden \ot \iden +\mathrm{SWAP}\ot \mathrm{SWAP})$,
\item $\E_{\bV} (\bV^{\dagger} \cdot \ketbra{\phi}{\phi} \cdot \bV)^{\ot 2}=\frac{1}{q(q+1)}\cdot(\iden+\mathrm{SWAP})$~~ for any state $\ket{\phi}$.
\end{enumerate}
\end{fact}

\begin{lemma}
        Fix $k\neq k'$. 
        In the context of \lemref{imperfect-expectation}, if $\bm{U}_{\mathsf{BC}}$ is drawn from the Haar measure over $\mathsf{BC}$ with local dimensions $d_0$, then for all $i\in[d_0]$, we have

    \begin{equation}
        \E_{\bm{U}}f(\bm{U}, i, k, k', \sigma_{\mathsf{C}}) = \frac{d_0-1}{d_0^4-1}\cdot \Tr[\sigma_C^2] \geq \frac{1}{2d_0^4}.\nonumber
    \end{equation}
    
\end{lemma}

\begin{proof} From \factref{haarfact} and by letting $\bm{V} = \bm{U}_{\mathsf{BC}}$ and $\mathcal{H}$ to be the space of subsystems $\mathsf{BC}$ with $q=d_0^2$, we get 
    \begin{align}
        &\E_{\bm{U}}\bigg[(\bm{U}^\dagger_{\mathsf{BC}}\otimes \bm{U}^\dagger_{\mathsf{B'C'}})(\ketbra{i}{i}_{\mathsf{B}}\otimes \ketbra{i}{i}_{\mathsf{B}'}\otimes \iden_{\mathsf{CC'}})(\bm{U}_{\mathsf{BC}}\otimes \bm{U}_{\mathsf{B'C'}})\bigg]\nonumber\\
        &= \frac{d_0^2-d_0^{-1}}{d_0^4-1}\cdot \iden_{\mathsf{BB'CC'}} +\frac{d_0-1}{d_0^4-1}\cdot  \operatorname{SWAP}_{\mathsf{BC}, \mathsf{B'C'}}.
    \end{align}
    Hence, assuming $k\neq k'$ 
    \begin{align}
        \E_{\bm{U}}[f(\bm{U}, i, k, k', \sigma_{\mathsf{C}})] &= \frac{d_0-1}{d_0^4-1}\Tr\bigg[ \operatorname{SWAP}_{\mathsf{BC}, \mathsf{B'C'}} \big(\ketbra{kk'}{k'k}_{\mathsf{BB'}} \otimes \sigma_{\mathsf{C}} \otimes \sigma_{\mathsf{C'}} \big)\bigg]\nonumber \\
        &= \frac{d_0-1}{d_0^4-1}\cdot \Tr(\sigma_{\mathsf{C}}^2) \geq \frac{d_0-1}{d_0^4-1}\cdot \frac{1}{d_0} \geq \frac{1}{2d_0^4}
    \end{align}
\end{proof}

\subsubsection{Random shallow circuits}\label{sec:ShallowUBC}

Now we consider unitaries $\bm{U}$ which originate from a 1D, low depth locally-random quantum circuits, as introduced in \secref{ReductionShallow}.

\begin{lemma}
    Fixed $k\neq k', i\in [d^{D-1}]$, if $\bm{U}$ is drawn from the distribution described above,
    \begin{equation}
              \E_{\bm{U}}f(\bm{U}, i, k, k', \sigma_{\mathsf{C}}) = \left(\frac{1}{(d^2+1)(d+1)}\right)^{D-1}\cdot \Tr[\sigma_{\mathsf{C}}^2] \gtrsim \frac{1}{d^{3D}}
    \end{equation}
    
\end{lemma}

\begin{proof}
Let us briefly high-level the proof of this claim. The key insight is two-fold: we picture the random quantum circuit in reverse, acting on the computational basis state $|b\rangle$ corresponding to the measurement outcome on $B$. We view the action of the locally random quantum circuit as a linear combination of Identity and Swap operators, which evolve and mix in time; in particular, the dynamics of the (reverse time evolved) circuit can be understood through the domain walls between $\iden$ and $S$ regions on the 1D chain. To show the lower bound in the claim, we show that each domain wall adds a positive contribution to the expectation, and we pick a single domain wall trajectory to lower bound the total sum. 

\begin{align}
    &\E_{\bm{U}}\bigg[(\bm{U}^\dagger_{\mathsf{BC}}\otimes \bm{U}^\dagger_{\mathsf{B'C'}})(\ketbra{i}{i}_{\mathsf{B}}\otimes \ketbra{i}{i}_{\mathsf{B}'}\otimes \iden_{\mathsf{CC'}})(\bm{U}_{\mathsf{BC}}\otimes \bm{U}_{\mathsf{B'C'}})\bigg]\nonumber\\
    &= \E_{\bm{U, \bm{W}}}\bigg[(\bm{U}^\dagger_{\mathsf{BC}}\otimes \bm{U}^\dagger_{\mathsf{B'C'}})(\bm{W}_{\mathsf{B}}^{\dagger} \otimes \bm{W}^{\dagger}_{\mathsf{B}'}\cdot \ketbra{i}{i}_{\mathsf{B}}\otimes \ketbra{i}{i}_{\mathsf{B}'}\cdot \bm{W}_{\mathsf{B}} \otimes \bm{W}_{\mathsf{B}'}\otimes \iden_{\mathsf{CC'}})(\bm{U}_{\mathsf{BC}}\otimes \bm{U}_{\mathsf{B'C'}})\bigg]\nonumber\\
    &=  \frac{1}{d^{D-1}(d+1)^{D-1}}\cdot \sum_{\gamma_1,\dots,\gamma_{D-1} \in \{\iden ,\operatorname{SWAP}\}}\E_{\bm{U}}\bigg[(U^\dagger_{\mathsf{BC}}\otimes U^\dagger_{\mathsf{B'C'}})(\otimes_{i\in [D-1]} \gamma_i\otimes \iden_{\mathsf{CC'}})(U_{\mathsf{BC}}\otimes U_{\mathsf{B'C'}})\bigg]
\end{align}
where we used the invariance property of Haar distribution of $\bm{U}$ under left multiplication by single qudit rotations, and applied the equality \textit{v.} of \factref{haarfact}. 
Each operator $\gamma_i$ acts on a qudit in $\mathsf{B}$ and its corresponding copy in $\mathsf{B'}$.

When the expectation over the two-local gates in $\bm{U}$ is taken, we get a mixture of the operators $\iden$ and $\operatorname{SWAP}$ following \factref{haarfact}.
The overall result is a linear combination over identity and swap operators given by $$\sum_{\gamma \in \{\iden, \operatorname{SWAP}\}^{\times 2(D-1)}} c_{\gamma} \cdot \gamma$$
where as seen in properties \textit{i.}-\textit{iv.} of \factref{haarfact}, the coefficients $c_{\gamma}$ are non-negative.
Here, $\gamma = \otimes_{i\in[2(D-1)]} \gamma_i$ where each $\gamma_i$ is supported on a qudit and its copy in $\mathsf{BB'}$ or $\mathsf{CC'}$. 
A simple property of the Swap operator ensures that these coefficients $c_\sigma$ give us a lower bound to the quantity of interest:

\begin{fact}\label{fact:SWAPproperty}
    The swap operator has a positive expectation over product states. That is, for any two density matrices $\rho_{\mathsf{X}}, \rho_{\mathsf{Y}}$, $\Tr(\operatorname{SWAP}_{\mathsf{XY}}\cdot \rho_{\mathsf{X}}\otimes \rho_{\mathsf{Y}}) \geq 0$.
\end{fact}

\noindent By applying this bound we get
\begin{align}
     \E_{\bm{U}}f(\bm{U},  i,k,k',\sigma_{\mathsf{C}}) &=  \sum_{\gamma \in \{\iden, \operatorname{SWAP}\}^{\times 2(D-1)}} c_\gamma \cdot \Tr\left(\gamma\cdot \ketbra{k}{k'}_{\mathsf{B}} \otimes  \ketbra{k'}{k}_{\mathsf{B'}} \otimes \sigma_{\mathsf{C}}\otimes \sigma_{\mathsf{C'}}\right)\nonumber\\
    &=\sum_{\gamma \in \{\iden, \operatorname{SWAP}\}^{\times 2(D-1)}} c_\gamma \cdot \Tr\left(\gamma\cdot \operatorname{SWAP}_{\mathsf{BB'}} \otimes \iden_{\mathsf{CC'}} \cdot \ketbra{k'}{k'}_{\mathsf{B}} \otimes  \ketbra{k}{k}_{\mathsf{B'}} \otimes \sigma_{\mathsf{C}}\otimes \sigma_{\mathsf{C'}}\right)\nonumber\\
    &\geq c_{\operatorname{SWAP}_{\mathsf{BB'}}, \iden_{\mathsf{CC'}}}.
\end{align}
The last inequality holds since $\ketbra{k'}{k'}_{\mathsf{B}} \otimes  \ketbra{k}{k}_{\mathsf{B'}} \otimes \sigma_{\mathsf{C}}\otimes \sigma_{\mathsf{C'}}$ is a density matrix over $\mathsf{B}, \mathsf{B'}, \mathsf{C}, \mathsf{C'}$ and a product state over $\mathsf{BC}, \mathsf{B'C'}$.
Hence, from \factref{SWAPproperty}, we see that the summation in the second line includes non-negative terms.
By choosing the term with $\gamma = \operatorname{SWAP}_{\mathsf{BB'}}\otimes \iden_{\mathsf{CC'}}$, we obtain the stated lower bound in terms of $c_{\operatorname{SWAP}_{\mathsf{BB'}}, \iden_{\mathsf{CC'}}}$.

To conclude, we find a lower bound for the coefficient $c_{\operatorname{SWAP}_{\mathsf{BB'}}, \iden_{\mathsf{CC'}}}$.
In the first place, as stated above, we apply a layer of single-qudit Haar random gates $\bm{W}_{\mathsf{B}}$ on subsystem $\mathsf{B}$ (and equivalently $\mathsf{B'}$).
This results in the coefficient corresponding to the operator $\operatorname{SWAP}_{\mathsf{BB'}}\otimes \iden_{\mathsf{CC'}}$ possessing an initial weight $\frac{1}{d^{D-1}(d+1)^{D-1}}$. 
Let qudit $i_0$ (resp. $i_0+1$) be the boundary qudit of region $\mathsf{B}$ (resp. $\mathsf{C}$). 
The link $(i_0,i_0+1)$ marks the transition from $\operatorname{SWAP}$ to $\iden$ in $\operatorname{SWAP}_{\mathsf{BB'}}\otimes \iden_{\mathsf{CC'}}$, and is often called the `domain wall'. 
We next see how the transformation of $\operatorname{SWAP}_{\mathsf{BB'}}\otimes \iden_{\mathsf{CC'}}$ due the application of the first level of two-qudit gates in $\bm{U}_{\mathsf{BC}}$ can be described by the movement of this domain wall. 
To this end, we divide this layer into three sets: $\bm{U}^{\dagger}_2$ denotes the gate acting on the boundary of $\mathsf{BC}$, and $\bm{U}^{\dagger}_1$ and $\bm{U}^{\dagger}_3$ denote the set of gates acting only on $\mathsf{B}$ and $\mathsf{C}$, respectively. 
We see from properties \textit{i.} and \textit{ii.} of \factref{haarfact} that 
$$\E_{\bm{U}_1} \left(\bm{U}^{\dagger \otimes 2}_1\cdot \operatorname{SWAP}_{\mathsf{BB'}} \otimes \iden_{\mathsf{CC'}}\cdot \bm{U}^{\otimes 2}_1 \right) = \E_{\bm{U}_3} \left(\bm{U}^{\dagger \otimes 2}_3\cdot \operatorname{SWAP}_{\mathsf{BB'}} \otimes \iden_{\mathsf{CC'}}\cdot \bm{U}^{\otimes 2}_3\right) = \operatorname{SWAP}_{\mathsf{BB'}} \otimes \iden_{\mathsf{CC'}}.$$
Moving to the middle gates $\bm{U}_2$, the property \textit{iii.} (or \textit{iv.}) of \factref{haarfact} yields $$\E_{\bm{U}_2}\left( \bm{U}^{\dagger \otimes 2}_2\cdot \operatorname{SWAP}_{i_0,i_0'} \otimes \iden_{i_0+1,i'_0+1}\cdot \bm{U}^{\otimes 2}_2\right) = \frac{d}{d^2+1}\cdot (\iden_{i_0,i'_0} \otimes\iden_{i_0+1,i'_0+1} + \operatorname{SWAP}_{i_0,i'_0}\otimes \operatorname{SWAP}_{i_0+1,i'_0+1}).$$ 
Combining these transformations, we see that the application of the first layer of gates on $\operatorname{SWAP}_{\mathsf{BB'}} \otimes \iden_{\mathsf{CC'}}$, results in a linear combination of two operators such that the domain wall has moved to the left in one or to the right in the other. 
In expectation over the remaining layers of $\bm{U}$, this domain wall propagates in time, resulting in a linear combination over operators that can be understood as a sum over domain wall trajectories. 

Suppose the depth $D$ is odd so that the number of layers $D-1$ in $\bm{U}$ is even.
Let us consider the domain wall trajectory which starts and ends in the same place, resulting in the same operator $\operatorname{SWAP}_{\mathsf{BB'}}\otimes \iden_{\mathsf{CC'}}$. 
One such path is the trivial alternating path, where the domain wall steps between $(i_0,i_0+1)$ and $(i_0+1,i_0+2)$ back and forth. 
By \factref{haarfact}, at each step the coefficient of this path decays by a factor of $\frac{d}{d^2+1}$, resulting in a contribution 
\begin{align}
    c_{\operatorname{SWAP}_{\mathsf{BB'}}, \iden_{\mathsf{CC'}}} &\geq \left(\frac{1}{d(d+1)}\right)^{D-1}\cdot \left(\frac{d}{d^2+1}\right)^{D-1} \gtrsim \frac{1}{d^{3D}}.\nonumber
    \end{align}
\end{proof}

\section{Details of numerical simulations}\label{sec:detailsNumerics}
For the RNN quantum states, both with and without phase information, we employ the tensor-gated recurrent unit (tensor-GRU) architecture \cite{Hibat2020RNN} with a batch size of $256$.
We utilize the Adam optimizer \cite{Adam}, setting the initial learning rate to $2 \times 10^{-3}$. 
The learning rate undergoes stepwise decays to $1 \times 10^{-3}$ after $8,000$ iterations and to $4 \times 10^{-4}$ after $24,000$ iterations. 
Each angle of the rotated cluster state is trained over a total of $10^5$ iterations.

The RBM ansatz is implemented using the NetKet package \cite{netket2:2019, netket3:2022} with a batch size of $2^{12}$. 
Parameters are optimized through stochastic gradient descent (SGD) combined with stochastic reconfiguration (SR). 
The learning rate follows a warmup cosine decay schedule \cite{deepmind2020jax}, starting at $2 \times 10^{-4}$, 
increasing linearly to $2 \times 10^{-4}$ by iteration $500$, and then decaying via a cosine curve to $2.5 \times 10^{-4}$ 
by iteration $2,500$. 

Finally, we investigate the conditional mutual information of random Matrix Product States (MPS) and Projected Entangled Pair States (PEPS) using Quimb \cite{gray2018quimb}. Density Matrix Renormalization Group (DMRG) calculations are performed with ITensor Julia \cite{itensor}.

In addition to the plots presented in this paper, which focus on tensor-GRU and RBM architectures discussed in the main text, our code includes three additional experiments outlined below:

(1) We explore the RWKV ansatz \cite{peng2023rwkv} and the Transformer Quantum State (TQS) ansatz \cite{sprague2024variational}. 
In our preliminary experiments, both architectures are trained with a batch size of $256$, using the Adam optimizer with a learning rate of $1.5 \times 10^{-4}$ with linear learning rate warmup from $10^{-6}$ for the first $10^3$ steps. 
Both ansatz exhibit similar behaviors to tensor-GRU.

(2) Apart from the 1D rotated cluster state model with single-qubit $R_y(\theta)$ rotations, we also observe similar connection between CMI correlation length and the performance of RNN quantum states for rotated graph state on a two-dimensional square lattice with:
$H_{\text{graph}} = -\sum_{1\leq i\leq N, 1\leq j\leq N}Z_{i,j}\prod_{\langle (i,j), (k, l)\rangle\in E} X_{k,l}$
where the subscript represents the location of the operator in the 2D-square lattice and $E$ is the set of all the edges of the 2D-square lattice.

(3) Inspired by the patched approach from \cite{sprague2024variational}, we modified the autoregressive output: instead of generating a two-fold probability distribution per qubit, the model produces a $2^p$-fold probability distribution for $p$ qubits at each step.
This modification reduces the effective depth of the autoregressive architecture at the cost of increasing the input and output dimensions at each step.

\end{document}